\documentclass{article}
\usepackage[vmargin=1in,hmargin=1in]{geometry}
\usepackage{amsmath,amssymb,amsfonts,graphicx,amsthm,nicefrac,mathtools,bm}
\usepackage{hyperref}
\usepackage[numbers]{natbib}
\usepackage[english]{babel}
\usepackage{times}
\usepackage[T1]{fontenc}
\usepackage{bbm,mdframed}
\usepackage[dvipsnames]{xcolor}
\usepackage{xspace}
\usepackage{rotating,multirow}
\usepackage{graphics,tikz}
\graphicspath{{./fdrplots/}}
\usepackage{grffile}
\usepackage{subcaption}
\usepackage{epstopdf} % for postscript graphics files 
\usepackage[ruled]{algorithm2e}
\usepackage{float}
\usepackage[toc,title,page]{appendix}

\newtheorem{theorem}{Theorem}

\newtheorem{lemma}{Lemma}
\newtheorem{proposition}{Proposition}

\theoremstyle{definition}
\newtheorem{definition}{Definition}
\theoremstyle{remark}
\newtheorem{remark}{Remark}

\makeatletter

\newtheorem*{rep@theorem}{\rep@title}
\newcommand{\newreptheorem}[2]
{\newenvironment{rep#1}[1]
	{\def\rep@title{#2 \ref{##1}} \begin{rep@theorem}}%
		{\end{rep@theorem}}}
\makeatother
\newreptheorem{theorem}{Theorem}
\newreptheorem{lemma}{Lemma}
\newreptheorem{corollary}{Corollary}
\newreptheorem{proposition}{Proposition}

\newcommand{\figref}[1]{Figure~\ref{fig:#1}}

\newcommand{\secref}[1]{Section~\ref{sec:#1}}
\newcommand{\appref}[1]{Appendix~\ref{app:#1}}

\newcommand{\defref}[1]{Definition~\ref{def:#1}}

\newcommand{\lemref}[1]{Lemma~\ref{lem:#1}}

\newcommand{\propref}[1]{Proposition~\ref{prop:#1}}

\newcommand{\thmref}[1]{Theorem~\ref{thm:#1}}

\newcommand{\tabref}[1]{Table~\ref{tab:#1}}

\newcommand{\eqnref}[1]{\eqnref{eq:#1}}

\newcommand{\norm}[1]{\left\lVert{#1}\right\rVert}
\newcommand{\PP}[1]{\textnormal{Pr}\!\left\{{#1}\right\}} % Probability
\newcommand{\EE}[1]{\mathbb{E}\left[{#1}\right]} % Expectation

\newcommand{\EEst}[2]{\mathbb{E}\left[{#1}\ \middle| \ {#2}\right]} % Conditional expectation
\def\H0{\mathcal{H}^0}

\newcommand{\ident}{\mathbf{1}}

\newcommand{\ignore}[1]{}
\newcommand{\ones}{\mathbf{1}}
\newcommand{\gap}{\Delta^{\textnormal{gap}}}

\usepackage{fancyhdr}
\newcommand{\thedate}{\today}
\newcommand{\theauthor}{}
\newcommand{\thetitle}{From local to global gene co-expression estimation using single-cell RNA-seq data}

\date{\thedate}
\author{\theauthor}
\title{\thetitle}

\fancypagestyle{plain}{%
	\fancyhf{}%
	\lhead{\thepage}
}

\usepackage[mmddyy]{datetime}

\def\F{\mathcal{F}}

\def\H{\mathcal{H}}

\makeatletter
\newcommand{\dotfrac}[2]{
	\mathchoice
	{\ooalign{$\genfrac{}{}{0pt}{0}{#1}{#2}$\cr\leavevmode\cleaders\hb@xt@ .22em{\hss $\displaystyle\cdot$\hss}\hfill\kern\z@\cr}}
	{\ooalign{$\genfrac{}{}{0pt}{1}{#1}{#2}$\cr\leavevmode\cleaders\hb@xt@ .22em{\hss $\textstyle\cdot$\hss}\hfill\kern\z@\cr}}
	{\ooalign{$\genfrac{}{}{0pt}{2}{#1}{#2}$\cr\leavevmode\cleaders\hb@xt@ .22em{\hss $\scriptstyle\cdot$\hss}\hfill\kern\z@\cr}}
	{\ooalign{$\genfrac{}{}{0pt}{3}{#1}{#2}$\cr\leavevmode\cleaders\hb@xt@ .22em{\hss $\scriptscriptstyle\cdot$\hss}\hfill\kern\z@\cr}}
}
\makeatother

\tikzstyle{none} = [rectangle, rounded corners, minimum width=0.5cm, minimum height=0.5cm,text centered, draw=black, fill=blue!20]
\tikzstyle{noneoff} = [rectangle, rounded corners, minimum width=0.5cm, minimum height=0.5cm,text centered, draw=black, fill=blue!5]

\tikzstyle{ada} = [rectangle, rounded corners, minimum width=0.5cm, minimum height=0.5cm, text centered, draw=black, fill=orange!20]
\tikzstyle{adaoff} = [rectangle, rounded corners, minimum width=0.5cm, minimum height=0.5cm, text centered, draw=black, fill=orange!5]

\tikzstyle{adadis} = [rectangle, rounded corners, minimum width=0.5cm, minimum height=0.5cm, text centered, draw=black, fill={rgb:red,10;green,20;yellow,54}]
\tikzstyle{adadisoff} = [rectangle, rounded corners, minimum width=0.5cm, minimum height=0.5cm, text centered, draw=black, fill=green!5]

\tikzstyle{arrow} = [thick,->,>=stealth, text width=3cm]
\tikzstyle{dotarrow} = [dashed, text width=3cm]

\tikzstyle{title} = [rectangle, rounded corners, minimum width=1cm, minimum height=1cm, text centered, draw=white, fill=none]

\begin{document}
	\author{
		Jinjin Tian, Jing Lei, Kathryn Roeder\\
		Department of Statistics and Data Science\\
		Carnegie Mellon University\\
		\texttt{\{jinjint,jinglei,roeder\}@andrew.cmu.edu}
	}
	\maketitle

\begin{abstract}
In genomics studies, the investigation of the gene relationship often brings important biological insights. Currently, the large heterogeneous datasets impose new challenges for statisticians because gene relationships are often local. They change from one sample point to another, may only exist in a subset of the sample, and can be non-linear or even non-monotone. Most previous dependence measures do not specifically target local dependence relationships, and the ones that do are computationally costly. In this paper, we explore a state-of-the-art network estimation technique that characterizes gene relationship at the single cell level, under the name of  \emph{cell-specific gene networks}. We first show that averaging the \emph{cell-specific gene relationship} over a population gives a novel univariate dependence measure that can detect any non-linear, non-monotone relationship. Together with a consistent nonparametric estimator, we establish its robustness on both the population and empirical levels. Simulations and real data analysis show that this measure outperforms existing independence measures like Pearson, Kendall's $\tau$, $\tau^\star$, distance correlation, HSIC, Hoeffding's D, HHG, and MIC, on various tasks. 
\end{abstract}

%%%%%%%%%%%%%%%%%%%%%%%%%%%%%%%%%%%%%%%%%%%%%%
%% Please use \tableofcontents for articles %%
%% with 50 pages and more                   %%
%%%%%%%%%%%%%%%%%%%%%%%%%%%%%%%%%%%%%%%%%%%%%%
%\tableofcontents

\section{Introduction}
Experimental biologists and clinicians seek a deeper understanding of biological processes and their link with disease phenotypes by characterizing cell behavior. Gene expression offers a fruitful avenue for insights into cellular traits and changes in cellular state.  Advances in technology that enable the measurement of RNA levels for individual cells via Single-cell RNA sequencing (scRNA-seq) significantly increase the potential to advance our understanding of the biology of disease by capturing the heterogeneity of expression at the cellular level \citep{Haque:2017}.  Gene differential expression analysis, which contrasts the marginal expression levels of genes between groups of cells, is the most commonly used mode of analysis to interrogate cellular heterogeneity. 
By contrast, the relational patterns of gene expression have received far less attention. The most intuitive relational effect is gene co-expression, a synchronization between gene expressions, which can vary dramatically among cells. Converging evidence has revealed the importance of co-expression among genes.  When looking at a collection of highly heterogeneous cells, such as cells from multiple cell types, significant gene co-expression may indicate rich cell-level structure. Alternatively, when looking at a batch of highly homogeneous cells, gene co-expression could imply gene cooperation through gene co-regulation \citep{raj2006stochastic,Emmert-Streib:2014}. %On the other hand,
Biochemistry offers a complementary motivation for the advantages of studying co-expression in addition to marginal expression levels of genes. The biological system of a cell is generally described by a non-linear dynamical system in which gene expression is variable \citep{raj2006stochastic}. Therefore, the observed gene expression level varies by time and condition, even within the same cell, while the cooperation between genes is more stable over time and condition. For this reason, it can be argued that co-expression may more reliably characterize the biological system or state of the cell \citep{dai2019cell}.  scRNA-seq,
allows us to investigate gene co-expression at different resolutions, to understand not only how genes interact with each other within different cells, but also how the interactions relate to cell heterogeneity.

The recent work by \cite{dai2019cell} attempts an ambitious task: characterizing the gene co-expression at a single cell level (termed  ``cell-specific network'' CSN). Specifically, for a pair of genes and a target cell, \citet{dai2019cell} construct a 2-way $2\times2$ contingency table test by binning all the cells based on whether they are in the marginal neighborhoods of the target cell and assigning the test results as a binary indicator of gene association in the target cell. Viewed over all gene pairs, the result is a cell-specific gene network.  %Emphasising less on the interpretation of detected associations, they use this concept more as a data transformation method: the aggregated detected associations for a gene in a cell is treated as a new measure of single-cell gene activeness level.
Forgoing interpretation of the detected associations, they utilize the CSN to obtain a data transformation.  Specifically, they replace the transcript counts in the gene-by-cell matrix with the degree sequence of each cell-specific network. %, considering this as a new measure of single-cell gene activeness level.
Although this data transformation shows encouraging success in various downstream tasks, such as cell clustering, it remains unclear what the detected ``cell-specific'' gene association network really represents. %, or what is the statistical mechanism behind this data transformation. 
The implementation details and interpretation of the results are presented at a heuristic level, making   
%Most interpretation in the paper remains on a highly heuristic level as well as the parameter choices. All these aspects make 
it difficult for others to appreciate and generalize this line of work. 

%A follow-up work, avgCSN \citep{wang2021constructing},  takes a first step in decipher this approach \citep{dai2019cell}. 
In a follow-up paper, \cite{wang2021constructing}  take the first steps to capitalize on the CSN approach by redirecting the concept to obtain an estimator of co-expression.   Specifically,
they propose averaging the ``cell specific" gene association indicators over cells in a class to recover a global measure of gene association (avgCSN). The resulting measure performs remarkably well in certain simulations and detailed empirical investigations of brain cell data. Compared to Pearson's correlation, the avgCSN gene co-expression appears less noisy and provides more accurate edge estimation in simulations. It is also more powerful in a test to uncover differential gene networks between diseased and control brain cells. Finally, it provides biologically meaningful gene networks in developing cells. 
%global gene association does have more interpretability and better performance compared to Pearson: it outputs less noisy gene co-expression and performs much better in terms of accurate edge estimation in simulation, as well as better uncovering of meaningful biological differential gene networks differences between diseased and healthy cases in real data analysis. 

The empirical success of avgCSN likely lies in the nature of gene expression data: often noisy, sparse and heterogeneous, meaning not all cells exhibit co-expression at all times due to cellular state and conditions. For this reason, a successful method must be robust and sensitive to local patterns of dependencies. %{\color{red} As avgCSN is essentially an average of a series of local contingency table tests, the testing nature limits error occurrence, while the non-negative summands ensure that local patterns are detected. To be more specific, non-negative summands ensure that detected local patterns are not canceled out; on the contrary, measures like Pearson can have both negative and positive summands, and therefore final value can be small even if the dependence structure is clear, as negatives can cancel out the positives.} 
{\color{black} Being an average of a series of binary local contingency table tests, the error in each entry of avgCSN is limited, meanwhile the non-negative summands ensure that local patterns are not cancelled out. By contrast, measures like Pearson's correlation can have both negative and positive summands, and therefore the final value can be small even if the dependence structure is clear for a subset of the cells.} To make the method more stable, \cite{wang2021constructing} proposed some heuristic and practical techniques to compute avgCSN, for which we would like to have more principled insights. Examples are the choice of window size in defining neighborhoods in the local contingency table test, the choice of thresholding in constructing an edge, and the range of cells to aggregate over. Many natural questions emerge: how does avgCSN relate to other gene co-expression measures and the full range of general univariate dependence measures, and why does it perform well in practice?  Through theoretical analysis and extensive experimental evaluations, we address these questions, revealing that avgCSN is an empirical estimator of a new dependency measure, which enjoys various advantages over the existing measures. 

For comparison, we briefly review the related work in gene co-expression measures and general univariate dependence.
Since the work by \citet{eisen1998cluster}, Pearson's correlation has been the most popular gene co-expression measure for its simple interpretation and fast computation. However, Pearson's correlation fails to detect non-linear relationships and is sensitive to outliers. Another class of co-expression methods is based on mutual information (MI) \citep{bell1962mutual,steuer2002mutual,daub2004estimating}. The computation of MI involves discretizing the data and tuning parameters, and the dependence measure does not have an interpretable scale.  \citet{reshef2011detecting} proposed the maximal information coefficient (MIC) as an extension of MI, but MIC was shown to be over-sensitive in practice. More comparisons of different co-expression measures and the constructed co-expression networks can be found in \cite{song2012comparison,allen2012comparing}.

In the broader statistical literature, the problem of finding gene co-expression is closely related to that of detecting univariate dependence between two random variables. Specifically, for a pair of univariate random variables $X, Y$, how to measure the dependence between them has been a long-standing problem. The problem is often described as finding a function $\delta(X,Y)$, which measures the discrepancy between the joint distribution $F_{XY}$ and product of marginal distribution $F_{X}F_{Y}$. Numerous solutions to this problem have been provided: include the Renyi correlation \citep{renyi1959measures} measuring the correlation between two variables after suitable transformations; various regression-based techniques; Hoeffding’s D \citep{hoeffding1948non}, distance correlation (dCor) \citep{szekely2007measuring}, kernel-based measure like HSIC \citep{gretton2005measuring} and rank based measure like Kendall's $\tau$ and the refinement later, $\tau^\star$ \citep{bergsma2014consistent}. Most of these methods have not yet been widely adopted in genetics applications. 

 Aside from avgCSN, the methods mentioned so far do not specifically target dependence relationships that are local and often assume the data are random samples from a common distribution (in contrast with a mixture distribution) in the theoretical analysis. However, real gene interactions may change as the intrinsic cellular state varies and may only exist under specific cellular conditions. Furthermore, with data integration now being a routine approach to combat the curse of dimensionality, samples from different experimental conditions or tissue types are likely to possess different gene relationships and thus create more complex situations for detecting gene interactions.  In this setting, much like avgCSN, an ideal measure accumulates subtle local dependencies, possibly only observed in a subset of the cells. A co-expression measure that aims to detect local patterns, developed by \cite{wang2014gene}, counts the proportion of matching patterns of local expression ranks as the measure of gene co-expression. Specifically, they aggregate the gene interactions across all subsamples of size $k$. However, despite its promising motivation, it has low power to detect non-monotone relationships.  MIC \citep{reshef2011detecting} and HHG \citet{heller2013consistent} are also measures that attempt to account for local patterns of dependencies.

In this paper, we first give a detailed review of the related methods in \secref{aLDGintr}. Then in \secref{pop}, we show that avgCSN is indeed an empirical estimate of a valid dependence measure, which we define as averaged Local Density Gap (aLDG). In \secref{robpop} and \secref{emp}, we formally establish its statistical properties, including estimation consistency and robustness. We also investigate data-adaptive hyperparameter selection to justify and refine the heuristic choices in application in \secref{chooset}. Finally, we provide a systematic comparison of aLDG and its competitors via both simulation and real data examples in \secref{aLDGcompare}.

\section{A brief review of dependence and association measures}\label{sec:aLDGintr}
Before starting on the description of the various dependence measures, let us remark that \citet{renyi1959measures} proposed that a measure of dependence between two stochastic variables $X$ and $Y$, $\delta(X,Y)$, should ideally have the following properties:
\begin{enumerate}
    \item[(i)] $\delta(X,Y)$ is defined for any $X,Y$ neither of which is constant with probability $1$.
    
    \item[(ii)] $\delta(X,Y)$=$\delta(Y,X)$.
    
    \item[(iii)] $0\leq \delta(X,Y) \leq 1$.
    
    \item[(iv)] $\delta(X,Y)=0$ if and only if $X$ and $Y$ are independent.
    
    \item[(v)] $\delta(X,Y)=1$ if either $X=g(Y)$ or $Y=f(X)$, where $f$ anf $g$ are measurable functions.
    
    \item[(vi)] If the Borel-measurable functions $f$ and $g$ map the real axis in a one-to-one way to itself, then $\delta(f(X),g(Y))=\delta(X,Y)$.
\end{enumerate}
Particularly, a measure satisfying (iv) is called a strong dependence measure. 

Apart from the above properties, there are two more properties that are particularly useful in single-cell data analysis. Single-cell data often contain a significant amount of noise, among which outliers account for a non-negligible fraction. Therefore \emph{robustness} is a desirable property in a dependence measure. Specifically, keeping with previous literature \citep{dhar2016study}, by robustness we mean that the value of the measure does not change much when a small contamination point mass, far away from the main population, is added. A formal description and corresponding evaluation metric will be described later. Another often overlooked property is \emph{locality}, which is a relatively novel concept and has not been properly defined to the best of our knowledge. Nevertheless, this concept has been catching attention over the recent decade \citep{reshef2011detecting,heller2013consistent,heller2016consistent,wang2014gene}, especially in work motivated by genetic data analysis. \emph{Locality} targets a special kind of dependence relationship that is generally restricted to a particular neighborhood in the sample space. A natural example is dependence that occurs in some, but not necessarily all of the components in a finite mixture.  Another is dependence within a moving time window in a time series. Generally speaking, the interactions change as the hidden condition varies, or only exist under a specific hidden condition. A dependence measure that is \emph{local} should be able to accumulate dependence in the local regions.  

No measure has all of the properties mentioned above, as far as we know. Our new measure possesses all but properties (v) and (vi). In the following, we review a selected list of univariate dependence measures in more details.

%\subsection{General univariate dependence measure}
\subsection{Moment based measures}
The first class of methods is based on various moment calculations. The main advantage is fast computation and minimum tuning, while the main drawback is non-robustness to outliers from their moment-based nature.

\paragraph*{Pearson's correlation} The simplest measure is the classical Pearson’s correlation:
\begin{equation}
    \text{Pearson's}\ \rho(X,Y):= \frac{\text{Cov}(X,Y)}{\sqrt{\text{Var}(X)\text{Var}(Y)}}.
\end{equation}
Plugin the sample estimation of covariance and variance, consistency and asymptotic normality can be proven using law of large numbers and the central limit theorem, respectively. Pearson's $\rho$ has been, and probably still is, the most extensively employed measure in statistics, machine learning, and real-world applications, due to its simplicity. However, it is known to detect only linear relationships. Also, as is the case for regression, it is well known that the product-moment estimator is sensitive to outliers: even just a single outlier may have substantial impact on the measure. 

\paragraph*{Maximal correlation} 
The maximal correlation (MC) is based on Pearson's $\rho$. It is constructed to avoid the problem that Pearson's $\rho$ can easily be zero even if there is strong dependence. \citet{gebelein1941statistische} first propose MC as
\begin{equation}
    \text{MC}(X,Y) := \sup_{f,g} \rho(f(X),g(X)).
\end{equation}
Here the supremum is taken over all Borel-measurable functions $f,g$ with finite and positive variance for $f(X)$ and $g(Y)$. The measure MC can detect non-linear relationships, and in fact, it is a strong dependence measure. However, often MC cannot be evaluated explicitly except in special cases, because there does not always exist functions $f_0$ and $g_0$ such that $\text{MC} = \rho(f_0(X), g_0(Y))$. Also, it has been found to be overly ``sensitive'', i.e. it gives high value for distributions arbitrarily ``close'' to independence in practice.

\paragraph*{Distance correlation} A recent surge of interests has been placed on using distance metrics to achieve consistent independence testing against all dependencies. A notable example is the distance correlation (dCor) proposed by \citet{szekely2007measuring}:  
\begin{align}
    \text{dCor} (X,Y)& := \frac{V(X,Y)}{\sqrt{V(X,X)V(Y,Y)}}, \\
    & \quad \text{where } V(X,Y)=\mathbb{E}{|X-X'||Y -Y'|}
+ \mathbb{E}{|X - X'|}\mathbb{E}{|Y - Y'|}\\
& \quad \quad \quad \quad \quad \quad \quad \quad \quad 
- 2\mathbb{E}_{X,Y}\Big[{\mathbb{E}_{X'}|X-X'| \mathbb{E}_{Y'}|Y - Y'|}\Big], \nonumber
\end{align}
with $(X',Y')$ an i.i.d copy of $(X,Y)$. The distance correlation enjoys universal consistency against any joint distribution of finite second moments; however, in practice, it does not work well for non-monotone relationship \citep{shen2020distance}. Also, it is not robust from its moment based nature, as proven by \citet{dhar2016study}.

\paragraph*{HSIC} Recall the definition and formula for the maximal correlation, about which we mentioned it is difficult to compute since it requires the supremum of the correlation $\rho(f(X),g(Y))$ taken over Borel-measurable $f$ and $g$. In the framework of reproducing kernel Hilbert spaces (RKHS), it is possible to pose this problem and compute an analogue of MC quite easily. A state-of-the-art method in this direction is the so-called Hilbert-Schmidt Independence Criterion (HSIC) \citep{gretton2005measuring}. Denote the support of $X$ and $Y$ as $\mathcal{X}$ and $\mathcal{Y}$ respectively, HSIC considers $f,g$ to be in RKHS $\mathcal{F}$ and $\mathcal{G}$ of functionals on sets $\mathcal{X}$ and $\mathcal{Y}$ respectively. Then HSIC is defined to be the Hilbert-Schmidt (HS) norm of a Hilbert-Schmidt operator. We refer the reader to \cite{gretton2005measuring} for detailed description.  What might be of interest is that, in many cases, HSIC is equivalent to dCor.
% What might be of interst is a striking similarity of HSIC with dCor. Let $k(\cdot,\cdot), l(\cdot,\cdot)$ be the kernel function\footnote{From Riesz representation theorem, for every $x \in \mathcal{X}$, there exists $k_x \in \mathcal{F}$, such that $f(x) = \langle f, k_x \rangle$. Apply this to $f=k_x$ and another point $x' \in \mathcal{X}$, we have $k_x(x') = \langle k_x, k_x'\rangle := k(x,x')$. The function $(x,x') \mapsto k(x,x')$ from $\mathcal{X} \times \mathcal{X}$ to $\mathbb{R}$ is the kernel of RKHS $\mathcal{F}$. Similarly we define the kernel $l(y,y')$ for RKHS $\mathcal{G}$.} on $\mathcal{F}$ and $\mathcal{G}$. Then we can write HSIC as
% \begin{align}
%     HSIC(X,Y) & = \mathbb{E}\left[k(X,X')l(Y,Y')\right]
% + \mathbb{E}\left[k(X, X')\right]\mathbb{E}\left[l(Y,Y') \right]\nonumber\\
% & \quad \quad \quad \quad \quad \quad \quad \quad \quad 
% - 2\mathbb{E}_{X,Y}\left[{\mathbb{E}_{X'}k(X,X') \mathbb{E}_{Y'}l(Y, Y')}\right],
% \end{align}
% with normalization, it is just distance correlation with Euclidean distance replaced by kernel function. One has that if the kernels $k$ and $l$ are universal (universal kernel is a mild continuity requirement on the kernel) on compact domains $\mathcal{X}$ and $\mathcal{Y}$, then HSIC is a strong dependence measure.

\subsection{Rank based measure}
Another line of work based on ordinal statistics is developed in parallel to the moment-based methods. A random variable $X$ is called ordinal if its possible values have an ordering, but no distance is assigned to pairs of outcomes. Ordinal data methods are often applied to data in order to achieve robustness. 

\paragraph*{Spearman's $\rho_S$, Kendall's $\tau$ and $\tau^\star$} The two most popular measures of dependence for ordinal random variables $X$ and $Y$ are Kendall’s $\tau$ and Spearman’s $\rho_S$. Both Kendall's $\tau$ and Spearman's $\rho_S$ are proportional to sign versions of the ordinary covariance, which can be seen from the following expressions for the covariance:
\begin{align*}
    \text{Cov}(X,Y) &= \frac12\EE{(X - X')(Y - Y')} \propto \text{Kendall} \\
       & =\EE{(X' -X'')(Y' -Y''')} \propto \text{Spearman},
\end{align*}
where $(X',Y'), (X'',Y''), (X''', Y''')$ are i.i.d replications of $(X,Y)$. Note that Kendall's $\tau$ is simpler than Spearman's $\rho_S$ in the sense that it can be defined using only two rather than three independent replications of $(X, Y)$, so often Kendall's $\tau$ is preferred. A concern from certain applications is that Kendall's $\tau$ and Spearman's $\rho_S$ are not \emph{strong} dependence measures, so tests based on them are inconsistent for the alternative of a general dependence. In fact, it is often observed that they have difficulty detecting nonmonotone relationship. Later, an extension $\tau^\star$ \citep{bergsma2014consistent} mitigates such deficiency by modifying Kendall's $\tau$ to a strong measure. 

\paragraph*{Hoeffding's D and BKR}
Related to the ordinal statistics-based methods, another class of methods start from the cumulative distribution function (CDF), some of which are equivalent to ordinal forms due to the relationship between CDF and ranks. The oldest example is the Hoeffing's D proposed by \citet{hoeffding1948non}:
\[
\text{Hoeffing's D} := \mathbb{E}_{X,Y} \Big[(F_{X,Y} - F_{X}F_{Y})^2\Big],
\]
where $F_X$, $F_Y$, $F_{X,Y}$ are the CDF of $X$, $Y$, $(X,Y)$ respectively. Still, Hoeffing's D is not a strong measure, while its modified version BKR \citep{blum1961distribution}:
\[
\text{BKR} := \mathbb{E}_{X}\mathbb{E}_{Y} \Big[(F_{X,Y} - F_{X}F_{Y})^2\Big]
\]
is. It turns out Hoeffding's D belongs to a more general family of coefficients, which can be formulated as
\[
\text{C}_{gh} := \int g(F_{X,Y} - F_{X}F_{Y}) d h(F_{XY})
\]
for some $g$ and $h$. We will abbreviate Hoeffding's D as HoeffD in the figures in the remainder of paper.

\subsection{Dependence measures aware of local patterns}

Most of the methods mentioned so far do not specifically target dependence relationships that can be local in nature. In the following, we describe a few measures that were designed to capture complex relationships, whether local or not. 

\paragraph*{Maximal Information Coefficient} The idea behind the Maximal Information Coefficient (MIC,\cite{reshef2011detecting} statistic consists in computing the mutual information locally over a grid in the data set and then take as statistic the maximum value of these local information measures over a suitable choice of grid. However, several examples were given in \citet{simon2014comment} and \citet{gorfine2012comment} where MIC is clearly inferior to dCor.

\paragraph*{HHG}
\citet{heller2013consistent} pointed out another way to account for local patterns: that is, looking at dependence locally and then aggregating the dependence over the local regions. The local regions is simply defined as bins via partitioning the sample space.
Additionally, HHG takes a multi-scale approach: multiple sample space partitions are conducted, and results are aggregated over all of them. This results in a provably consistent permutation test. However, the cost of implementation is significantly longer computation time than its competitors: it takes $O(n^3)$ computation time while its competitors normally take at most $O(n^2)$.

\paragraph*{Matching ranks} Another method that developed specifically for accounting local pattern 
is proposed by \citep{wang2014gene}. Given $n$ pair of observations of $(X,Y)$, $\{(x_i,y_i)\}_{i=1}^n$, they propose to count the number of size $k$ subsequences $(x_{i_1}, x_{i_2}, \dots x_{i_k})$ and $(y_{i_1}, y_{i_2}, \dots y_{i_k})$ such that their rank is matched. We refer to this measure as MR (Matching Ranks).  Specifically, we write the scaled version of MR such that it is in range [0,1]:
\begin{align*}
    \text{MR} :=   \frac{1}{2{n\choose k}}\sum_{1\leq i_1<i_2\dots<i_k \leq n} & \Big(\ident\{ rank(x_{i_1}, x_{i_2}, \dots x_{i_k}) = rank(y_{i_1}, y_{i_2}, \dots y_{i_k})\} \\
    & + \ident\{ rank(x_{i_1}, x_{i_2}, \dots x_{i_k}) = rank(-y_{i_1}, -y_{i_2}, \dots -y_{i_k})\}\Big),
\end{align*}
where $rank(a_1,\dots,a_k) = (r(a_1),\dots,r(a_k))$ where $r(a_i)$ is the rank of element $a_i$ within the sequences $(a_1,\dots,a_k)$, and the equality inside the indicator function applies element-wisely. Though claimed to be able to detect complex relationship, this measure is inferior to others in some non-monotone dependence case like quadratic relationship. 

\section{Our method: averaged Local Density Gap}\label{sec:aLDGdef}

%In the remainder of the paper, we consider only a pair of random variables $X,Y$ whose joint and marginal densities exist and have the same support, and denote $f_{XY}, f_{X}, f_{Y}$ as their joint and marginal densities. Also, let $\widehat{f}_{XY}, \widehat{f}_{X}, \widehat{f}_{Y}$ be the estimated densities given observations of $(X,Y)$, and $\widehat{p}_{X,Y}(x,y)$ be the proportion of samples points in a square of side length $h$ centered at $(x, y)$, and $\widehat{p}_{X}$ and $\widehat{p}_{Y}$ be defined similarly for the marginal distribution. Denote the sample size as $n$, and the $i$-th observation of $(X,Y)$ as $(X_i, Y_j)$.

First, we elaborate on the origin of our work, which was inspired by gene co-expression analysis using single-cell data. In the context of gene co-expression analysis, the pair of random variables $X,Y$ represents the expression level of a pair of genes, and the goal is to find the relationship between them. Pearson's correlation is one commonly used metric for this task. In light of the many shortcomings of this global measure of dependence, \citet{dai2019cell} proposed to characterize the gene relationships for every cell. Their method takes the following approach: for the gene pair $(X,Y)$, and a target cell $j$, partition the $n$ samples based on whether $|X_{\cdot} - X_j| < h_x$ and $|Y_{\cdot} - Y_j| < h_y$, where $h_x$ and $h_y$ are predefined window sizes. This partition can be summarized as a  $2 \times 2$ contingency table (\tabref{contingency}). Then evidence against independence in this $2\times 2$ table can be quantified by a general contingency table test statistic. \citet{dai2019cell} uses 
\begin{equation}
    S_{X,Y}^{(j)} := \frac{ \sqrt{n} \left(n_{x,y}^{(j)}n  - n_{x,\cdot}^{(j)} n_{\cdot,y}^{(j)}\right)}{\sqrt{n_{x,\cdot}^{(j)} n_{y}^{(j)} (n-n_{x,\cdot}^{(j)}) (n-n_{\cdot,y}^{(j)})}},
\end{equation}
and conducts a one-sided $\alpha$ level test based on its asymptotic normality, that is
\begin{equation}\label{contingency_test}
    I_{XY}^{(j)}:= \mathbb{I}\{S_{X,Y}^{(j)} > \Phi^{-1}
(1-\alpha)\}.
\end{equation}

\begin{table}[H]
\centering
    \begin{tabular}{c|c|c|c}
                   &  $|Y_{\cdot} - Y_j| \leq h_y$ &  $|Y_{\cdot} - Y_j| > h_y$ &   \\  \hline
     $|X_{\cdot} - X_j| \leq h_x$   &      $n_{x,y}^{(j)} $     &   & $n_{x,\cdot}^{(j)}$   \\  \hline
       $|X_{\cdot} - X_j| > h_x$   &            &  &    \\  \hline
      &  $n_{\cdot,y}^{(j)}$  &   & $n$
    \end{tabular}
    \caption{The $2\times 2$ contingency table based on distance from $j$-th sample.}
    \label{tab:contingency}
\end{table}
%{\color{red} \citet{dai2019cell} use $I_{XY}&{(j)}$ to indicate whether or not gene pairs $X$ and $Y$ are dependent in cell $j$, and refer to the detected dependence as \emph{local dependence}. Though interesting as a novel concept, it does not have sensible probabilistic meaning: in a simple mixture of independent Gaussians, the \emph{local dependence} \citet{dai2019cell} referred to should be zero, but $I_{XY}$ they used gives one for most of the cells. We formally claim that, $I_{XY}$ infers a way of violating \emph{local independence}, where we define $X$ and $Y$ being \emph{local independent} at position $(x,y)$ as 
{\color{black} \citet{dai2019cell} claim that  $I_{XY}{(j)}$ indicates whether or not gene pairs $X$ and $Y$ are dependent in cell $j$, and refer to the detected dependence as \emph{local dependence}. Though interesting as a novel concept, it lacks rigor and interpretability. Alternatively we propose to define $X$ and $Y$ as being \emph{locally independent} at position $(x,y)$ as 
\begin{equation}
    f_{XY}(x,y) = f_{X}(x)f_{Y}(y),
\end{equation}
then $I_{XY}$ provides a way of assessing \emph{local independence}.
Specifically, as a one-sided test, $I_{XY}(j)$ assesses whether or not $f_{XY}(x,y) > f_{X}(x)f_{Y}(y)$, at position $(x,y)$ marked by cell $j$.  To assess global independence, aggregation, as proposed by \citet{wang2021constructing}, is needed. Their empirical measure can be formally written as:
\begin{equation}\label{avgcsn}
    \text{avgCSN} := \frac{1}{n} \sum_{i=1}^n I_{XY} ^{(j)}.
\end{equation}
}
% One simple counter intuitive case is that, when $X,Y$ follows a two component Gaussian mixture distribution, e.g. their joint density 
% \[f_{X,Y} \propto 0.5 N_2([0,0]^\top,\mathbf{I}) + 0.5 N_2([1,1]^\top,\mathbf{I}),\]
% where $N_2$ represents the bivariate normal density. Then $X\not\perp Y$, but inside each component $X\perp Y$. From the conceptual level, one would expect $I_{XY}^{(j)}=0$ for most $j$, as it is a measure of \emph{local} dependence. However simulation shows 
% $I_{XY}^{(j)}=1$ instead for majority of $j$, which contradicts with its meaning. 
% It turns out that, $I_{XY}^{(j)}$ is not a valid test for local dependence, but it could be a valid contributor to global dependence. Following this intuition, \citet{wang2021constructing} proposed averaging the cell specific gene-gene relationship indicator across cells to represent a overall dependence level:
% A heuristic justification of avgCSN as a sensible dependence measure goes as follows: if $X$ and $Y$ are dependent and have a continuous joint density, then there exists at least one point $(x_0, y_0)$ in the sample space of $(X, Y)$, with radii $h_x$ and $h_y$ around $x_0$ and $y_0$, respectively, such that the joint distribution of $X$ and $Y$ is different than the product of the marginal distributions in the cartesian product of balls around $(x_0,y_0)$. Because it is impossible to choose the oracle position, we aggregate over all the possible locations, i.e. all the given sample points, to discover any points of dependence. 
% A more rigorous justification requires further notation. 
{\color{black}Some simple approximations gives us a population correspondence of avgCSN.} Assume the variables $X,Y$ have joint density $f_{XY}$, and marginal densities, $f_{X}$ and $f_{Y}$, that have common support.  Let $\widehat{f}_{XY}, \widehat{f}_{X}, \widehat{f}_{Y}$ be the estimated densities given observations of $(X,Y)$. 
Under the assumption that the bandwidth $h_x, h_y\to 0$ and $ \sqrt{h_x h_y n}\to\infty$, with some simple algebra (see \appref{derive} for detailed derivation), we see that
\begin{align}\label{deriveavgcsn}
    \text{avgCSN} &\approx \frac{1}{n} \sum_{i=1}^n \ones\left\{ \frac{\widehat{f}_{X,Y}(x_{i}, y_{i}) - \widehat{f}_{X}(x_{i}) \widehat{f}_{Y}(y_{i}) }{\sqrt{ \widehat{f}_{X}(x_{i})\widehat{f}_{Y}(y_{i})}} \geq t_n \right\},\quad \text{where } t_n = \frac{\Phi^{-1}(1-\alpha)}{\sqrt{n h_x h_y}},
\end{align}
and $\alpha\in[0,1]$ is some hyperparameter related to the test level of the local contingency test (usually $\alpha$ is set to 0.05 or 0.01). Because $t_n \downarrow 0$ as $n$ goes to infinity, we naturally think of the following population dependence measure:
\begin{equation*}
  \text{Pr}_{X,Y}\left\{ \frac{f_{X,Y}(X,Y) - f_X(X) f_{Y}(Y)}{\sqrt{f_{X}(X)f_{Y}(Y)}} > 0\right\}.
\end{equation*}

In the remainder of this section, we formally define a generalized version of this measure in \secref{pop}, along with its properties on the population level. Then we discuss consistent and robust estimation in \secref{emp} and provide guidance on hyper-parameter selection in \secref{chooset}. Finally, we comment on the relationship between our measure and some of the previous work in \secref{relation}.

% \begin{table}[H]
% {\scriptsize
% \begin{tabular}{c|c|c|c|c|c}
% \hline
%  Method  &  Expression (population level) & Computation & Robust & Strong & No-linearity \\ [10pt] 
%  \hline 
% Pearson & $\EE{XY}-\EE{X}\EE{Y}$ & $O(n)$ & No & No & No \\ [10pt] 
% Spearman's $\rho_S$ & $ \EE{\delta\Big((X_1-X_2)(Y_1-Y_2)\Big)}$  & $O(n\log{n})$ & $\checkmark$ & No & $\checkmark$ (monotone) \\ [10pt] 
% $\tau^{*}$  & $\EE{\delta\Big(a(X_1,X_2,X_3,X_4)\Big)\delta\Big(a(Y_1,Y_2,Y_3,Y_4)\Big)}$ & $O(n^2\log{n})$ & $\checkmark$ & $\checkmark$ & $\checkmark$\\ [10pt] 
% dCor
% & $ \EE{a(X_1,X_2,X_3,X_4)a(Y_1,Y_2,Y_3,Y_4)}^{\frac{1}{2}}$  & $O(n\log{n})$ & No & $\checkmark$ & $\checkmark$ (monotone)\\ [10pt] 
% $\textnormal{HSIC}$  & $ ||\EE{K(X) L(Y)} - \EE{K(X)}\EE{L(Y) }||_{HS}$ & $O(n^2)$ & $\checkmark$ & $\checkmark$ & $\checkmark$\\ [10pt] 
% $\textnormal{Hoeffdings'D}$ & $ \EE{F_{XY}-F_{X}F_{Y}}$ & $O(n^2)$ & $\checkmark$ & $\checkmark$ & $\checkmark$\\ [10pt] 
% \textbf{aLDG (new)} & $\EE{\delta(f_{XY} - f_X f_Y)}$ &  $O(n\log{n})$ & $\checkmark$ & $\checkmark$ & $\checkmark$\\
% \hline
% \end{tabular}
% }
% \caption{\label{tab:summary} Summary of popular dependency measure of a pair of random variables X and Y\footnote{For the expression we omit the constant factors and normalization. And the functions $\delta := \ones\{z > 0\}$ and $a(z_1, z_2, z_3, z_4) := \{|z_1-z_2|+|z_3-z_4|-|z_1-z_3|-|z_2-z_4|\}$, each $\{\cdot_{i}\}$ are i.i.d copies; and $K,L$ are kernel transformation.}.
% }
% \end{table}

\subsection{Definition and basic properties}\label{sec:pop}
\begin{definition}(averaged Local Density Gap)
Consider a pair of random variables $X,Y$ whose joint and marginal densities both exist, and denote $f_{XY}, f_{X}, f_{Y}$ as their joint and marginal densities. The averaged Local Density Gap (aLDG) measure is then defined as 
\begin{equation}
    \text{aLDG}_t := \text{Pr}_{X,Y}\left\{ T(X,Y) > t\right\}, \quad \text{where } T(X,Y):= \frac{f_{X,Y}(X,Y) - f_X(X) f_{Y}(Y)}{\sqrt{f_{X}(X)f_{Y}(Y)}}
\end{equation}
and $t\geq 0$ is a tunable hyper-parameter.
\end{definition}

\noindent
From the definition, one can immediately realize the following lemma.
\begin{lemma}\label{lem:simpleprop}
For a pair of random variables $X,Y$ whose joint and marginal densities both exist, we have
\begin{enumerate}
    \item $X \perp Y  \Longleftrightarrow \text{aLDG}_0 =0 $;
    
    \item  if $t>0$, then  $X \perp Y \Longrightarrow \text{aLDG}_t =0$;

    \item $\text{aLDG}_t \text{ is non-increasing with regard  } t$ for all $t \geq 0$;
    
    \item $\text{aLDG}_t \in [0,1]$;
    
    \item $\text{aLDG}_t(X,Y) = \text{aLDG}_t(Y,X)$;
\end{enumerate}
\end{lemma}
As a concrete example of the $\text{aLDG}$ measure, the left plot of \figref{aLDGpop} displays $\text{aLDG}$, given different $t$ for a bivariate Gaussian with different choices of correlation. We can see that (1) $\text{aLDG}_t$ is non-increasing with regard $t$ as our \lemref{simpleprop} suggests; (2) $\text{aLDG}_t$ equals zero at independence for all $t\geq 0$, while $\text{aLDG}_0$ equals zero if and only if there is no dependence, as our \lemref{simpleprop} suggests; (3) $\text{aLDG}_t$ increases with the dependency level, indicating that it is a sensible dependence measure.

Note that, from \lemref{simpleprop}, $\text{aLDG}_0$ is a \emph{strong}\footnote{Recall that a measure of dependence between a pair of random variable $X,Y$ is \emph{strong} if it equals zero if and only if $X$ and $Y$ are independent.} measure of dependence. While being strong is a desirable feature of a dependence measure, for $\text{aLDG}$ type of measure, we find that it comes with the sacrifice of robustness under independence (\propref{indeprob}). On the other hand, setting $t>0$ could result in insensitivity under weak dependence, but with a provable guarantee of robustness (\thmref{aLDGrobpop}). In summary, the hyper-parameter $t$ serves as a trade-off between robustness and sensitivity. In \secref{chooset} we will discuss the practical choice of $t$ in more detail. For now, we treat it as a predefined non-negative constant.

\subsection{Robustness analysis}\label{sec:robpop}
In the following, we present a formal robustness analysis. An important tool to measure the robustness of a statistical measure is the influence function (IF). It measures the influence of an infinitesimal amount of contamination at a given value on the statistical measure. The Gross Error Sensitivity (GES) summarizes IF in a single index by measuring the maximal influence an observation could have. 
\begin{definition}[Influence function (IF) and Gross Error Sensitivity (GES)] Assume that the bivariate random variable $(X,Y)$ follows a distribution $F$, the influence function of a statistical functional $R$ at $F$ is defined as
\begin{align}\label{if}
    \text{IF}\big((x,y), R, F\big) := \lim_{\epsilon \to 0} \frac{R\big((1-\epsilon) F + \epsilon \delta_{(x,y)}\big) - R(F)}{\epsilon}
\end{align}
where $\delta_{(x,y)}$ is a Dirac measure putting all its mass at $(x,y)$. The Gross Error Sensitivity (GES) summarizes IF in a single index by measuring the maximal influence over all possible contamination locations, which is defined as
\begin{equation}
    \text{GES}(R, F) := \sup_{(x,y)} \mid \text{IF}\big((x,y), R, F\big)\mid.
\end{equation}
An estimator is called $B$-robust if its GES is bounded. 
\end{definition}

Among the related work we have mentioned, only the robustness of $\tau$, $\tau^\star$, and $\text{dCor}$ have been theoretically investigated to the best of our knowledge. \citet{dhar2016study} proved that $\text{dCor}$ is not robust while $\tau$ and $\tau^\star$ are. Their evaluation criteria is a bit different from ours. We investigate the limit of the ratio when the contamination mass goes to zero. They investigate the ratio limit when the contamination position goes far away, given fixed contamination mass. We argue that our analysis aligns better with the main statistical literature. In the following, we show that $\text{aLDG}_t$ with $t>0$ is $B$-robust, under some reasonable regularity conditions.

\begin{theorem}\label{thm:aLDGrobpop}
Consider $t>0$, and a bivariate distribution $F$ of variable $(X,Y)$ whose joint and marginal densities exist as $f_{XY}$, $f_{X}$, $f_{Y}$, and satisfy
\begin{equation}
    f_{\text{max}}:=||\sqrt{f_{X}f_{Y}}||_{\infty} <\infty; \quad \quad  |\text{aLDG}_{t - \epsilon} - \text{aLDG}_{t}| \leq L\epsilon,\ \forall \ \epsilon >0;
\end{equation}
then we have
\begin{equation}
    \text{GES}(\text{aLDG}_t, F) \leq L f_{\text{max}} +1 < \infty.
\end{equation}
\end{theorem}

The proof of \thmref{aLDGrobpop} is in \appref{aLDGrobpop}.
The first assumption about the boundness of density is common in density based statistical analysis. The second assumption about the $\text{aLDG}_{t}$ smoothness may look less familiar, however after a transformation, it is no more than a CDF-smoothness assumption: recall that $T(X,Y) := \frac{f_{XY}(X)-f_{X}(X)f_{Y}(Y)}{\sqrt{f_{X}(X)f_{Y}(Y)}}$, then 
\begin{align}
     |\text{aLDG}_{t-\epsilon} - \text{aLDG}_{t}|< L\epsilon \Longleftrightarrow \mathbb{P}\{|T(X,Y)-t|\leq \epsilon\} \leq L\epsilon,
\end{align}
that is, the CDF of random variable $T(X,Y)$ is L-lipschitz around $t$ for $t>0$. In \figref{smooth} we show the empirical density of $T(X,Y)$ for bivariate Gaussian of different correlation, which is generally bounded by some constant $L$ at positive values.
\begin{figure}[H]
    \centering
    \includegraphics[width=1\linewidth]{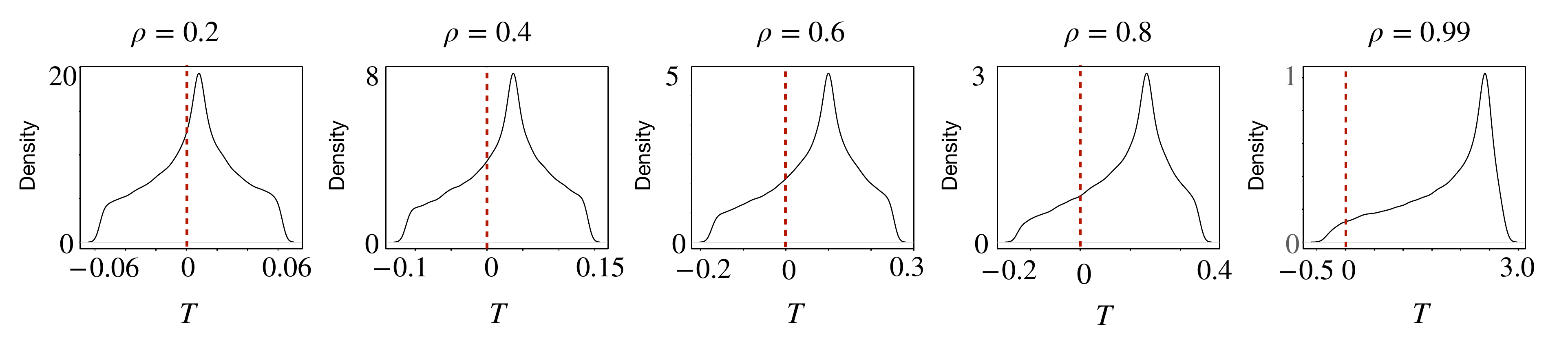}
    \caption{The empirical density of statistics $T$. The underlying bivariate distribution is Gaussian, and the value of $T$ is calculated using the true Gaussian density. We can see that, as the correlation increases, the density of $T$ near zero (annotated by the red dashed line) is smaller.}
    \label{fig:smooth}
\end{figure}
In the following, we show that $\text{aLDG}_0$ is not robust under independence. 
\begin{proposition}\label{prop:indeprob}
For any distribution $F$ over a pair of independent random variables $(X,Y)$ whose joint and marginal density exists and are smooth almost everywhere, we have
\begin{equation}
    \text{GES}(\text{aLDG}_0, F) = \infty
\end{equation}
if and only if X is independent of Y.
\end{proposition}
The proof of \propref{indeprob} is in \appref{indeprob}. 
The right plot in \figref{aLDGpop} provides some empirical evidence of the non-robustness of $\textnormal{aLDG}_0$ under independence. Specifically, we plot the population value of the ratio inside limitation \eqref{if}, under bivariate Gaussian with small enough contamination proportion $\epsilon$, to approximately show that the IF value of $\textnormal{aLDG}_t$ at independence indeed goes to infinity as $t$ goes to zero. 

\begin{figure}[H]
    \centering
    \includegraphics[width=\linewidth]{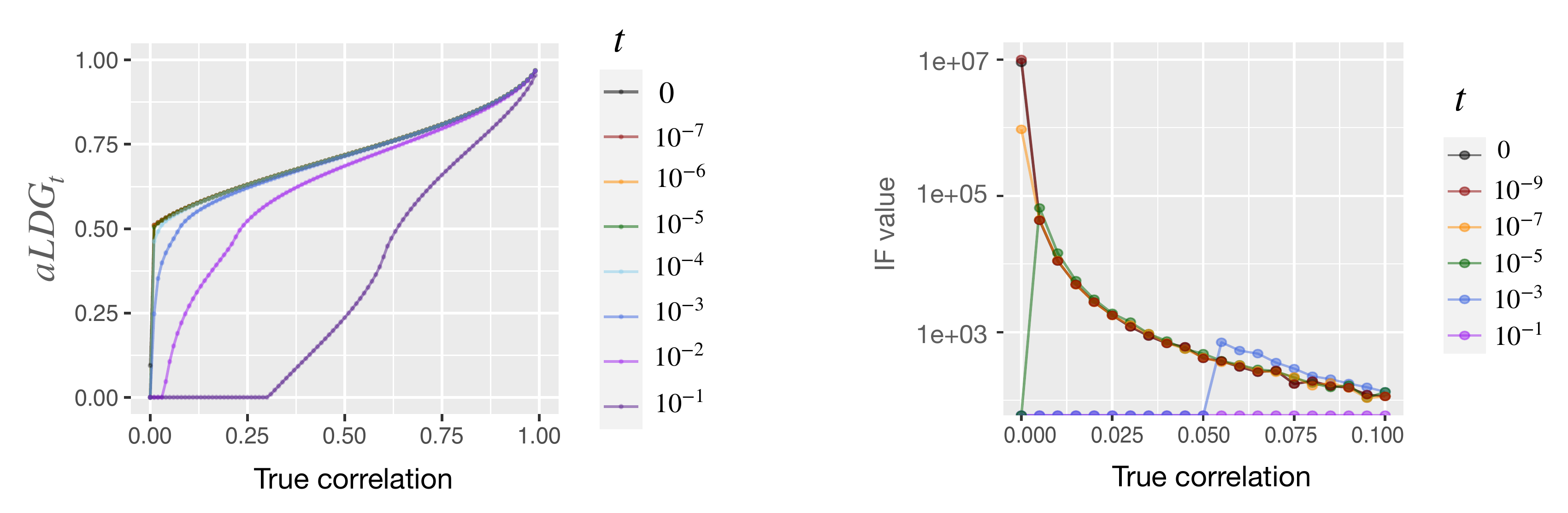}
    \caption{\textbf{(Left)} The true $\textnormal{aLDG}_t$ value for bivariate Gaussian with different levels of correlation under different choices of $t$. \textbf{(Right)} The influence function value approximated by setting the contamination proportion very small ($\epsilon = 10^{-6}$).}
    \label{fig:aLDGpop}
\end{figure}

% \begin{definition}
% For any $t>0$, a function $g: \mathcal{X}\to \mathbb{R}$ is said to be \textbf{smooth at level $t$ with respect to measure $P$}, if there exists constants $L>0$, $\epsilon_0>0$ such that for all $0< \epsilon < \epsilon_0$, 
% \begin{equation}
%     P \{ \bm{x}\in \mathcal{X}:  |g(x)-t|\leq \epsilon\} \leq L \epsilon.
% \end{equation}
% \end{definition}

\subsection{Consistent and robust estimation}\label{sec:emp}
In this section we investigate estimation of $\text{aLDG}_t$ given finite samples. One natural way to estimate $\text{aLDG}_t$ is using the following plug-in estimator: recall that $\widehat{f}_{XY}, \widehat{f}_{X}, \widehat{f}_{Y}$ are the estimated joint and marginal densities, then given $n$ observations $\{(x_1,y_1),\dots,(x_n,y_n)\}$ of $(X,Y)$, $\text{aLDG}_t$ can be estimated by 
\begin{align}
\label{eq:aLDGemp}
    \widehat{\text{aLDG}}_t & := \frac{1}{n}\sum_{i=1}^n \ones\left\{ \widehat{T}(x_i,y_i) \geq t \right\},\quad \text{where } \widehat{T}(x_i,y_i):=\frac{\widehat{f}_{X,Y}(x_{i}, y_{i}) - \widehat{f}_{X}(x_{i}) \widehat{f}_{Y}(y_{i}) }{\sqrt{ \widehat{f}_{X}(x_{i})\widehat{f}_{Y}(y_{i})}}
\end{align}
In the following, we establish the non-asymptotic high probability bound of the estimation error using the above simple plug-in estimator $\widehat{\text{aLDG}}_t$. The error rate is determined by the density estimation error for variable $X, Y$, as well as the probability estimation error for $T(X,Y)$. 
\begin{theorem}\label{thm:aLDGconsist}
Consider $t>0$, and a bivariate distribution $F$ of variable $(X,Y)$ whose joint and marginal densities exist as $f_{XY}$, $f_{X}$, $f_{Y}$, and satisfy
\begin{align*}
&\inf_{x,y}f_{XY}(x,y),\  \inf_{x}f_X(x) \inf_{y}f_Y(y) \geq c_{\min},\\
& \sup_{x,y}f_{XY}(x,y),\  \sup_{x}f_X(x) \sup_{y}f_Y(y) \leq c_{\max},
\end{align*}
and for some $\eta_n$ with $\lim_{n\to\infty}\eta_n \to 0$, with probability at least $1-\frac{1}{n}$ 
\begin{equation}
  ||\widehat{f}_{XY}-f_{XY}||_{\infty}, ||\widehat{f}_{X}-f_{X}||_{\infty}, ||\widehat{f}_{Y}-f_{Y}||_{\infty} \leq \eta_n; 
\end{equation} 
and for some constant $0<L<\infty$, 
\begin{equation}
    |\text{aLDG}_{t-\epsilon}-\text{aLDG}_t| \leq L\epsilon\quad  \text{for all} \ \epsilon>0.
\end{equation}
Then we have, with probability at least $1-\frac{2}{n}$, we have
\begin{equation}
    \left|\widehat{\text{aLDG}}_t - \text{aLDG}_t\right| \leq  LC\eta_n + \sqrt{\frac{2\log{n}}{n}},
\end{equation}
where $C$ depends only on $c_{\min}, c_{\max}$.
\end{theorem}
\thmref{aLDGconsist} is flexible in the sense that one can plug-in any kind of density estimator and its error rate to obtain the error rate of the corresponding $\widehat{\text{aLDG}}$ estimator. The proof of \thmref{aLDGconsist} is in \appref{pfconsist}. Though \thmref{aLDGconsist} was for fixed $t$, we also provide similar result that holds true uniformly over all possible $t$ in \appref{uniform}.

As for a concrete example, we provide explicit results for a special class of bivariate density and a simple density estimator. Specifically, we consider the true marginal density $f_X$, $f_Y$ that are L-Lipschitz, and the joint density $f_{XY}$ that are simply the product of $f_X$, $f_Y$; we also consider the following density estimator\footnote{The density estimator used here is not chosen to be minimax optimal. We instead design it to align the best with the practical methods \citet{dai2019cell} and \citet{wang2021constructing}, such that we can better justify and refine their heuristic choices of hyperparameter by theory.}:
\begin{align}\label{densest}
    &\widehat{f}_{X}(\cdot) = \frac{1}{n} \sum_{j=1}^n K_{h_n}(\cdot, x_j) , \quad \widehat{f}_{Y}(\cdot) = \frac{1}{n} \sum_{j=1}^n K_{h_n}(\cdot, y_j), \nonumber\\
    & \quad \widehat{f}_{XY}(\cdot, \cdot) = \frac{1}{n} \sum_{j=1}^n K_{h_n}(\cdot, x_j)K_{h_n}(\cdot, y_j),
\end{align}
where $K_{h_n}(\cdot,u):=\ones\{|\cdot-u|\leq h_n\}/(2 h_n)$ is one-dimensional boxcar kernel smoothing function with bandwidth $h_n$. From \propref{densest} in \appref{densest}, the uniform estimation error rate $\eta_n$ in this setting is $O(n^{-1/6}\sqrt{\log{n}})$, given the asymptotic near-optimal bandwidth $h = O(n^{-1/6})$. Therefore, applying \thmref{aLDGconsist} gives us estimation error rate of $O(n^{-1/6}\sqrt{\log{n}})$ for $\textnormal{aLDG}_t$. 

% \figref{consist} demonstrates empirical evidence of the proven consistency of $\widehat{\text{aLDG}}_t$ using the above simple boxcar kernel product density estimator. Particularly, we simulate $X,Y$ as correlated Gaussian variables, and compute $|\widehat{\text{aLDG}}_t - \text{aLDG}_t|$ given different number of samples. 

We also include robustness analysis of $\widehat{\text{aLDG}}_t$ in \appref{emprob}. Specifically, we consider an empirical contamination model that is commonly encountered in single-cell data analysis: a small proportion of the sample points are replaced by ``outliers'' far away from the rest samples. We show that $\widehat{\text{aLDG}}_t$ with and without outliers are close as long as the outlier proportion is small. This suggests that the estimator of $\text{aLDG}_t$ preserves its robust nature.

\subsection{Selection of hyper-parameter $t$}\label{sec:chooset}
In this section, we propose two methods for selecting $t$, each of which has merit. We also provide guidance on which one is preferable in different practice settings. %We leave the choices to readers when it comes to using.

\paragraph*{Uniform error method} From the results in the previous section, we learn that $\text{aLDG}_0$ is not robust under independence. To prevent $\widehat{\text{aLDG}}_t$ from approaching $\text{aLDG}_0$ under independence, it is sufficient to make sure that the estimation error of $T$ under independence is uniformly dominated by $t$ with high-probability. To compute the uniform estimation error of $T$ under independence, we first manually construct the independence case via random shuffle. Given $n$ samples $\{(x_i,y_i)\}_{i=1}^n$ of $(X,Y)$, denote the corresponding empirical joint distribution as $\widehat{F}_{XY}$, and marginal joint distribution as $\widehat{F}_{X}$ and $\widehat{F}_{Y}$. Applying the random shuffle function $\pi$ on indices of one dimension (i.e. $Y$), we have 
\begin{equation}
    \{(x_i,y_{\pi(i)})\}_{i=1}^n \sim \widehat{F}_{X}\widehat{F}_{Y},
\end{equation}
that is the shuffled samples  $\{(x_i,y_{\pi(i)})\}$ now come from a different joint distribution where $(X,Y)$ are independent. 

We can then use the shuffled samples to compute the uniform estimation error of $T$ under independence. Note that $T$ under independence is exactly zero, therefore its uniform estimation error is just the uniform upper bound of its estimation. To stabilize the estimation of such upper bound, we use the median of estimated upper bound from $\max\{\lfloor 1000/n \rfloor, 5\}$ different random shuffles as the final estimation. We call this $t$ selection method the \emph{uniform error} method.

% Recall that we estimate $T$ as 
% \begin{equation}
%     \widehat{T}(x,y):=\ones\left\{ \frac{\widehat{f}_{X,Y}(x, y) - \widehat{f}_{X}(x) \widehat{f}_{Y}(y) }{\sqrt{ \widehat{f}_{X}(x)\widehat{f}_{Y}(y)}} \geq t \right\},
% \end{equation}
% where $\widehat{f}_{X,Y}$, $\widehat{f}_{X}$ and $\widehat{f}_{Y}$ are estimator of the joint and marginal density of $(X,Y)$ given $n$ samples of $(X,Y)$, therefore uniform estimation error of $T$ is just $\sup_{(x,y)}\widehat{T}(x,y)$. We further approximate this supreme by $\max_{i}\widehat{T}_i$, where
% \begin{equation}
%     \widehat{T}_i:=\ones\left\{ \frac{\widehat{f}_{X,Y}(x_{i}, y_{z}) - \widehat{f}_{X}(x_{i}) \widehat{f}_{Y}(y_{i}) }{\sqrt{ \widehat{f}_{X}(x_{i})\widehat{f}_{Y}(y_{i)})}} \geq t \right\}.
% \end{equation}

\paragraph*{Asymptotic norm method} When using $\textnormal{aLDG}_t$ in large-scale data analysis, choosing $t$ using the above data-dependent choice may be undesirable because it requires additional computations. In extensive simulations we observe that a simple alternative also performs fine in terms of maintaining consistency, power and robustness:
\begin{equation}\label{choosetnorm}
    t = \Phi^{-1}\left(1-\frac{1}{n}\right)\Big/\left(\sqrt{\sigma_{X}\sigma_Y} n^{1/3}\right).
\end{equation} 
 This choice is motivated by the following heuristic. Recall our derivation of aLDG statistics from avgCSN around \eqref{deriveavgcsn}: as the sample size $n$ goes to infinity, and $h_x, h_y \to 0$, $h_xh_yn\to\infty$, the empirical estimation of $\text{aLDG}_t$ using the boxcar kernel cioncide with avgCSN. Therefore, $t_n$ in \eqref{deriveavgcsn} could serve as a natural choice for $t$, but one need to be extra careful about $\alpha$, which is the test level of local contingency test \eqref{contingency_test} in definition towards avgCSN. We specically modify $\alpha$ to decrease with $n$ instead of a fixed value like $0.05$ since we desire consistency: i.e. $\text{aLDG}_t$ under independence should goes to zero as $n$ goes to infinity. Finally, plugging in our choice of bandwidth $h_x = \sigma_X n^{-1/6}$, $h_y = \sigma_Y n^{-1/6}$ together with the new $\alpha_n$ in place of $\alpha$ into $t_n$ \eqref{deriveavgcsn}, we get \eqref{choosetnorm}.  We call this $t$ selection method the \emph{asymptotic norm} method. 
 
Empirically we find that the \emph{asymptotic norm} method is often too conservative given the small sample size (which is expected since it is based on the asymptotic normality of a contingency table test statistic). In practice, we recommend people use \emph{uniform error} over \emph{asymptotic norm} when the sample size is not too big (e.g., no bigger than 200). When the sample size is big enough (e.g., bigger than 200), and the computation budget is limited, we recommend the \emph{asymptotic norm} method. In the rest of the paper, we use the \emph{uniform error} method when the sample size is no bigger than 200 and the \emph{asymptotic norm} method when the sample size is bigger than 200. We admit that there could be other promising ways of selecting $t$, for example, a geometry way we provided in \appref{chooset}. Here we only present the methods that we found working the best after a careful evaluation (see \appref{chooset}).

\subsection{Relationships to HHG}\label{sec:relation}
The method that is most similar to aLDG is HHG (\cite{heller2013consistent}).
%: it is consistent by using the rank information and the Pearson’s chi-square test but has better finite-sample testing powers over dCor in a collection of common non-linear dependencies. 
Like aLDG, HHG \citep{heller2013consistent} is based on aggregation of multiple contrasts between the local joint and marginal distributions 
\begin{align*}
   HHG := \sum_{i\neq j} M(i,j),\quad M(i,j) := (n-2) \frac{\Big(p_{XY}(B_{XY}^{i,j})-p_{X}(B_{X}^{i,j})p_{Y}(B_Y^{ij})\Big)^2}{p_{X}(B_{X}^{i,j})\Big(1-p_{X}(B_{X}^{i,j})\Big)p_{Y}(B_Y^{ij})\Big(1-p_{Y}(B_Y^{ij})\Big)},
\end{align*}
with $B_{X}^{i,j} = \{x: |x-x_i|\leq |x_i - x_j|\}$, $B_{Y}^{i,j} = \{y: |y-y_i|\leq |y_i - y_j|\}$ and $B_{XY}^{i,j} = B_{X}^{i,j}  \otimes B_{Y}^{i,j} $, $p_{XY}, p_X, p_Y$ are joint probability function for $(X,Y)$ and marginal probability function for $X$ and $Y$ respectively. While the two measures appear quite similar, they differ in two critical aspects.

\paragraph*{The efficiency of single scale bandwidth}  One notable difference between HHG and aLDG is that the former relies on a multi-scale choice of bandwidth for each sample point.  Specifically, it  utilizes multiple ($O(n)$) bandwidths for each data point. This results in a provably consistent permutation test; however, the cost of implementation is significantly longer computation time than its competitors. aLDG takes a single-scale approach, which considerably improves the computation efficiency.  Moreover, the aLDG formulation provides a direct analogy to a density functional, which allows us to exploit existing work in density estimation to determine an appropriate bandwidth. This single-scale approach, though may not optimal, achieves comparable power to HHG, as shown in the upcoming simulation studies.

\paragraph*{The merit of thresholding} Another difference is that empirically aLDG aggregates over thresholded summands, see \eqref{eq:aLDGemp}. It turns out thresholding brings implicit robustness to noise. By contrast, consider the non-thresholded version of aLDG:
\begin{equation}
    \text{aLDG}_{non}:= \EE{T(X,Y)}.
\end{equation}
Even with slight departures from independence, $\text{aLDG}_{non}$ can go to infinity.  For example, consider the following joint and marginal distribution that admits a kernel product density mixture:
\begin{align*}
    & f_{XY}(x,y) = \alpha k_{0,r}(x)k_{0,r}(y) + (1-\alpha)k_{0,1}(x)k_{0,1}(y),\\
    & f_{X}(x) = \alpha k_{0,r}(x) + (1-\alpha)k_{0,1}(x), \quad f_{Y}(y) = \alpha k_{0,r}(y) + (1-\alpha)k_{0,1}(y)
\end{align*}
where $\alpha \in (0,1)$, $0<r\ll1$ and  $k_{\mu,r}(\cdot):=\frac{1}{r}k(\frac{\cdot-\mu}{r})$, with $k$ as the density of 1-dim uniform distribution supported on $[-1,1]$. 

Note that as $\alpha\to0$ and $r\to 0$, the model is essentially an independence case contaminated with a small point mass.  Additionally with $\alpha/r \to \infty$, we can show that (see \appref{thred} for details)
\begin{equation}\label{nonthred}
   \EE{T(X,Y)} \approx \frac{\alpha}{r} \to \infty,
\end{equation}
that is the non-thresholded version of $\text{aLDG}$ is very large under such simple case of small departure from independence, therefore is problematic.  With thresholding, however, $\text{aLDG}$ is guaranteed to be approximately $\alpha$, which goes to zero for small perturbations, as one would desire.

%%%%%%%%%%%%%%%%%%%%%%%%%%%%%%%%%%%%%%%%%%%%%%%%%%%%%%%%%%%%%%%

\section{Empirical evaluation}\label{sec:aLDGcompare}

%%%%%%%%%%%%%%%%%%%%%%%%%%%%%%%%%%%%%%%%%%%%%%%%%%%%%%%%%%%%%%%
\subsection{Single-cell data application}\label{sec:real}
In this section, we evaluate aLDG among the other measures using scRNA-seq data from two studies.

\paragraph*{Chu dataset} This dataset \citep{chu2016single} contains 1018 cells
of human embryonic stem cell-derived lineage-specific progenitors. The seven cell types, including H1 embryonic stem cells (H1), H9 embryonic stem cells (H9), human foreskin fibroblasts (HFF), neuronal progenitor cells (NPC), definitive endoderm cells (DEC), endothelial cells (EC), and trophoblast-like cells (TB), were identified by fluorescence-activated cell sorting (FACS) with their respective markers. On
average, 9600 genes are measured per cell. In the following, we show some special gene pairs that exhibit strong, weak, or no relational patterns and the corresponding dependence values produced by different measures. We find that only aLDG gives a high value for strong relational patterns no matter how complex the pattern composition is; maintains near-zero values for known independent cases; and avoids a spurious relationship skewed by technical noise and sparsity (Figure \ref{fig:realbi}).

\begin{figure}[H]
    \centering
    \includegraphics[width=0.8
    \linewidth]{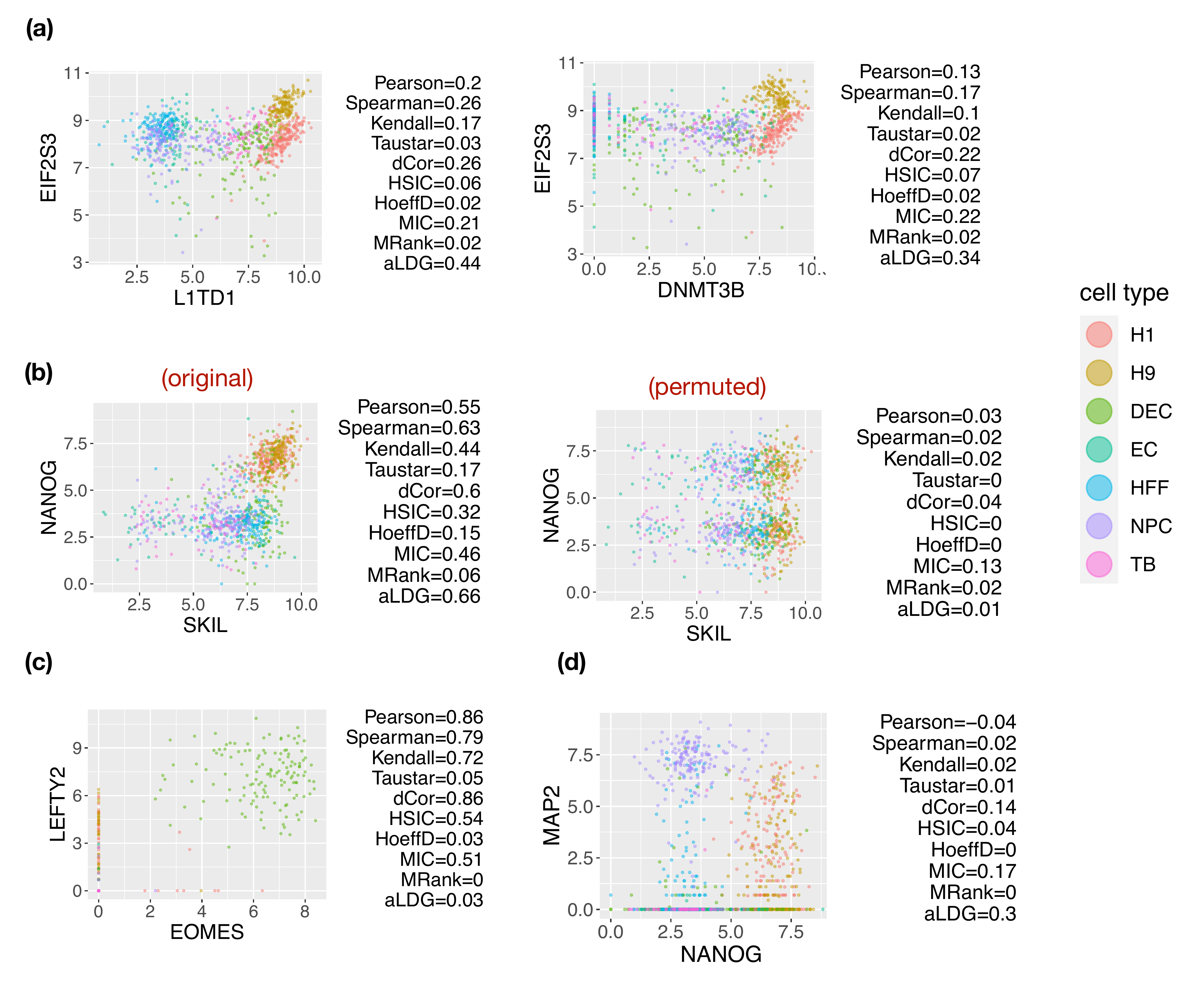}
    \caption{Example of gene pair scatter plots from the Chu dataset, which has 1018 cells from 7 cell types. Gene expression is recorded as counts per million (CPM) and $\log_2$ transformed. In each plot, we show the scatter plot of $\log_2(\text{CPM}+1)$ for a pair of genes and provide the corresponding estimated dependence values using different methods to the right of the plots.  \textbf{(a)} aLDG gives a much higher value than the others in these scenarios which appear to illustrate a strong mixture dependence pattern, even when the signal is predominantly in one cell type. \textbf{(b)} aLDG produces a high value for the obvious three mixture relationship in the first subplot. By contrast, in the second subplot, the cell identity are randomly shuffled for each gene pair, resulting in a constructed case of independence. Most measures, including aLDG, give near-zero values in this setting. The exception is MIC, which gives a misleadingly high value.  \textbf{(c)} This example illustrate performance when there is a high level of sparsity: MIC and the moment-based methods like Pearson, dCor, and HSIC provide estimates that are greatly overestimated, while aLDG, TauStar, and Hoeffding's D are not influenced by this phenomenon.   \textbf{(d)} This gene pair combines the challenge of sparsity with considerable noise: aLDG is still able to capture the less noisy, local cluster pattern in the upper left corner. }
    \label{fig:realbi}
\end{figure}

\paragraph*{Autism Spectrum Disorder (ASD) Brain dataset}
\citet{velmeshev2019single} includes scRNA-seq data
from an ASD study that collected 105 thousand nuclei from cortical samples taken from 22 ASD and 19 control samples. Samples were matched for age, sex, RNA integrity number, and postmortem interval. In the following, we compare control and ASD groups by testing for differences in their gene co-expression matrices using the sparse-Leading-Eigenvalue-Driven (sLED) test \citep{zhu2017testing}. sLED takes the gene co-expression matrices for both control and ASD groups as input, and outputs a $p$-value indicating the significance of their difference. This method is particularly designed to detect differential signals attributable to a small fraction of the genes. To emphasize the contrast with differentially expressed genes, \cite{wang2021constructing} call these differential network genes.  

Here we compare the power of the test for various co-expression measures. We use cells classified as L2/3 excitatory neurons (414 cells from ASD samples and 358 from control samples) and a set of 50 genes chosen randomly among the top 500 genes deferentially expressed between ASD and control samples. In addition, we manually add noise by randomly swapping 10\% of the control and ASD labels in the original data to see which measures detect the signal in the presence of greater noise. We omit HHG for this task as it requires too much computation time. Boxplots of $p$-values from sLED test across 10 independent trials (different random swapping each trial) are shown for all the remaining measures (\figref{vel23power}). Among the remaining measures, we find that HSIC, $\tau^\star$, Hoeffding's D, MIC, and aLDG perform well compared to Pearson, Spearman, Kendall, MRank and dCor. A visualization of the corresponding control versus ASD co-expression differences is displayed in \figref{vel23}, showing that the winners produce difference matrices with a few dominating entries, which is favored by the sLED test, while the others produce relatively flat and noisy patterns.

\begin{figure}[H]
    \centering
    \includegraphics[width=0.6\linewidth]{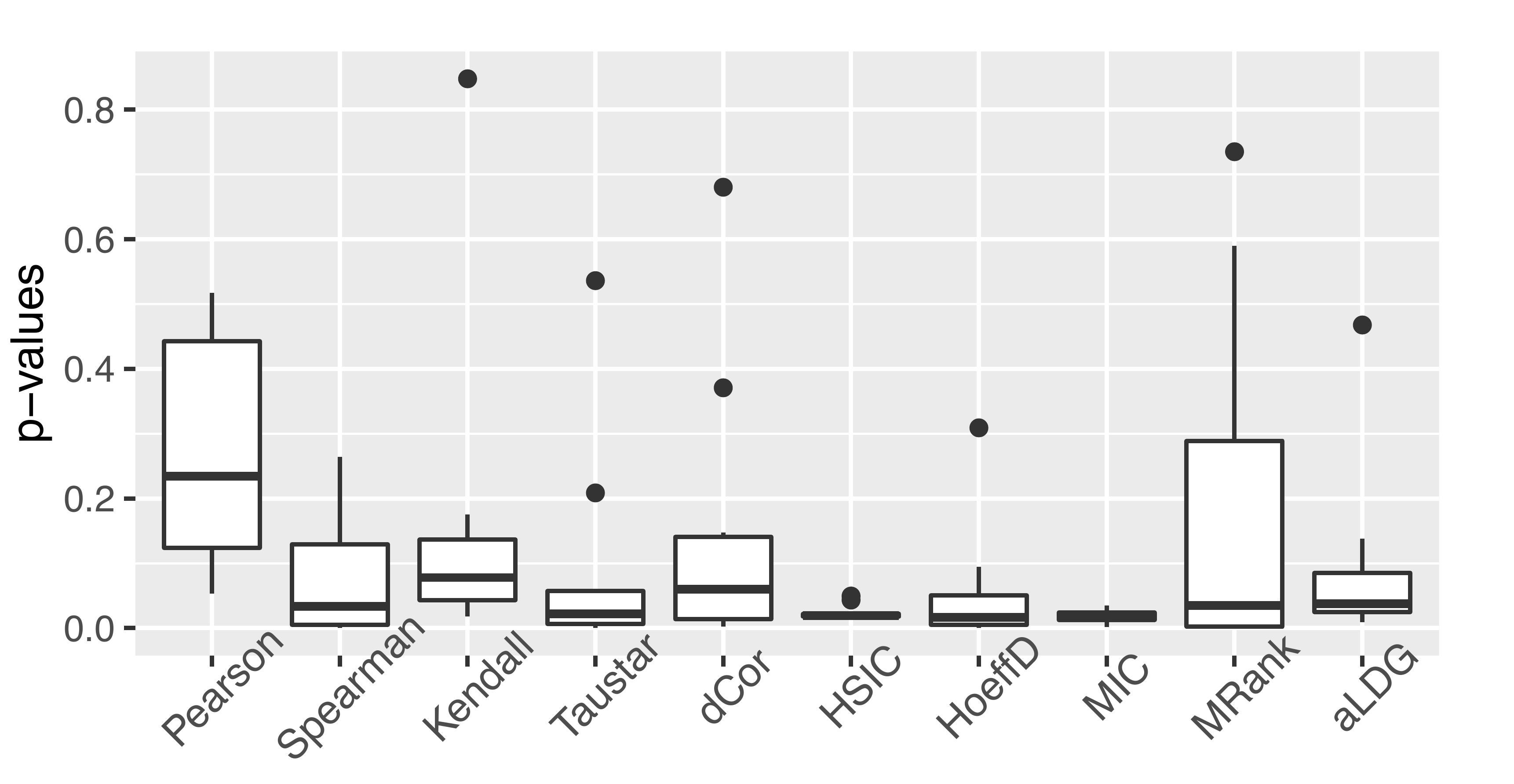}
    \caption{The estimated $p$-values obtained using sLED permutation tests for different dependency measures. We manually added noise by randomly swapping 10\% of the control and ASD labels in the original data to see which measures detect the signal in the presence of greater noise. Boxplots show the results from 10 independent repetitions. }
    \label{fig:vel23power}
\end{figure}

\begin{figure}[H]
    \centering
    \includegraphics[width=\linewidth]{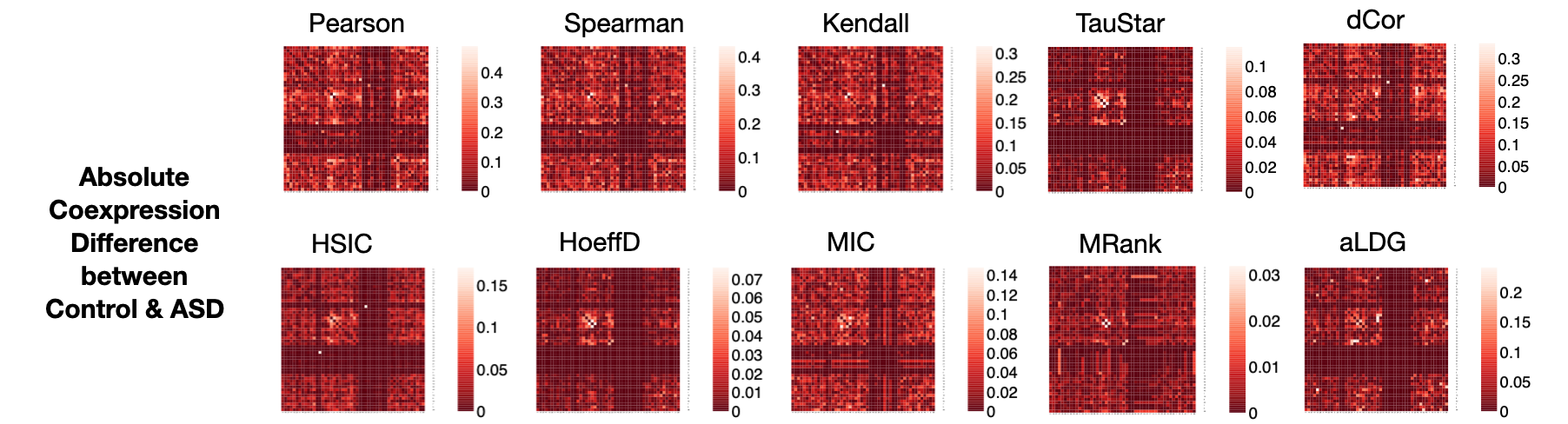}
    \caption{Estimated co-expression differences matrices (i.e. the absolute differences of the dependency matrices for control samples and ASD samples) obtained for different dependency measures.}
    \label{fig:vel23}
\end{figure}

\subsection{Simulation results}\label{sec:simu}
 In this section, we consider simulations that resembling  single-cell data to gain insights underlying the behavior of aLDG relative to the other methods. Specifically, we investigate scenarios where the bivariate relationship is (1) finite mixture; (2) linear or nonlinear; (3) monotone or non-monotone. See \figref{data} for all the synthetic data distributions we considered. We evaluate each dependence measure from the following perspective: (1) ability to capture complex relationship; (2) ability to accumulate subtle local dependence; (3) interpretation of strength of dependence in common sense; (4) power as an independence test; and (5) computation time. In the following, we focus on one perspective in each subsection, showing selective examples that inform our conclusions, relegating other examples to supplementary materials.

% {\color{blue} Summarize the data simulation process and show corresponding scatter plot? Motivate those choices? (Maybe put in appendix though)\\ KR:I think these details go to appendix.
% \begin{itemize}
%     \item Mixture relationship
%     \item Independent relationship
%     \item Linear relationship
%     \item non-linear relationship
%     \item non-monotone relationship
% \end{itemize}
% } 

\paragraph*%{aLDG detects non-linear, non-monotone relationships}
{Detecting nonlinear, non-monotone relationships} By construction, aLDG is expected to detect any non-negligible deviation from independence. Though many existing measures, such as HSIC, Hoeffding's D, dCor, $\tau^\star$, claim to be sensitive to nonlinear, non-monotone relationships, some approaches are known to perform poorly under certain circumstances.  By contrast, aLDG outperforms most of its competitors in the following standard evaluation experiment.  \figref{nonlinear} illustrates three points: (1) at independence, except for dCor, HHG, and MIC, most measures produce negligible values, as desired; (2) for linear and monotone relationship, all measures produce high values as expected; and (3) for nonlinear non-monotone relationships only aLDG, dCor, HHG and MIC produce high values consistently. In conclusion, only aLDG can effectively detect various types of dependency relationships while maintaining near-zero value at independence. dCor, HHG, and MIC are known to be sensitive to small, artificial deviations from independence, and these simulations reveal that they are indeed too sensitive as they often produce high values at independence.  A big portion of scRNA-seq data are collected over time; therefore, nonlinear, non-monotone and specifically oscillatory relationships are expected to happen. Therefore it is desirable to have a measure that is sensitive to dependence while remaining near zero of true independence, even under small perturbations.

\begin{figure}[H]
    \centering
    \includegraphics[width=\linewidth]{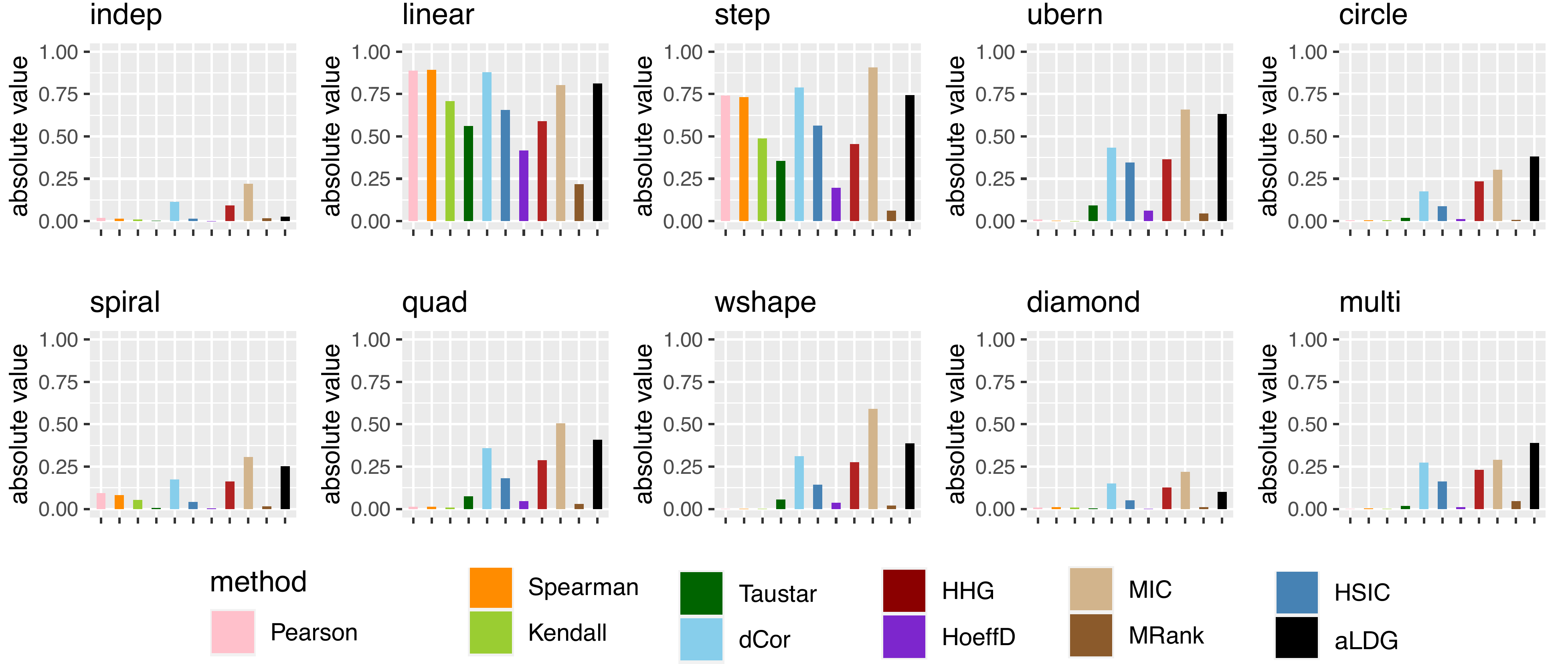}
    \caption{Empirical dependency estimates obtained for different data distributions for a variety of relationships between a pair of variables. For the visualization of different data distributions, see \figref{data}. Here we show the corresponding dependence level given by different measures using 200 samples (averaged over 50 trials).}
    \label{fig:nonlinear}
\end{figure}

\paragraph*%{aLDG accumulates subtle local dependencies} 
{Accumulating subtle local dependencies} aLDG detects the subset of the sample space that shows a pattern of dependence. In \figref{mix}, we simulated data as a bivariate Gaussian mixture consisting of three components with a varying proportion of highly dependent components and estimated the corresponding dependence level.  We find that aLDG, together with other dependence measures designed to capture local dependence (HHG and MIC)  increase with the proportion of highly correlated components, indicates that these global dependence measures can also detect subtle local dependence structure. Similar results are obtained for Negative Binomial mixtures \figref{nbmix}. As the finite mixture relationship is a common choice of model for scRNA-seq data, this suggests that measures able to accumulate dependencies across individual components could considerably benefit scRNA-seq data analysis. 

\begin{figure}[H]
    \centering
    \includegraphics[width=\linewidth]{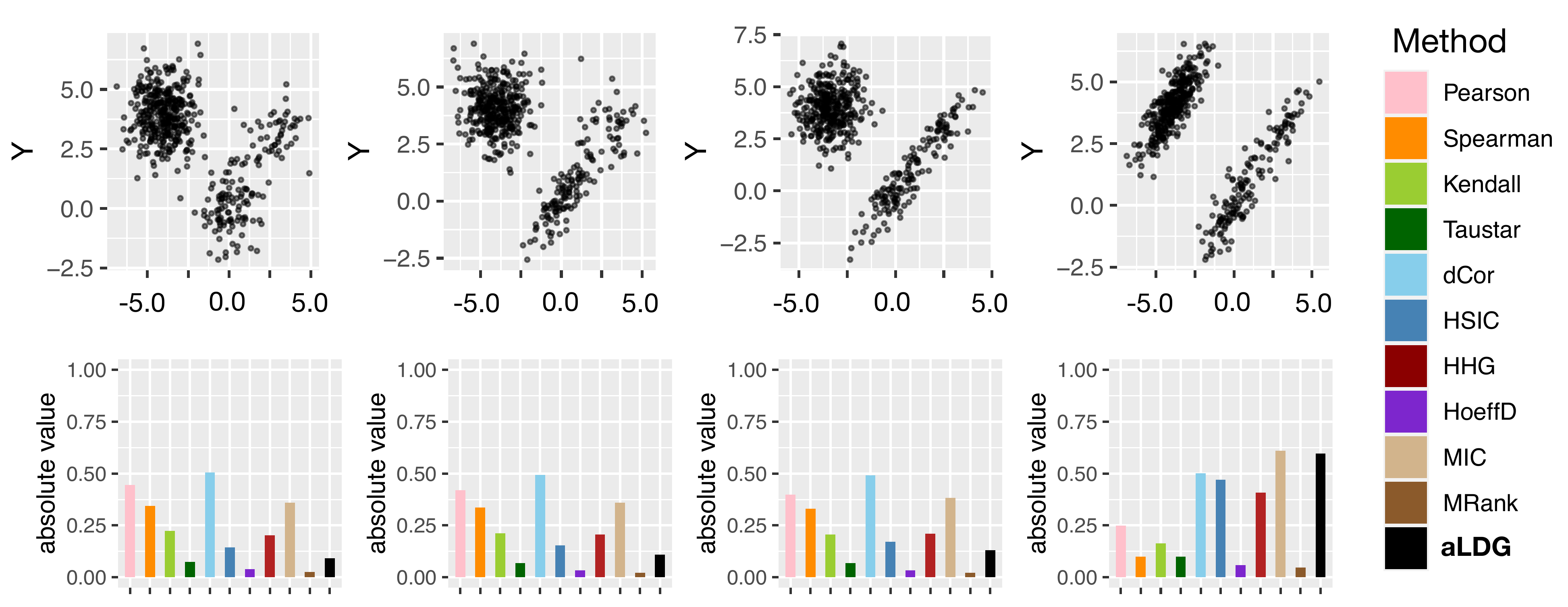}
    \caption{Empirical aLDG value for Gaussian mixtures. In each plot we show the dependence level given by different measures for 200 samples (averaged over 50 trials). The data are generated as a three-component Gaussian mixture. From left to right, there are 0, 1, 2 and 3 out of 3 components with correlation of 0.8, while the remaining components have correlation 0, i.e., the dependence level increases from left to right. For the visualization of these different data distributions, see \figref{data}.
    % {\color{red} I think the independence panel you show in Fig 8 is needed to understand Fig 6.  The issue is that dCor, HHG and MIC give pretty big measures, even under independence. When you compare them without any scaling, it is confusing that dCor looks so strong.  Alternatively, I suggest that you move the section going with Fig 8 to the first display of simulations.  The appeal is that those experiments are familiar and show dLDG works well in the expected setting.  Then you can refer back to the independence panel in the upcoming sections.  }
    }
    \label{fig:mix}
\end{figure}

\paragraph*%{aLDG interprets degree of dependencies} 
{Degree of dependencies} While it is hard to define the relative dependence level in general, we argue that when one random variable is a function of the other,  $Y=h(X)$, then the pair should be regarded as having the perfect dependence (and be assigned of dependence level $1$). Moreover, the dependence level should decrease as independent noise is added. That is, for $Y_\epsilon = h(X) + \epsilon$, where $\epsilon \perp X$, one should expect the dependence measure $\delta$ to satisfy  $\delta(Y_{\epsilon},X) < \delta(Y,X)$.  We checked this monotonicity property by simulating data with several bivariate relationships and varying levels of noise (\figref{mono}).   Specifically, we simulate the noise $\epsilon$ to be standard normal, and $Y = h(X) + c\epsilon$ where $c\in[0,1]$ indicates the noise level. We find that aLDG, HSIC, MIC, dCor, and HHG all show a clear decreasing pattern as the noise level increases; however, aLDG shows the most consistent monotonic drop from perfect dependence as the noise level increased.

\begin{figure}[H]
    \includegraphics[width=1\linewidth]{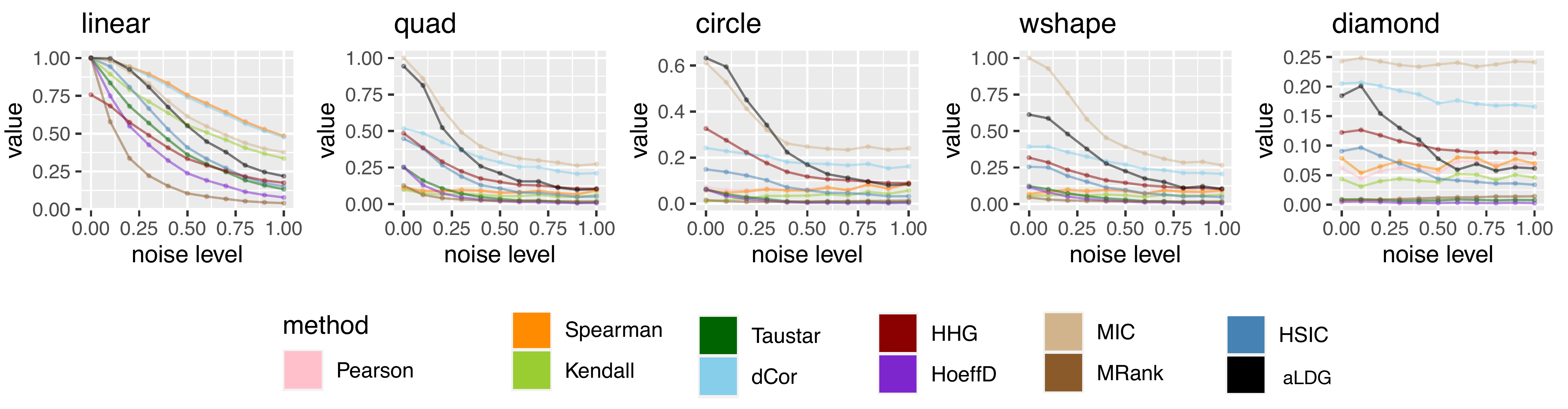}
    \caption{Empirical dependence measure versus noise levels for different bivariate relationships. For the visualization of different data distributions, see \figref{data}. The results are shown for 100 samples (averaged over 50 trials). We claim that the higher the noise level is, the lower the estimated degree of dependence should be. Compared with other measures, aLDG decreases significantly as the noise level increases, and hence correctly infers the relative degree of dependence. }
    \label{fig:mono}
\end{figure}

\paragraph*%{aLDG is powerful as an independence test} 
{Power as an independence test} Dependence measures are natural candidates for tests of independence. In this context, most existing dependence measures rely on bootstrapping or permutation to determine significance; hence we adopt this practice for all the dependence measures under comparison. \figref{non-linearpower} shows the empirical power under test level 0.05 for various types of data distribution and sample size, where we do 200 repetitions of permutations to estimate the null distribution. We observe the following outcomes: (1) almost all tests have controlled type-I error under independence; (2) Pearson's $\rho$, Spearman's $\rho_S$ and Kendall's $\tau$ are powerless for testing nonlinear and non-monotone relationships; (3) aLDG, HHG, and HSIC are consistently among the top three most powerful approaches for testing both linear and nonlinear, monotone and non-monotone relationships. Similar observations can be made for tests based on Gaussian mixtures \figref{gaussmixpower} and Negative Binomial mixtures \figref{nbmixpower}.

\begin{figure}[H]
    \centering
    \includegraphics[width=\linewidth]{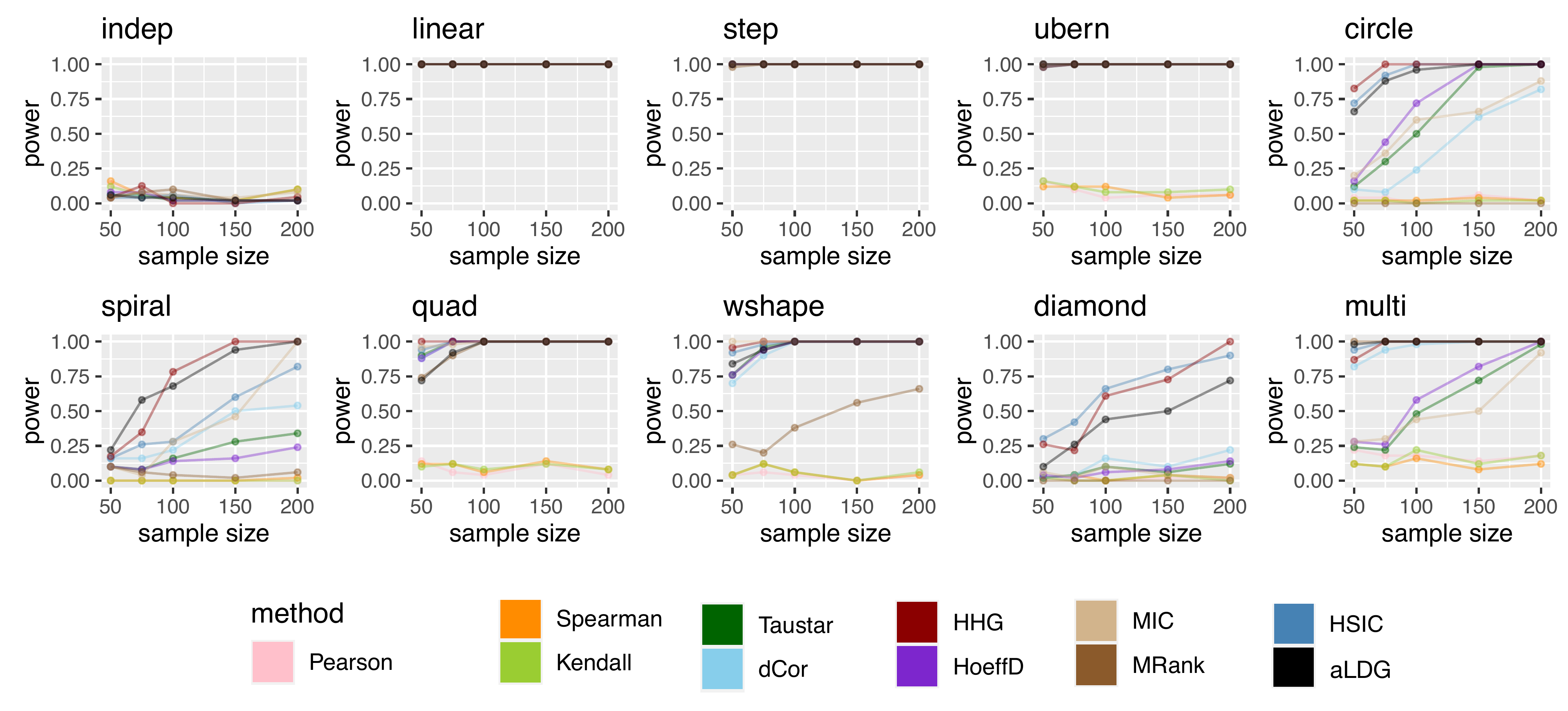}
    \caption{The empirical power of permutation test at level 0.05, based on different dependency measures under different data distributions and sample sizes. For the visualization of different data distributions, see \figref{data}. The power is estimated using 50 independent trials.}
    \label{fig:non-linearpower}
\end{figure}

\paragraph*%{aLDG computes fast} 
{Computational comparisons}
%We also compare the computation time of aLDG with all its competitors. 
Theoretically speaking, aLDG requires $O(n^2)$ in time of computation (where $n$ is the number of samples), which is comparable to reported requirements for most dependence measures that can detect complex relationships. This empirically confirmed in a comparison of the computation time of aLDG with all its competitors. In \figref{time} we plot the time of computation versus sample size $n$ for different dependence measures\footnote{The time include some constant wrapper function loading time, therefore, might be longer than a direct function call; however, the relative scale is still correct.}. In previous evaluations, we saw that HHG as a method motivated from capturing local dependence structure, was indeed a strong competitor to aLDG: it has high power as an independence test across almost all the data distribution we considered; however, it requires $O(n^3)$ time of computation, and \figref{time} shows this large discrepancy from all the other methods, which normally takes $O(n^2)$ time. 
% {\color{blue}Would table be better?\\ NO, the figure is great!}

\section{Conclusion and Discussion}

In this paper,  we formalize the idea of 
averaging the \emph{cell-specific gene association} \citep{dai2019cell,wang2021constructing} under a general statistical framework. We show that this approach produces a novel univariate dependence measure, called aLDG, that can detect nonlinear, non-monotone relationships between a pair of variables. We then develop the corresponding theoretical properties of this estimator, including robustness and consistency.  We also provide several hyper-parameter choices that are more justifiable and effective. Extensive simulations, motivated by expected scRNA-seq gene co-expression relationships and real data applications, show that this measure outperforms existing independence measures in various aspects: (1) it accumulates subtle local dependence over sub-populations; (2) it successfully interprets the relative strength of a monotonic function of dependence in the presence of noise better than many other measures that arose from independence test; (3) it is sensitive to complex relationships while robustly maintaining near-zero value at true independence, while several other measures are often overly sensitive to slight perturbations from independence and noise; (4) it computes comparatively rapidly compared to other dependence measures designed to capture complex relationships.  Other measures perform well in some settings but fail in others that are highly relevant to the single-cell setting. For instance, MIC performed well as part of the sLED test for differences in co-expression matrices, but this measure tends to produce a high estimate of dependence even when the variables are independent, or nearly so (Figure \ref{fig:nonlinear} and Figure \ref{fig:mono}). The moment-based methods like Pearson, dCor, and HSIC perform poorly when the expression values are sparse, producing false indications of correlation (Figure \ref{fig:realbi}), and yet sparsity is the norm in most single cell data. Our method is implemented in the R package aLDG\footnote{\url{https://github.com/JINJINT/aLDG}}, where we also include all the other methods that we have compared with.

The aLDG method does have some practical challenges: as a measure based on density estimation, the hyperparameter choices such as bandwidth can affect the performance of the measure. Though we provide some asymptotically optimal choices of those hyperparameters, in practice, they can fail due to the small sample size. For any given setting, the hyperparameters can be adjusted based on realistic simulations of the actual data and a solid understanding of the scRNA-seq data distribution. Similarly, due to the reliance on density estimation, it is hard to extend this measure to a multivariate setting. The sample size required for accurate estimation grows exponentially with the dimension. In practice, this limitation has little practical importance because gene co-expression studies focus on bivariate relationships.

%%%%%%%%%%%%%%%%%%%%%%%%%%%%%%%%%%%%%%%%%%%%%%
%% Support information, if any,             %%
%% should be provided in the                %%
%% Acknowledgements section.                %%
%%%%%%%%%%%%%%%%%%%%%%%%%%%%%%%%%%%%%%%%%%%%%%
\paragraph{Acknowledgments}
The authors would like to thank Xuran Wang for helpful comments.

%%%%%%%%%%%%%%%%%%%%%%%%%%%%%%%%%%%%%%%%%%%%%%
%% Funding information, if any,             %%
%% should be provided in the                %%
%% funding section.                         %%
%%%%%%%%%%%%%%%%%%%%%%%%%%%%%%%%%%%%%%%%%%%%%%
\paragraph{Funding}
This project is funded by National Institute of Mental Health (NIMH) grant R01MH123184 and NSF DMS-2015492.

\bibliographystyle{unsrtnat}
\bibliography{ref}

\begin{thebibliography}{32}
\providecommand{\natexlab}[1]{#1}
\providecommand{\url}[1]{\texttt{#1}}
\expandafter\ifx\csname urlstyle\endcsname\relax
  \providecommand{\doi}[1]{doi: #1}\else
  \providecommand{\doi}{doi: \begingroup \urlstyle{rm}\Url}\fi

\bibitem[Haque et~al.(2017)Haque, Engel, Teichmann, and
  L{\"o}nnberg]{Haque:2017}
Ashraful Haque, Jessica Engel, Sarah~A Teichmann, and Tapio L{\"o}nnberg.
\newblock A practical guide to single-cell rna-sequencing for biomedical
  research and clinical applications.
\newblock \emph{Genome Med}, 9\penalty0 (1):\penalty0 75, 08 2017.
\newblock \doi{10.1186/s13073-017-0467-4}.

\bibitem[Raj et~al.(2006)Raj, Peskin, Tranchina, Vargas, and
  Tyagi]{raj2006stochastic}
Arjun Raj, Charles~S Peskin, Daniel Tranchina, Diana~Y Vargas, and Sanjay
  Tyagi.
\newblock Stochastic mrna synthesis in mammalian cells.
\newblock \emph{PLoS biology}, 4\penalty0 (10), 2006.

\bibitem[Emmert-Streib et~al.(2014)Emmert-Streib, Dehmer, and
  Haibe-Kains]{Emmert-Streib:2014}
Frank Emmert-Streib, Matthias Dehmer, and Benjamin Haibe-Kains.
\newblock Gene regulatory networks and their applications: understanding
  biological and medical problems in terms of networks.
\newblock \emph{Front Cell Dev Biol}, 2:\penalty0 38, 2014.
\newblock \doi{10.3389/fcell.2014.00038}.

\bibitem[Dai et~al.(2019)Dai, Li, Zeng, and Chen]{dai2019cell}
Hao Dai, Lin Li, Tao Zeng, and Luonan Chen.
\newblock Cell-specific network constructed by single-cell rna sequencing data.
\newblock \emph{Nucleic acids research}, 47\penalty0 (11):\penalty0 e62--e62,
  2019.

\bibitem[Wang et~al.(2021)Wang, Choi, and Roeder]{wang2021constructing}
Xuran Wang, David Choi, and Kathryn Roeder.
\newblock Constructing local cell sepcific networks from single cell data.
\newblock \emph{bioRxiv}, 2021.

\bibitem[Eisen et~al.(1998)Eisen, Spellman, Brown, and
  Botstein]{eisen1998cluster}
Michael~B Eisen, Paul~T Spellman, Patrick~O Brown, and David Botstein.
\newblock Cluster analysis and display of genome-wide expression patterns.
\newblock \emph{Proceedings of the National Academy of Sciences}, 95\penalty0
  (25):\penalty0 14863--14868, 1998.

\bibitem[Bell(1962)]{bell1962mutual}
CB~Bell.
\newblock Mutual information and maximal correlation as measures of dependence.
\newblock \emph{The Annals of Mathematical Statistics}, pages 587--595, 1962.

\bibitem[Steuer et~al.(2002)Steuer, Kurths, Daub, Weise, and
  Selbig]{steuer2002mutual}
Ralf Steuer, J{\"u}rgen Kurths, Carsten~O Daub, Janko Weise, and Joachim
  Selbig.
\newblock The mutual information: detecting and evaluating dependencies between
  variables.
\newblock \emph{Bioinformatics}, 18\penalty0 (suppl\_2):\penalty0 S231--S240,
  2002.

\bibitem[Daub et~al.(2004)Daub, Steuer, Selbig, and Kloska]{daub2004estimating}
Carsten~O Daub, Ralf Steuer, Joachim Selbig, and Sebastian Kloska.
\newblock Estimating mutual information using b-spline functions--an improved
  similarity measure for analysing gene expression data.
\newblock \emph{BMC bioinformatics}, 5\penalty0 (1):\penalty0 1--12, 2004.

\bibitem[Reshef et~al.(2011)Reshef, Reshef, Finucane, Grossman, McVean,
  Turnbaugh, Lander, Mitzenmacher, and Sabeti]{reshef2011detecting}
David~N Reshef, Yakir~A Reshef, Hilary~K Finucane, Sharon~R Grossman, Gilean
  McVean, Peter~J Turnbaugh, Eric~S Lander, Michael Mitzenmacher, and Pardis~C
  Sabeti.
\newblock Detecting novel associations in large data sets.
\newblock \emph{science}, 334\penalty0 (6062):\penalty0 1518--1524, 2011.

\bibitem[Song et~al.(2012)Song, Langfelder, and Horvath]{song2012comparison}
Lin Song, Peter Langfelder, and Steve Horvath.
\newblock Comparison of co-expression measures: mutual information,
  correlation, and model based indices.
\newblock \emph{BMC bioinformatics}, 13\penalty0 (1):\penalty0 1--21, 2012.

\bibitem[Allen et~al.(2012)Allen, Xie, Chen, Girard, and
  Xiao]{allen2012comparing}
Jeffrey~D Allen, Yang Xie, Min Chen, Luc Girard, and Guanghua Xiao.
\newblock Comparing statistical methods for constructing large scale gene
  networks.
\newblock \emph{PloS one}, 7\penalty0 (1):\penalty0 e29348, 2012.

\bibitem[R{\'e}nyi(1959)]{renyi1959measures}
Alfr{\'e}d R{\'e}nyi.
\newblock On measures of dependence.
\newblock \emph{Acta Mathematica Academiae Scientiarum Hungarica}, 10\penalty0
  (3-4):\penalty0 441--451, 1959.

\bibitem[Hoeffding(1948)]{hoeffding1948non}
Wassily Hoeffding.
\newblock A non-parametric test of independence.
\newblock \emph{The annals of mathematical statistics}, pages 546--557, 1948.

\bibitem[Sz{\'e}kely et~al.(2007)Sz{\'e}kely, Rizzo, Bakirov,
  et~al.]{szekely2007measuring}
G{\'a}bor~J Sz{\'e}kely, Maria~L Rizzo, Nail~K Bakirov, et~al.
\newblock Measuring and testing dependence by correlation of distances.
\newblock \emph{The annals of statistics}, 35\penalty0 (6):\penalty0
  2769--2794, 2007.

\bibitem[Gretton et~al.(2005)Gretton, Bousquet, Smola, and
  Sch{\"o}lkopf]{gretton2005measuring}
Arthur Gretton, Olivier Bousquet, Alex Smola, and Bernhard Sch{\"o}lkopf.
\newblock Measuring statistical dependence with hilbert-schmidt norms.
\newblock In \emph{International conference on algorithmic learning theory},
  pages 63--77. Springer, 2005.

\bibitem[Bergsma and Dassios(2014)]{bergsma2014consistent}
Wicher Bergsma and Angelos Dassios.
\newblock A consistent test of independence based on a sign covariance related
  to kendall’s tau.
\newblock \emph{Bernoulli}, 20\penalty0 (2):\penalty0 1006--1028, 2014.

\bibitem[Wang et~al.(2014)Wang, Waterman, and Huang]{wang2014gene}
YX~Rachel Wang, Michael~S Waterman, and Haiyan Huang.
\newblock Gene coexpression measures in large heterogeneous samples using count
  statistics.
\newblock \emph{Proceedings of the National Academy of Sciences}, 111\penalty0
  (46):\penalty0 16371--16376, 2014.

\bibitem[Heller et~al.(2013)Heller, Heller, and Gorfine]{heller2013consistent}
Ruth Heller, Yair Heller, and Malka Gorfine.
\newblock A consistent multivariate test of association based on ranks of
  distances.
\newblock \emph{Biometrika}, 100\penalty0 (2):\penalty0 503--510, 2013.

\bibitem[Dhar et~al.(2016)Dhar, Dassios, Bergsma, et~al.]{dhar2016study}
Subhra~Sankar Dhar, Angelos Dassios, Wicher Bergsma, et~al.
\newblock A study of the power and robustness of a new test for independence
  against contiguous alternatives.
\newblock \emph{Electronic Journal of Statistics}, 10\penalty0 (1):\penalty0
  330--351, 2016.

\bibitem[Heller et~al.(2016)Heller, Heller, Kaufman, Brill, and
  Gorfine]{heller2016consistent}
Ruth Heller, Yair Heller, Shachar Kaufman, Barak Brill, and Malka Gorfine.
\newblock Consistent distribution-free k-sample and independence tests for
  univariate random variables.
\newblock \emph{The Journal of Machine Learning Research}, 17\penalty0
  (1):\penalty0 978--1031, 2016.

\bibitem[Gebelein(1941)]{gebelein1941statistische}
Hans Gebelein.
\newblock Das statistische problem der korrelation als variations-und
  eigenwertproblem und sein zusammenhang mit der ausgleichsrechnung.
\newblock \emph{ZAMM-Journal of Applied Mathematics and Mechanics/Zeitschrift
  f{\"u}r Angewandte Mathematik und Mechanik}, 21\penalty0 (6):\penalty0
  364--379, 1941.

\bibitem[Shen et~al.(2020)Shen, Priebe, and Vogelstein]{shen2020distance}
Cencheng Shen, Carey~E Priebe, and Joshua~T Vogelstein.
\newblock From distance correlation to multiscale graph correlation.
\newblock \emph{Journal of the American Statistical Association}, 115\penalty0
  (529):\penalty0 280--291, 2020.

\bibitem[Blum et~al.(1961)Blum, Kiefer, and Rosenblatt]{blum1961distribution}
Julius~R Blum, Jack Kiefer, and Murray Rosenblatt.
\newblock Distribution free tests of independence based on the sample
  distribution function.
\newblock \emph{The annals of mathematical statistics}, pages 485--498, 1961.

\bibitem[Simon and Tibshirani(2014)]{simon2014comment}
Noah Simon and Robert Tibshirani.
\newblock Comment on" detecting novel associations in large data sets" by
  reshef et al, science dec 16, 2011.
\newblock \emph{arXiv preprint arXiv:1401.7645}, 2014.

\bibitem[Gorfine et~al.(2012)Gorfine, Heller, and Heller]{gorfine2012comment}
Malka Gorfine, Ruth Heller, and Yair Heller.
\newblock Comment on detecting novel associations in large data sets.
\newblock \emph{Science}, pages 1--6, 2012.

\bibitem[Chu et~al.(2016)Chu, Leng, Zhang, Hou, Mamott, Vereide, Choi,
  Kendziorski, Stewart, and Thomson]{chu2016single}
Li-Fang Chu, Ning Leng, Jue Zhang, Zhonggang Hou, Daniel Mamott, David~T
  Vereide, Jeea Choi, Christina Kendziorski, Ron Stewart, and James~A Thomson.
\newblock Single-cell rna-seq reveals novel regulators of human embryonic stem
  cell differentiation to definitive endoderm.
\newblock \emph{Genome biology}, 17\penalty0 (1):\penalty0 1--20, 2016.

\bibitem[Velmeshev et~al.(2019)Velmeshev, Schirmer, Jung, Haeussler, Perez,
  Mayer, Bhaduri, Goyal, Rowitch, and Kriegstein]{velmeshev2019single}
Dmitry Velmeshev, Lucas Schirmer, Diane Jung, Maximilian Haeussler, Yonatan
  Perez, Simone Mayer, Aparna Bhaduri, Nitasha Goyal, David~H Rowitch, and
  Arnold~R Kriegstein.
\newblock Single-cell genomics identifies cell type--specific molecular changes
  in autism.
\newblock \emph{Science}, 364\penalty0 (6441):\penalty0 685--689, 2019.

\bibitem[Zhu et~al.(2017)Zhu, Lei, Devlin, and Roeder]{zhu2017testing}
Lingxue Zhu, Jing Lei, Bernie Devlin, and Kathryn Roeder.
\newblock Testing high-dimensional covariance matrices, with application to
  detecting schizophrenia risk genes.
\newblock \emph{The annals of applied statistics}, 11\penalty0 (3):\penalty0
  1810, 2017.

\bibitem[Trinh(2019)]{trinh2019volume}
Duc~Tai Trinh.
\newblock Volume of sublevel sets versus area of level sets via gelfand-leray
  form.
\newblock \emph{Acta Mathematica Vietnamica}, 44\penalty0 (4):\penalty0
  915--922, 2019.

\bibitem[Wainwright(2019)]{wainwright2019high}
Martin~J Wainwright.
\newblock \emph{High-dimensional statistics: A non-asymptotic viewpoint},
  volume~48.
\newblock Cambridge University Press, 2019.

\bibitem[Gin{\'e} and Guillou(2002)]{gine2002rates}
Evarist Gin{\'e} and Armelle Guillou.
\newblock Rates of strong uniform consistency for multivariate kernel density
  estimators.
\newblock In \emph{Annales de l'Institut Henri Poincare (B) Probability and
  Statistics}, volume~38, pages 907--921. Elsevier, 2002.

\end{thebibliography}
	
\newpage
\renewcommand\thefigure{S.\arabic{figure}}    
\setcounter{figure}{0} 

\renewcommand{\thesection}{\Alph{section}.\arabic{section}}
\setcounter{section}{0}
   
\begin{appendices}
\section{From avgCSN to aLDG}\label{app:derive}
Recall that we consider only a pair of random variables $X,Y$ whose joint and marginal densities exist and have the same support, and denote $f_{XY}, f_{X}, f_{Y}$ as their joint and marginal densities. Also, let $\widehat{f}_{XY}, \widehat{f}_{X}, \widehat{f}_{Y}$ be the estimated densities given observations of $(X,Y)$, and $\widehat{p}_{X,Y}(x,y)$ be the proportion of samples points in a square of side length $h$ centering at $(x, y)$, and $\widehat{p}_{X}$ and $\widehat{p}_{Y}$ be defined similarly for the marginal distribution.

First we point out that a reformulation of avgCSN statistics reveals its link to the population dependence measure we are going to introduce. Under our notation, the original avgCSN \citet{wang2021constructing} can be written as
\begin{align*}
    \text{avgCSN} & := \frac{1}{n} \sum_{i=1}^n \ones\left\{ \frac{\widehat{p}_{X,Y}(x_{i}, y_{i})  - \widehat{p}_{X}(x_{i}) \widehat{p}_{Y}(y_{i}) }{\sqrt{ \widehat{p}_{X}(x_{i})(1-\widehat{p}_{X}(x_{i}))\widehat{p}_{Y}(y_{i})(1-\widehat{p}_{Y}(y_{i}))}} \geq \frac{\Phi^{-1}(1-\alpha)}{\sqrt{n}}\right\},
\end{align*}
where $\Phi^{-1}$ is the quantile function of standard normal. When using a particular choice $\widehat{f}_{XY} = \widehat{p}_{X,Y}/h^2, \widehat{f}_{X} = \widehat{p}_{X}/h, \widehat{f}_{Y} = \widehat{p}_{Y}/h$, we have
\begin{align*}
    \text{avgCSN} = \frac{1}{n} \sum_{i=1}^n \ones\left\{ \frac{\widehat{f}_{X,Y}(x_{i}, y_{i}) h^2 - \widehat{f}_{X}(x_{i})h \widehat{f}_{Y}(y_{i})h }{\sqrt{ \widehat{f}_{X}(x_{i})h(1-\widehat{f}_{X}(x_{i})h)\widehat{f}_{Y}(y_{i})h(1-\widehat{f}_{Y}(y_{i})h)}} \geq \frac{\Phi^{-1}(1-\alpha)}{\sqrt{n}}\right\}. 
\end{align*}
Assuming the bandwidth $h\to 0$ and $ h\sqrt{n}\to\infty$, the expression can be approximated by the following
\begin{align*}
    \text{avgCSN} &\approx \frac{1}{n} \sum_{i=1}^n \ones\left\{ \frac{\widehat{f}_{X,Y}(x_{i}, y_{i}) - \widehat{f}_{X}(x_{i}) \widehat{f}_{Y}(y_{i}) }{\sqrt{ \widehat{f}_{X}(x_{i})\widehat{f}_{Y}(y_{i})}} \geq t_n \right\},\quad \text{where } t_n = \frac{\Phi^{-1}(1-\alpha)}{h\sqrt{n}}.
\end{align*}

\section{Proof for \thmref{aLDGrobpop}}\label{app:aLDGrobpop}
 \begin{proof}
 Denote the joint and marginal density of $F$ as $f_{X,Y}, f_{X}, f_{Y}$. Consider a fixed contamination position $(x',y')$, then we have the corresponding contaminated joint and marginal density as
 \begin{align*}
     & f^{(x')}_{X}(x):= 
     \begin{cases}
     (1-\epsilon)f_{X} (x), & if \ x \neq x',
     \\
     \infty , & if\ x = x';
     \end{cases},
     \quad 
      f^{(y')}_{Y}(y):= 
     \begin{cases}
     (1-\epsilon)f_{Y} (y), & if \ y \neq y',
     \\
     \infty , & if\ y = y';
     \end{cases}\\
      & \quad 
      f^{(x',y')}_{X,Y}(x,y):= 
     \begin{cases}
     (1-\epsilon)f_{X,Y} (x,y), & if \ (x,y) \neq (x',y'),
     \\
     \infty, & if\ (x,y) = (x',y').
     \end{cases}.
 \end{align*}
 
% Denote the set 
% $A'(t):=\{(x,y): f_{12}' - f_1' f_2' > t\}$, and $A(t):=\{(x,y): f_{12} - f_1 f_2 > t\}$, then we have 
% \begin{equation}
%     A(\frac{t}{1-\epsilon}) \subseteq A'(t) \subseteq  A(t-\epsilon f_{\text{max}}),
% \end{equation}
%  where $f_{\text{max}}:= ||f_1||_{\infty}\vee ||f_2||_{\infty}$. 
 
%  Then we have
%  \begin{align}
%      aLDG_{t}' & = \mathbb{P}_{f_{12}'}(A'(t)) = (1-\epsilon) \mathbb{P}_{f_{12}}(A'(t)) \\
%      & \in [(1-\epsilon)\mathbb{P}_{f_{12}}(A(\frac{t}{1-\epsilon})),\  (1-\epsilon)\mathbb{P}_{f_{12}}(A(t - \epsilon f_{\text{max}}))]\\
%      & \subseteq [(1-\epsilon)\text{aLDG}_{t}(\frac{t}{1-\epsilon}),\  (1-\epsilon)\text{aLDG}_{t}(t - \epsilon f_{\text{max}})]
%  \end{align}
\noindent 
Denote the density gap under original distribution $F$ as $\gap:= f_{X,Y} - f_{X} f_{Y}$, and the corresponding density gap under contaminated distribution as $\Delta_{\text{gap}}^{(x',y')}:= f_{X,Y}^{(x',y')} - f_{X}^{(x')} f_{Y}^{(y')}$, then
\begin{align*}
     \Delta_{\text{gap}}^{(x',y')}(x,y)=
     (1-\epsilon)\Big(\Delta_{\text{gap}}(x,y) + \epsilon f_{X}(x)f_{Y}(y)\Big)  \quad \text{if}\ x \neq x'\ \text{and}\ y \neq y',  
\end{align*}
 and the contaminated $\text{aLDG}_t$ statistics
 
 \begin{align}\label{aLDG1}
     & \text{aLDG}_t^{(x',y')}  = \text{Pr}_{F'}\left\{\Delta_{\text{gap}}^{(x',y')} > t \sqrt{f^{(x')}_{X}(x) f^{(y')}_{Y}(y)}\right\} \nonumber\\
     \leq &  \ \text{Pr}_{F'}\left\{(1-\epsilon)\Big(\Delta_{\text{gap}}(x,y) + \epsilon f_{X}(x)f_{Y}(y)\Big) > t(1-\epsilon)\sqrt{f_{X}(x) f_{Y}(y)},\ (x,y)\neq (x',y')\right\} \nonumber\\ 
     &+  \text{Pr}_{F'}\left\{(x,y)\neq (x',y')\right\} \nonumber\\
    = &\ (1-\epsilon) \text{Pr}_{F}\left\{\frac{\Delta_{\text{gap}}(x,y)}{\sqrt{f_{X}(x) f_{Y}(y)}}  + \epsilon \sqrt{f_{X}(x) f_{Y}(y)} > t\right\} +  \epsilon \nonumber\\
    \stackrel{(a)}{\leq}  & (1-\epsilon) \text{Pr}_{F}\left\{\frac{\Delta_{\text{gap}}(x,y)}{\sqrt{f_{X}(x) f_{Y}(y)}} + \epsilon f_{\text{max}} > t\right\} +  \epsilon
    =  (1-\epsilon) \text{aLDG}_{t-\epsilon f_{\text{max}}} + \epsilon
    \nonumber\\ \leq & (1-\epsilon)\big(\text{aLDG}_t + |\text{aLDG}_{t-\epsilon f_{\text{max}}} - \text{aLDG}_t|\big) + \epsilon  \nonumber\\  \stackrel{(b)}{\leq} & (1-\epsilon)\big(\text{aLDG}_t + L f_{\text{max}} \epsilon\big) + \epsilon \nonumber,
\end{align}
where (a) comes from the assumption that $f_{\max}:=||\sqrt{f_{X}f_{Y}}||_\infty < \infty$, and (b) comes from the assumption that $|\text{aLDG}_{t-\epsilon} - \text{aLDG}_{t}| \leq L \epsilon$ for all $\epsilon >0$. 
% On the other hand, we have
% \begin{align}\label{aLDG11}
%      & \text{aLDG}_t^{(x',y')}  = \text{Pr}_{F'}\left\{\Delta_{\text{gap}}^{(x',y')} > t \sqrt{f^{(x')}_{X}(x) f^{(y')}_{Y}(y)}\right\} \nonumber\\
%      \geq &  \ \text{Pr}_{F'}\left\{(1-\epsilon)\Big(\Delta_{\text{gap}}(x,y) + \epsilon f_{X}(x)f_{Y}(y)\Big) > t(1-\epsilon)\sqrt{f_{X}(x) f_{Y}(y)},\ (x,y)\neq (x',y')\right\} \nonumber\\
%     = &\ (1-\epsilon) \text{Pr}_{F}\left\{\frac{\Delta_{\text{gap}}(x,y)}{\sqrt{f_{X}(x) f_{Y}(y)}}  + \epsilon \sqrt{f_{X}(x) f_{Y}(y)} > t\right\}  \nonumber\\
%     \stackrel{(a)}{\geq}  & (1-\epsilon) \text{Pr}_{F}\left\{\frac{\Delta_{\text{gap}}(x,y)}{\sqrt{f_{X}(x) f_{Y}(y)}} > t\right\} 
%     =  (1-\epsilon) \text{aLDG}_{t} 
% \end{align}

% Combining \eqref{aLDG1} and \eqref{aLDG11}, we have
% \begin{align}
%     |\text{aLDG}_t^{(x',y')}-\text{aLDG}_t| & \leq \epsilon \cdot \max\left\{|\text{aLDG}_t+(1-\epsilon) L f_{max} -1|, \text{aLDG}_t\right\} \nonumber\\
%     &\leq \epsilon \cdot (L f_{max} + 1).
% \end{align}
\noindent
Therefore,
 \begin{align*}
     \text{IF}\left((x',y'), R_{\text{aLDG}_t}, F\right)&:=\lim_{\epsilon\to 0}\frac{\text{aLDG}_t^{(x',y')}-\text{aLDG}_t}{\epsilon} \\&\leq -\text{aLDG}_t+(1-\epsilon) L f_{max} +1\\&
     \leq L f_{max} +1
 \end{align*}
 Since the upper bound of IF does not depend on location of $(x',y')$, therefore,
 \begin{align*}
    \text{GES}(R_{\text{aLDG}_t}, F) \leq L f_{\text{max}} +1 < \infty.
\end{align*}
 \end{proof}

\section{Proof for \propref{indeprob}}\label{app:indeprob}
 \begin{proof}
 Denote the joint and marginal density of $F$ as $f_{X,Y}, f_{X}, f_{Y}$. Consider a fixed contamination point $(x',y')$ with mass $\epsilon$, then we have the corresponding contaminated joint and marginal density as
 \begin{align*}
     & f^{(x')}_{X}(x):= 
     \begin{cases}
     (1-\epsilon)f_{X} (x), & if \ x \neq x',
     \\
     \infty , & if\ x = x';
     \end{cases},
     \quad 
      f^{(y')}_{Y}(y):= 
     \begin{cases}
     (1-\epsilon)f_{Y} (y), & if \ y \neq y',
     \\
     \infty , & if\ y = y';
     \end{cases}\\
      & \quad 
      f^{(x',y')}_{X,Y}(x,y):= 
     \begin{cases}
     (1-\epsilon)f_{X,Y} (x,y), & if \ (x,y) \neq (x',y'),
     \\
     \infty, & if\ (x,y) = (x',y').
     \end{cases}.
 \end{align*}
Recall that the density gap $\gap:= f_{X,Y} - f_{X} f_{Y}$, and hence the  contaminated gap,
 \begin{align*}
     \Delta_{\text{gap}}^{(x',y')}(x,y) =
     (1-\epsilon)\Big(\Delta_{\text{gap}}(x,y) + \epsilon f_{X}(x)f_{Y}(y)\Big), & \quad \text{if}\ (x,y)\neq (x',y')
 \end{align*}
 and the contaminated aLDG statistics
 \begin{align*}
     & \text{aLDG}_0^{(x',y')}  = \text{Pr}_{F'}\{\Delta_{\text{gap}}^{(x',y')} > 0\} \nonumber\\
     \leq &  \ \text{Pr}_{F'}\left\{(1-\epsilon)\Big(\Delta_{\text{gap}}(x,y) + \epsilon f_{X}(x)f_{Y}(y)\Big) > 0, (x,y)\neq (x',y')\right\} +  \text{Pr}_{F'}\left\{(x,y)\neq (x',y')\right\} \nonumber\\
    = &\ (1-\epsilon) \text{Pr}_{F}\left\{\Delta_{\text{gap}}(x,y) + \epsilon f_{X}(x)f_{Y}(y) > 0\right\} +  \epsilon. \nonumber
    \end{align*}
Note that
\begin{align*}
   & \text{Pr}_{F}\left\{\Delta_{\text{gap}}(x,y) + \epsilon f_{X}(x)f_{Y}(y) > 0\right\}\nonumber\\
   = & \text{Pr}_{F}\left\{\Delta_{\text{gap}}(x,y)  > 0 \right\} + \text{Pr}\left\{- \epsilon f_{X}(x)f_{Y}(y) < \gap(x,y) \leq 0\right\} \nonumber\\
   = &  \text{aLDG}_0 + \text{Pr}_{F}\left\{ 1- \epsilon  < \frac{f_{X,Y}(x,y)}{f_{X}(x)f_{Y}(y)} \leq 1\right\} \nonumber\\
   = &\text{aLDG}_0 + \text{Pr}_{F}\{ 1- \epsilon  < c_F(u,v) \leq 1\}, \nonumber
 \end{align*}
 where $c_F(u,v)$ is the joint density of $u := F_{X}^{-1}(x), v:= F_{Y}^{-1}(x)$, i.e. the corresponding copula representation of distribution $F$. Then, denoting the volume of set $\Gamma_t: \{(u,v,t): c_F(u,v) \leq t\}$ as $\text{Vol}(t)$, and the area of sublevel set $\gamma_t: \{(u,v): c_F(u,v) \leq t\}$ as $\text{A}(t)$, and the contour line $\mathcal{C}(t) : = \{(u,v): c_F(u,v) = t\}$, we have
 \begin{align*}
     & \lim_{\epsilon \to 0} \frac{1}{\epsilon}\PP{ 1- \epsilon  < c_F(u,v) \leq 1} = \lim_{\epsilon \to 0} \frac{1}{\epsilon}\int_{1-\epsilon < c_F(u,v) \leq 1} c_F(u,v) du dv \nonumber\\
     = & \lim_{\epsilon \to 0} \frac{\textnormal{Vol}(1) -\textnormal{Vol}(1-\epsilon)}{\epsilon} =\frac{d\textnormal{Vol}}{d t}\mid_{t=1} \nonumber\\ & \stackrel{(a)}{=} \frac{A(1)}{||\nabla c_F(u_0,v_0)||_2} \stackrel{(b)}{\leq} \frac{1}{||\nabla c_F(u_0,v_0)||_2} \nonumber
\end{align*}
where $(u_0,v_0)$ is some point on $\mathcal{C}_t$ and $\nabla c_F(u_0,v_0)$ is the gradient of $c_F$ at $(u_0,v_0)$, and (a) comes from Theorem 1 in \cite{trinh2019volume} using the a.e. smoothness of the joint and marginal densities $f_{XY}$, $f_{X}$, $f_{Y}$; (b) uses the trivial bound $A(1) \leq 1$ since we are working on $[0,1]^2$ space.

Plug the above calculation back to IF function, we get
\begin{align*}
    \text{IF}\Big((x',y'), R_{\text{aLDG}_0}, F\Big) & = \frac{(1-\epsilon) \Big(\text{aLDG}_0 + \text{Vol}(1)-\text{Vol}(1-\epsilon)\Big) + \epsilon}{\epsilon} \\
    &= 1- \text{aLDG}_0 - \text{Vol}(1) + \lim_{\epsilon \to 0} \text{Vol}(1-\epsilon) + \lim_{\epsilon \to 0} \frac{1}{\epsilon} \left(\text{Vol}(1)-\text{Vol}(1-\epsilon)\right) \\
    & \leq 1- \text{aLDG}_0 + \frac{1}{||\nabla c_F(u_0,v_0)||_2},
\end{align*}
where $(u_0,v_0)$ is some point on the contour line $\mathcal{C}_t:=\{(u,v): c_F(u,v) = t\}$, and $\nabla c_F(u_0,v_0)$ is the gradient of $c_F$ at $(u_0,v_0)$.

Note that this upper bound is irrelevant with $(x',y')$, therefore we have
\begin{equation*}
    \text{GES}(R_{\text{aLDG}}, F) \leq 1- \text{aLDG}_0(F) + \frac{1}{||\nabla c_F(u_0,v_0)||_2} < \infty,
\end{equation*}
as long as $X,Y$ is not independent.
\\
\\
However, when $X,Y$ are independent, we have $c_F(u,v)\equiv1$ for all $(u,v)\in[0,1]^2$, and $\text{aLDG}_0=0$, then we have 
\begin{align*}
    \text{aLDG}_0^{(x',y')} &\geq \text{Pr}_{F'}\{\Delta_{\text{gap}}(x,y) + \epsilon f_{X}(x)f_{Y}(y) > 0, (x,y)\neq (x',y')\} \nonumber\\
    &= (1-\epsilon)\Big(\text{aLDG}_0 + \PP{1-\epsilon < c_F(u,v)\leq 1}\Big) \nonumber\\
    & =  (1-\epsilon)( 0 + 1) = 1-\epsilon,
\end{align*}
and hence
\begin{equation*}
    \text{IF}\Big((x',y'), R_{\text{aLDG}_0}, F\Big) \geq \lim_{\epsilon\to 0} \frac{1-\epsilon}{\epsilon} = \infty.
\end{equation*}
Again this lower bound is irrelevant with $(x',y')$, therefore we have $\text{GES}(R_{\text{aLDG}_0}, F)=\infty$.
\end{proof}

\section{Proof for \thmref{aLDGconsist}}\label{app:pfconsist}
\begin{proof}
Denote the set
\begin{equation*}
    S_t :=\left \{(x,y): \frac{f_{XY}(x,y)-f_X(x) f_Y(y)}{\sqrt{f_{X}(x)f_{Y}(y)}} > t\right\}, \quad \widehat{S}_t := \left\{(x,y): \frac{\widehat{f}_{XY}(x,y) -\widehat{f}_{X}(x)\widehat{f}_{Y}(y) }{\sqrt{\widehat{f}_{X}(x)\widehat{f}_{Y}(y)}} > t\right\}.
\end{equation*}
From the assumption that $||\widehat{f}_{XY}-f_{XY}||_{\infty}, ||\widehat{f}_{X}-f_{X}||_{\infty}, ||\widehat{f}_{Y}-f_{Y}||_{\infty} \leq \eta_n$ with probability at least $1-\frac{1}{n}$, we have the following holds for some constant $c > 0$ with probability at least $1-\frac{1}{n}$:
\begin{equation*}
    \sup_{x,y} \left| \frac{f_{XY}-f_X f_Y}{\sqrt{f_{X}f_{Y}}} - \frac{\widehat{f}_{XY} -\widehat{f}_{X}\widehat{f}_{Y} }{\sqrt{\widehat{f}_{X}\widehat{f}_{Y}}}\right| \leq  \frac{  (3c_{\max}+1)\eta_n}{c_{\min}^{\frac12}} + \frac{(3c_{\max}+1)\eta_n^2 + 2c_{\max}\eta_n}{c_{\min}^{\frac32}}< C \eta_n,
\end{equation*}
where $C := \frac{  (3c_{\max}+1)}{c_{\min}^{\frac12}} + \frac{(3c_{\max}+1) + 2c_{\max}}{c_{\min}^{\frac32}}$
and correspondingly
\begin{equation*}
    S_{t + C\eta_n} \subseteq \widehat{S}_{t} \subseteq S_{t - C\eta_n}.
\end{equation*}
As a result, applying the empirical measure $\widehat{P}(S):=\frac{1}{n}\sum_{i}\ones\{(x_i,y_i) \in S\}$ on these three sets, we get 
\begin{equation}\label{initbound}
    \widehat{P}(S_{t + C\eta_n}) \leq \widehat{P}(\widehat{S}_{t}) = \widehat{\text{aLDG}}(t) \leq  \widehat{P}(S_{t - C\eta_n}).
\end{equation}
Using the Hoeffding's inequality on binomials, we get 
\begin{equation*}
  |\widehat{P}(S) -  P(S)| < \sqrt{\frac{2\log{n}}{n}}
\end{equation*}
with probability at least $1-\frac{1}{2n}$ for any deterministic set $S$. Applying this inequality to $\widehat{P}(S_{t + C\eta_n})$ and $\widehat{P}(S_{t - c\eta_n})$ in \eqref{initbound}, we get
\begin{equation*}
    P(S_{t + C\eta_n}) - \sqrt{\frac{2\log{n}}{n}}\leq \widehat{P}(\widehat{S}_{t}) \leq  P(S_{t - C\eta_n}) + \sqrt{\frac{2\log{n}}{n}}
\end{equation*}
with probability at least $1-\frac{2}{n}$. This further implies that
\begin{equation*}
    \text{aLDG}_{t + C\eta_n} - \sqrt{\frac{2\log{n}}{n}} \leq \widehat{\text{aLDG}}(t) \leq \text{aLDG}_{t - C \eta_n} + \sqrt{\frac{2\log{n}}{n}}
\end{equation*}
with probability at least $1-\frac{2}{n}$.
With the condition that $|\text{aLDG}_{t-\epsilon}-\text{aLDG}_t| \leq L\epsilon$ for all $\epsilon>0$, we have 
\begin{equation*}
    \text{aLDG}_t -  LC\eta_n - \sqrt{\frac{2\log{n}}{n}} \leq \widehat{\text{aLDG}}_t \leq \text{aLDG}_t +  LC\eta_n + \sqrt{\frac{2\log{n}}{n}},
\end{equation*}
that is 
\begin{equation*}
    \left| \widehat{\text{aLDG}}_t -  \text{aLDG}_t \right| \leq LC\eta_n + \sqrt{\frac{2\log{n}}{n}}
\end{equation*}
with probability at least $1-\frac{2}{n}$.

\end{proof}

\section{A uniform variant of consistency}\label{app:uniform}
\begin{theorem}\label{thm:aLDGuniform}
Consider a bivariate distribution $F$ of variable $(X,Y)$ whose joint and marginal densities exist as $f_{XY}$, $f_{X}$, $f_{Y}$, and satisfy
\begin{align*}
&\inf_{x,y}f_{XY}(x,y),\  \inf_{x}f_X(x) \inf_{y}f_Y(y) \geq c_{\min},\\
& \sup_{x,y}f_{XY}(x,y),\  \sup_{x}f_X(x) \sup_{y}f_Y(y) \leq c_{\max},
\end{align*}
and for some $\eta_n$ with $\lim_{n\to\infty}\eta_n \to 0$, with probability at least $1-\frac{1}{n}$ 
\begin{equation*}
  ||\widehat{f}_{XY}-f_{XY}||_{\infty}, ||\widehat{f}_{X}-f_{X}||_{\infty}, ||\widehat{f}_{Y}-f_{Y}||_{\infty} \leq \eta_n; 
\end{equation*} 
and for some constant $0<L<\infty$, 
\begin{equation*}
    |\text{aLDG}_{t-\epsilon}-\text{aLDG}_t| \leq L\epsilon\quad  \text{for all} \ \epsilon>0, \quad \text{for all } t\geq0.
\end{equation*}
Then we have, with probability at least $1-\frac{2}{n}$, we have
\begin{equation*}
    \sup_{t\geq 0}\left|\widehat{\text{aLDG}}_t - \text{aLDG}_t\right| \leq  LC\eta_n + 10\sqrt{\frac{\log{n}}{n}},
\end{equation*}
where $C$ depends only on $c_{\min}, c_{\max}$.
\end{theorem}
\begin{proof}
Recall the bivariate functional 
\begin{equation*}
    T: (x,y)\mapsto \frac{f_{XY}(x,y)-f_X(x) f_Y(y)}{\sqrt{f_{X}(x)f_{Y}(y)}}, \quad \widehat{T}: (x,y)\mapsto \frac{\widehat{f}_{XY}(x,y)-\widehat{f}_X(x) \widehat{f}_Y(y)}{\sqrt{\widehat{f}_{X}(x)\widehat{f}_{Y}(y)}}.
\end{equation*} 
Correspondingly, for a $t\geq 0$, denote the set
\begin{equation*}
    S_t :=\left \{(x,y): T(x,y) > t\right\}, \quad \widehat{S}_t := \left\{(x,y): \widehat{T}(x,y)> t\right\}.
\end{equation*}
We also denote the collection of such set over all $t\geq 0$ as $\mathcal{S} = \{S_t: t \geq 0\}$. 

From proposition 4.20 \citep{wainwright2019high}, it is easy to see that the class $\mathcal{S}$ has VC dimension at most $1$, since it can be written as the subgraph class of the function class $\{g_t: (x,y) \mapsto t-T(x,y); t\geq 0\}$ is a vector space of $dim(1)$ (as function $T$ is deterministic and only $t$ is changing). Using VC theorem, we get 
\begin{equation*}
  \sup_{S\in\mathcal{S}}|\widehat{P}_n(S) -  P(S)| \leq \sqrt{\frac{32}{n}\left(\log(n+1) + \log(16n)\right)} \leq 10\sqrt{\frac{\log{n}}{n}}
\end{equation*}
with probability at least $1-\frac{1}{2n}$, where $\widehat{P}(S):=\frac{1}{n}\sum_{i}\ones\{(x_i,y_i) \in S\}$ is the empirical measure.

From the assumption that $||\widehat{f}_{XY}-f_{XY}||_{\infty}, ||\widehat{f}_{X}-f_{X}||_{\infty}, ||\widehat{f}_{Y}-f_{Y}||_{\infty} \leq \eta_n$ with probability at least $1-\frac{1}{n}$, we have the following holds for some constant $c > 0$ with probability at least $1-\frac{1}{n}$:
\begin{equation}\label{vc}
    \sup_{x,y} \left| \frac{f_{XY}-f_X f_Y}{\sqrt{f_{X}f_{Y}}} - \frac{\widehat{f}_{XY} -\widehat{f}_{X}\widehat{f}_{Y} }{\sqrt{\widehat{f}_{X}\widehat{f}_{Y}}}\right| \leq  \frac{  (3c_{\max}+1)\eta_n}{c_{\min}^{\frac12}} + \frac{(3c_{\max}+1)\eta_n^2 + 2c_{\max}\eta_n}{c_{\min}^{\frac32}}< C \eta_n,
\end{equation}
where $C := \frac{  (3c_{\max}+1)}{c_{\min}^{\frac12}} + \frac{(3c_{\max}+1) + 2c_{\max}}{c_{\min}^{\frac32}}$
and correspondingly
\begin{equation*}
    S_{t + C\eta_n} \subseteq \widehat{S}_{t} \subseteq S_{t - C\eta_n} \quad \text{for all } t\geq 0.
\end{equation*}
As a result, applying the empirical measure $\widehat{P}(S)$ on these three sets, we get 
\begin{equation}\label{initbound2}
    \widehat{P}(S_{t + C\eta_n}) \leq \widehat{P}(\widehat{S}_{t}) = \widehat{\text{aLDG}}(t) \leq  \widehat{P}(S_{t - C\eta_n}) \quad \text{for all } t\geq 0.
\end{equation}
 Applying \eqref{vc} to $\widehat{P}(S_{t + C\eta_n})$ and $\widehat{P}(S_{t - c\eta_n})$ in \eqref{initbound2}, we get
\begin{equation*}
    P(S_{t + C\eta_n}) - 10\sqrt{\frac{\log{n}}{n}}\leq \widehat{P}(\widehat{S}_{t}) \leq  P(S_{t - C\eta_n}) + 10\sqrt{\frac{\log{n}}{n}} \quad \text{for all } t\geq 0
\end{equation*}
with probability at least $1-\frac{2}{n}$. This further implies that
\begin{equation*}
    \text{aLDG}_{t + C\eta_n} - 10\sqrt{\frac{\log{n}}{n}} \leq \widehat{\text{aLDG}}(t) \leq \text{aLDG}_{t - C \eta_n} + 10\sqrt{\frac{\log{n}}{n}} \quad \text{for all } t\geq 0
\end{equation*}
with probability at least $1-\frac{2}{n}$.
With the condition that $|\text{aLDG}_{t-\epsilon}-\text{aLDG}_t| \leq L\epsilon$ for all $\epsilon>0$ and $t\geq 0$, we have 
\begin{equation*}
    \text{aLDG}_t -  LC\eta_n - 10\sqrt{\frac{\log{n}}{n}} \leq \widehat{\text{aLDG}}_t \leq \text{aLDG}_t +  LC\eta_n + 10\sqrt{\frac{\log{n}}{n}}, \quad \text{for all } t\geq 0
\end{equation*}
that is 
\begin{equation*}
    \sup_{t\geq 0}\left| \widehat{\text{aLDG}}_t -  \text{aLDG}_t \right| \leq LC\eta_n + 10\sqrt{\frac{\log{n}}{n}}
\end{equation*}
with probability at least $1-\frac{2}{n}$.

\end{proof}

\section{Uniform estimation error of product kernel density estimator }\label{app:densest}
\begin{definition}
Let $\beta$ be a positive integer, we define $G(\beta)$ as the class of one-dimensional kernel function $K$, in which $K$ has support $[-1,1]$, and $\int K = 1$, $\int |K|^p < \infty$ for any $p\geq 1$, $\int |t|^\beta K(t) d  t < \infty$ and $\int t^s K(t) d t = 0$ for any $1 \leq s \leq \beta$. 
\end{definition}

\begin{definition}
Let $\beta$ be a positive integer, $L$ be a positive constant, we define $H(\beta,L)$ as the class of one-dimensional density $k$, such that
\begin{equation*}
    \left|\frac{d^{\beta-1} k(x)}{x^{\beta-1}} - \frac{d^{\beta-1} k(y)}{x^{\beta-1}}\right| \leq L|x-y|,\quad \text{for all } x,y
\end{equation*}
\end{definition}

In the following we analyse a special class of multivariate density function together with a special class of density estimator. Specifically, for positive integer $\beta$, consider density function $k \in H(\beta, L)$, and kernel function $K \in G(\beta)$. For dimension $d\geq 1$, we consider the following multivariate density function in $\mathbb{R}^d$:
\begin{equation}\label{mixdens}
\mathbf{k}_{\alpha,\mu,r}(\bm{x}) := \prod_{i=1}^d k_{\alpha,\mu,r}(\cdot)(x_i),  \quad \text{where } k_{\alpha, \mu,r}(\cdot) = (1-\alpha) k\left(\cdot \right)+ \alpha\frac{1}{r}k\left(\frac{\cdot-\mu}{r}\right), 
\end{equation}
with $\alpha\in [0,1]$ as the mixture proportion, $\mu\geq0$ the relative location, and $r >0$ as the relative scale; we also consider the following multivariate kernel function
\begin{equation*}
    \mathbf{K}_h(\bm{x}) := \prod_{i=1}^d K_h(x_i),\quad \text{where } K_h(\cdot) := \frac{1}{h}K(\frac{\cdot}{h}),
\end{equation*}
with $h>0 \in \mathbb{R}$; and the corresponding empirical kernel density estimator
\begin{equation}\label{mixkernel}
    \widehat{\bm{K}}_h(\cdot) := \frac{1}{n}\sum_{i=1}^n \bm{K}_h(\bm{X}_i-\cdot),
\end{equation}
given $n$ observations $\bm{X}_1,\dots ,\bm{X}_n$ in $\mathbb{R}^d$.

\begin{proposition}\label{prop:densest}
Consider $k_{\alpha,\mu,r}$ in \eqref{mixdens} and $\widehat{K}_h$ in \eqref{mixkernel}. Then for any $\delta > 0$, we have
\begin{align*}
    \PP{\sup_{\bm{x}\in \mathbb{R}^d}|\widehat{\bm{K}}_h(\bm{x}) - \bm{k}_{\alpha, \mu, r}(\bm{x})|> \sqrt{\frac{C \log{(1/\delta)} (1-\alpha + \frac{\alpha}{r})^d}{n h^d}} + c \left(1-\alpha + \frac{\alpha}{r^{\beta+1}}\right)^d h^{d\beta} } < \delta,
\end{align*}
where $C$ and $c$ are positive constants which do not depend on $h,\alpha,\mu,r$. Particularly,
choosing adaptively 
\begin{equation*}
    h = \left(\frac{C\log{\frac{1}{\delta}}(1-\alpha+\frac{\alpha}{r})^{d}}{c^2 n(1-\alpha+\frac{\alpha}{r^{\beta+1}})^{2d}}\right)^{\frac{1}{(2\beta +1) d}},
\end{equation*}
 we have
\begin{align*}
     \PP{\sup_{\bm{x}\in \mathbb{R}^d}|\widehat{\mathbf{K}}_h(\bm{x}) - \mathbf{k}_{\mu,r}(\bm{x})| > 2c\left(\frac{C\log{\frac{1}{\delta}}}{c^2 n} \right)^{\frac{\beta}{2\beta+1}} \left(1-\alpha + \frac{\alpha}{r^{\beta+1}}\right)^{\frac{\beta+1}{2\beta+1}d}}  < \delta.
\end{align*}
\end{proposition}

\begin{remark}
Back to the example in the main paper, the joint density for $X,Y$ we considered is in fact $f_{X,Y}(x,y) = k(x)k(y)$ with $k\in H(1,L)$. And the density estimator we considered is in fact $\widehat{K}_{h}$ in \eqref{mixkernel} with the one-dimensional kernel function $K$ as boxcar kernel smoothing function (which obviously belongs to $G(1)$). Then use \propref{densest} with $\beta=1,\alpha=0, d=2$, we have with probability at least $1-1/n$, 
\begin{equation*}
    ||f_{XY} - \widehat{f}_{XY}||_{\infty} \leq O(n^{-\frac13}\sqrt{\log{n}}).
\end{equation*}
Similarly, for the marginal densities, we have that, with bandwidth $h_n = O(n^{-1/6})$,
\begin{equation*}
    ||f_{X} - \widehat{f}_{X}||_{\infty} \leq O(n^{-\frac16}\sqrt{\log{n}}), \quad ||f_{Y} - \widehat{f}_{Y}||_{\infty} \leq O(n^{-\frac16}\sqrt{\log{n}}).
\end{equation*}
Finally, recall the definition of error rate $\eta_n$, we have
\begin{align*}
    \eta_n:= \sup\{||f_{XY} - \widehat{f}_{XY}||_{\infty}, ||f_{X} - \widehat{f}_{X}||_{\infty}, ||f_{Y} - \widehat{f}_{Y}||_{\infty}\} \leq O(n^{-\frac16}\sqrt{\log{n}})
\end{align*}
with probability at least $1-1/n$.
\end{remark}

% \begin{proposition}(\citet{rinaldo2010generalized} Proposition 9 )\label{prop:biasest}
% Assume that the kernel satisfies condition (VC) and that
% \begin{equation}\label{varcond}
%     \sup_{t \in \mathcal{R}^d} \sup_{h>0} \int_{\mathcal{R}^d} K^2_h(t-x)d P(x) < D < \infty.
% \end{equation}
% \begin{itemize}
%     \item[(a)] Let $h$ be fixed. Then, there exist constants $L > 0$ and $C >0$, which depend only on the VC characteristics of $K$, such that, for any $c_1 \geq C$ and   $0 < \epsilon \leq \frac{c_1 D}{\norm{K}_\infty}$, there exists an $n_0 > 0$, which depends on $\epsilon, D, \mathcal{K}_{\infty}$ and the $VC$ characteristics of $K$, such that, for all $n>n_0$, 
%     \begin{equation}\label{densbound}
%         \PP{\sup_{x\in\mathbb{R}^d} |\widehat{p}_h(x) - p_h(x)| > 2 \epsilon} \leq L \exp \left\{-\frac{1}{L}\frac{\log\{1 + \frac{c_1}{4L}\}}{c_1}\frac{n h^d \epsilon^2}{D} \right\}.
%     \end{equation}
    
%     \item[(b)] Let $h_n \to 0$ as $n \to \infty$ in such a way that $\frac{n h_n^d}{\log{h_n^d}} \to \infty$. if $\{\epsilon_n\}$ is a sequence such that 
%     \begin{equation}
%         \epsilon = \Omega\left( \frac{\log{r_n}}{n h_n^d} \right)
%     \end{equation}
%     where $r_n = \Omega(h_n^{\frac{d}{2}})$, then, for all $n$ large enough, \eqref{densbound} holds with $h$ and $\epsilon$ replaced by $h_n$ and $\epsilon_n$, respectively. In particularly, the term on the right hand side of \eqref{densbound} vanishes at the rate $O(r_n^{-1})$. 
% \end{itemize}
% \end{proposition}
\begin{proof}
We can decompose the deviation as the following:
\begin{align}\label{infdecomp}
    \norm{\widehat{\bm{K}}_h - \bm{k}_{\alpha,\mu,r}}_{\infty} \leq \norm{\widehat{\bm{K}}_h - \EE{\widehat{\bm{K}}_h}}_{\infty} + \norm{\EE{\widehat{\bm{K}}_h} - \bm{k}_{\alpha,\mu,r}}_{\infty},
\end{align}
where the expectation in $\EE{\widehat{\bm{K}}_h}$ is taken over given samples $X_1,\dots,X_n$. In the following, we bound each term separately, throughout which we denote expressions that do not depend on $h,\alpha,r,\mu$ as constants terms.

\begin{itemize}
    \item[\textbf{Step 1.}]
   
   To bound the first term in \eqref{infdecomp}, we use  Corollary 2.2 in  \citet{gine2002rates}. Firstly we introduce the required condition. 
% $\bm{K}_h$, that is $\bm{K}_h$ belongs to VC class \defref{vc}. This requirement is satisfied for a large class of kernel smoothing functions, including, for example, any compact supported polynomial kernel and the Gaussian kernel. Therefore the kernel we considered satisfy this condition. 
\begin{definition}\label{def:vc}(VC class)
Let $\mathcal{F}$ be a uniformly bounded collection of measurable functions on $\mathbb{R}^d$. We say that $\mathcal{F}$ is a bounded measurable VC class of functions if the class $\mathcal{F}$ is separable and if there exist positive numbers $A$ and $v$ such that, for every probability measure $P$ on $\mathbb{R}^d$ and every $0 < \epsilon < 1$, 
\begin{equation}\label{vccond}
    \sup_{P} N(\mathcal{F}, L_2(P), \epsilon \norm{F}_{L_2(P)}) \leq \left(\frac{A}{\epsilon}\right)^{v},
\end{equation}
where $N(T,d,\epsilon)$ denote the $\epsilon$-covering number of the metric space $(T,d)$, $F$ is the envelope function of $\mathcal{F}$ and the supremum is taken over the set of all probability measure on $\mathbb{R}^d$. The quantities $A$ and  $v$ are called the $VC$ characteristics of $\mathcal{F}$.
\end{definition}

\begin{lemma}\label{lem:2002}(\citet{gine2002rates} Corollary 2.2)
Consider $\mathcal{F}$ be a
measurable uniformly bounded VC class of functions on $\mathbb{R}^d$ whose VC characters are $A, v$, and 
\begin{equation}\label{varcond}
    \sup_{f \in \F} \text{Var}_P[f] \leq \sigma^2; \quad \sup_{f\in \F} ||f||_{\infty} \leq U,
\end{equation}
with $0 < \sigma^2 < \frac{U}{2}$, and $\sqrt{n}\sigma \geq U \sqrt{\log{(\frac{U}{\sigma})}}$.
Then there exist positive constants $C$ and $C_0$ depending only on $A$ and $v$ such that for all $\lambda \geq C_0$ and $t$ satisfying
\begin{equation*}
C_0\sqrt{n}\sigma \sqrt{\log{\frac{U}{\sigma}}} \leq t \leq \lambda \frac{n\sigma^2}{U},
\end{equation*}
we have
\begin{equation*}
    \PP{\sup_{f\in \F} |\sum_{i=1}^n f(X_i) - f(X_1)| \geq t} \leq C \exp \left\{ -\frac{\log{(1+\frac{\lambda}{4C})}}{\lambda C} \frac{t^2}{n\sigma^2}\right\},
\end{equation*}
where $X_1,\dots,X_n \stackrel{iid}{\sim} P$. 
\end{lemma}
Denote the class of functions 
\begin{equation*}
    \mathcal{F}_h := \left \{ \bm{K}_h\left( \cdot-\bm{x}\right),\ \bm{x}\in \mathbb{R}^d\right\}.
\end{equation*}
Then we can write
\begin{equation*}
   \norm{\widehat{\bm{K}}_h - \EE{\widehat{\bm{K}}_h}}_{\infty} = \sup_{\bm{x}\in \mathbb{R}^d}\left|\widehat{\bm{K}}_h(\bm{x}) - \mathbb{E}\left[\widehat{\bm{K}}_h(\bm{x})\right]\right| = \frac{1}{n}\sup_{f \in \mathcal{F}_h} \left|\sum_{i=1}^n \Big(f(\bm{X}_i) - f(\bm{X}_1)\Big)\right|,
\end{equation*}
where $\bm{X}_1, \dots, \bm{X}_n \stackrel{iid}{\sim} \bm{k}_{\alpha,\mu,r}$. 

First we examine that $\mathcal{F}_h$ is VC class for $K \in G(\beta)$. Since $K$ is compact supported and polynomial, therefore $\mathcal{F}_h$ is a VC class with $v = {d+\beta \choose d}$, and some constant $A$.

Then we examine the variance and infinity norm condition in \eqref{varcond}: note
\begin{align*}
    \sup_{f\in \F} \text{Var}_P[f] & = \sup_{\bm{x} \in \mathbb{R}^d} \text{Var}_{\bm{u}\sim P}[\bm{K}_h(\bm{u}-\bm{x})] \leq \sup_{\bm{x} \in \mathbb{R}^d} \int_{\bm{u}\in \mathbb{R}^d} \bm{K}^2_h (\bm{u}-\bm{x})\bm{k}_{\alpha,\mu,r}(\bm{u})d\bm{u} \\
    & =\sup_{\bm{x} \in \mathbb{R}^d} \frac{1}{h^{2d}}\prod_{i=1}^d \int_{\mathbb{R}} K^2 (\frac{u_i-x_i}{h})k_{\alpha,\mu,r}(u_i)d u_i \\
    &
    \stackrel{\bm{u}=\bm{\bm{x}+h\bm{v}}}{=} \sup_{\bm{x} \in \mathbb{R}^d} \frac{1}{h^{d}}\prod_{i=1}^d \int_{\mathbb{R}} K^2 (v_i)k_{\alpha,\mu,r}(x_i+hv_i)d v_i 
    \\
    & \leq  \sup_{\bm{x} \in \mathbb{R}^d} \frac{1}{h^{d}}\prod_{i=1}^d \left(||k_{\alpha,\mu,r}||_{\infty}\int_{\mathbb{R}} K^2 (v_i)d v_i\right)  \\
    &= \left(\frac{ (1-\alpha + \frac{\alpha}{r})}{h}\right)^d \left(||k||_{\infty}\int_{\mathbb{R}} K^2 (x) d x\right)^d := C_1 \sigma^2,
\end{align*}
where $C_1 = \left(||k||_{\infty}\int_{\mathbb{R}} K^2 (x) d x\right)^d$ is constant only depends on $k$ and $K$. Also note
\begin{align*}
    \sup_{f\in\F} ||f||_{\infty} = \sup_{\bm{x},\bm{u}\in\mathbb{R}^d} ||\bm{K}_h(\bm{u}-\bm{x})||_{\infty} =  ||\bm{K}_h||_{\infty} = ||K_h||_{\infty}^d = \frac{||K||_{\infty}^d}{h^d}.
\end{align*}
Let $U = 2 C_2 (1-\alpha + \frac{\alpha}{r})^d  \frac{1}{h^d}$, with $C_2 = \norm{k}_{\infty}\norm{K}_{\infty}$, then it is easy to verify that
\begin{equation*}
    \sup_{f \in \F} ||f||_{\infty} < U, \quad 0 < \sigma^2 < U/2,
\end{equation*}
since $\int K^2 \leq \norm{K}_{\infty} \int K = \norm{K}_{\infty}$, and $\frac12\norm{k}_{\infty} < 1< 1-\alpha + \frac{\alpha}{r}$.

Since both $\sigma^2$ and $U$ do not depend on $n$, therefore condition $\sqrt{n}\sigma \geq U \sqrt{\log{(\frac{U}{\sigma})}}$ is satisfied for all $n$ bigger than finite $n_0:=\frac{U^2}{\sigma^2}\log{\frac{U}{\sigma}}$. Consider $0< \epsilon < C_0\frac{\sigma^2}{U}$, $\lambda = C_0$, and $n> (C_0^2\vee 1) n_0$, we can finally apply \lemref{2002} and get
\begin{align*}
    \PP{\sup_{\bm{x}\in \mathbb{R}^d}|\widehat{\bm{K}}_h - \EE{\widehat{\bm{K}}_h}|> \epsilon } & = \PP{\sup_{f \in \F}|\sum_{i=1}^n \left( f(X_i)-f(X_1)\right)|> \epsilon n } \\
    & \leq C \exp \left\{ -\frac{C_1\log{(1+\frac{C_0}{4C})}}{C_0 C} \frac{\epsilon^2 n h^d}{ (1-\alpha + \frac{\alpha}{r})^d}\right\}.
\end{align*}
Let the right hand side equals $\delta$, in turn we have, for $\delta$ small enough (solve the upper bound on $\epsilon$ to get the lower bound on $\delta$),
\begin{align*}
    \PP{\sup_{\bm{x}\in \mathbb{R}^d}|\widehat{\bm{K}}_h - \EE{\widehat{\bm{K}}_h}|> \sqrt{\frac{C_3 \log{(C/\delta)} (1-\alpha + \frac{\alpha}{r})^d}{n h^d}} } \leq \delta,
\end{align*}
where $C_3:= \sqrt{\frac{C_0 C}{C_1 \log{(1+\frac{C_0}{4C})}}}$.

\item[\textbf{Step 2.}]
For the second term in \eqref{infdecomp},
first we prove that if $k \in H(\beta,L)$, then $k_{\alpha,\mu,r} \in H(\beta, (1-\alpha + \frac{\alpha}{r^{\beta + 1}}) L )$. Note that for this argument, we are only considering the one-dimensional case, therefore
\begin{equation}\label{hbeta}
 k \in H(\beta, L) \Longleftrightarrow \sup_{x}\left|\frac{d^{\beta} k(x)}{d x^{\beta}} \right| \leq L. 
\end{equation}
Using the chain rule, we have
\begin{equation*}
  \frac{d^{\beta} k_{\alpha,\mu,r}(x)}{d x^{\beta}} = (1-\alpha)\frac{d^{\beta}  k(x)}{d x^{\beta}} + \frac{\alpha}{r}\frac{d^{\beta} k\left(\frac{x-\mu}{r}\right)}{d x^{\beta}}  = (1-\alpha)\frac{d^{\beta}  k(u)}{d u^{\beta}}\mid_{u=x} + \frac{\alpha}{r^{1+\beta}}\frac{d^{\beta} k\left(u\right)}{d u^{\beta}}\mid_{u = \frac{x-\mu}{r}}.
\end{equation*}
Therefore using \eqref{hbeta}, we have
\begin{equation*}
\sup_{x}\left|\frac{d^{\beta} k_{\alpha,\mu,r}(x)}{d x^{\beta}} \right| \leq \Big((1-\alpha) +  \frac{\alpha}{r^{1+\beta}}\Big) L,
\end{equation*}
that is $k_{\alpha, \mu, r} \in H\left(\beta, (1-\alpha + \frac{\alpha}{r^{\beta+1}}) L\right)$.

 Then we have
\begin{align*}
  \norm{\EE{\widehat{\bm{K}}_h} - \bm{k}_{\alpha,\mu,r}}_{\infty} & = \sup_{\bm{x}} |\int  \bm{K}_h\left(\norm{\bm{u}-\bm{x}}\right)\bm{k}_{\alpha,\mu,r}(\bm{u}) d \bm{u}  - \bm{k}_{\alpha,\mu,r}(\bm{x})|\\
  & = \sup_{\bm{x}} \prod_{i=1}^{d} \int   K_h\left(\norm{u_i-x_i}\right)\left(k_{\alpha,\mu,r}(u_i)  - k_{\alpha,\mu,r}(x_i) \right) d u_i \\
  & = \sup_{\bm{x}} \left| \prod_{i=1}^{d} \int K\left(|v_i|\right)\Big(k_{\alpha,\mu,r}(x_i + h v_i) - k_{\alpha,\mu,r}(x_i) \Big)  \right|\\
  &
  \leq \sup_{\bm{x}} \prod_{i=1}^{d} \Big\{\left| \int K\left(|v_i|\right)\Big(k_{\alpha,\mu,r}(x_i + h v_i) - k^{x_i,\beta}_{\alpha,\mu,r}(x_i+ h v_i) \Big) \right| \\&\quad \quad \quad \quad \quad \quad \quad + \left| \int K\left(|v_i|\right)\Big(k^{x_i,\beta}_{\alpha,\mu,r}(x_i + h v_i) - k_{\alpha,\mu,r}(x_i) \Big) \right|\Big\}\\
  & \stackrel{(i)}{=} \sup_{\bm{x}} \prod_{i=1}^{d} \left| \int K\left(|v_i|\right)\Big(k_{\alpha,\mu,r}(x_i + h v_i) - k^{x_i,\beta}_{\alpha,\mu,r}(x_i+ h v_i) \Big) \right|
  \\
  & \stackrel{(ii)}{\leq}  \sup_{\bm{x}} \prod_{i=1}^{d} \left| \int K\left(|v_i|\right)\Big((1-\alpha + \frac{\alpha}{ r^\beta})L h^\beta |v_i|^\beta \Big) \right|\\
   & = \left((1-\alpha + \frac{\alpha}{ r^{\beta+1}})L h^\beta \left| \int K\left(|v|\right) |v|^\beta\right|\right)^d : = (1-\alpha + \frac{\alpha}{ r^{\beta+1}})^d h^{d\beta} C_4,
\end{align*}
where $\cdot^{x,\beta}$ is the taylor expansion of $\cdot$ at $x$ to order $\beta-1$, and $C_4 := L^d\left| \int K\left(|v|\right) |v|^\beta\right|^d$.  Specifically, (i) is true since $k_{\alpha,\mu,r} \in H(\beta, (1-\alpha + \frac{\alpha}{r^{\beta}})L)$, and therefore $\Big(k^{x_i,\beta}_{\alpha,\mu,r}(x_i + h v_i) - k_{\alpha,\mu,r}(x_i) \Big)$ is a polynomial of degree $\beta-1$, then use the fact that $K \in G(\beta)$, we have the second term is zero; and (ii) is true from the fact that $k_{\alpha,\mu,r} \in H(\beta, (1-\alpha + \frac{\alpha}{r^{\beta+1}})L)$. \\
\end{itemize}

\noindent
Combining the above analysis, we have
\begin{align*}
    \PP{\sup_{\bm{x}\in \mathbb{R}^d}|\widehat{\bm{K}}_h - \bm{k}_{\alpha, \mu, r}|> \sqrt{\frac{C_3 \log{(1/\delta)} (1-\alpha + \frac{\alpha}{r})^d}{n h^d}} + C_4 \left(1-\alpha + \frac{\alpha}{r^{\beta+1}}\right)^d h^{d\beta} } \leq \delta,
\end{align*}
where $C_3, C_4$ are constants that do not depend on $h,\alpha,\mu,r$, but depend on $k, K, d, n, \beta, L$.
\end{proof}

\section{Robustness on the empirical level}\label{app:emprob}
\begin{definition}\label{def:empmodel}
(Empirical contamination model) Given $n$ bivariate samples $(x_1,y_1),\dots,(x_n,y_n)$, we consider the corresponding contaminated samples $\{(x_i', y_i')\}_{i=1}^n$ that satisfying
\begin{equation*}
    (x_i', y_i') = (x_i, y_i)\ \text{for} \ 1 \leq i \leq d_n;\quad \quad (x_i', y_i') = (x', y')\ \text{for}\ d_n + 1 \leq i \leq n,
\end{equation*}
where $1\leq d_n \ll n$ is the number of outliers.  
\end{definition}

Denote the empirical $\text{aLDG}_t$ under the contamination model \defref{empmodel} as $\widehat{\text{aLDG}}_t'$. We consider characterizing the following modified influence function (defined to adapt empirical setting)
\begin{equation*}
    \text{MIF}((x',y'), \widehat{\text{aLDG}}, F_n):= |\widehat{\text{aLDG}}_t' - \widehat{\text{aLDG}}_t|.
\end{equation*}

In \thmref{aLDGrobemp} we give an upper bound on MIF, which depends on the number of outliers $d_n$ and sample size $n$. 

% Specifically, 
% we need the contamination mass to decay with the sample size, at rate no faster than the estimation error rate $O(\eta_n \vee \frac{\log_n}{n})$, where the term $\eta_n$ arise from the density estimation error, and the term $\frac{\log{n}}{n}$ arose from the probability estimation error.
 
%On the other hand, in \thmref{aLDGrobemp2} we show that, if the number of outliers $d_n\equiv d$ for all $n$, then the corresponding modified influence function is unbounded.

\begin{theorem}\label{thm:aLDGrobemp}
Consider the contamination model in \defref{empmodel} with $d_n$ outliers, and the empirical $\widehat{\text{aLDG}}_t$ in \eqref{eq:aLDGemp} using boxcar kernel density estimator \eqref{densest} with bandwidth $h_n$. Assume the point mass $(x', y')$ is far away from all the $n$ uncontaminated samples:
\begin{equation*}
    (x', y'): \quad \min_{j\in [n]}|x_j-x'| > h_n,\ \min_{j\in [n]}|y_j-y'| > h_n.
\end{equation*}
Under the same conditions on the true data distribution as in \thmref{aLDGconsist}, 
then with high probability, we have
\begin{equation*}
    \text{MIF}((x',y'), \widehat{\text{aLDG}}, F_n) := \Big|\widehat{\text{aLDG}}'_t - \widehat{\text{aLDG}}_t\Big| < 2\epsilon_n + \eta_{n-d_n}+ 2\sqrt{\epsilon_n+\eta_{n-d_n}} \sqrt{\frac{\log{n}}{n}},
\end{equation*}
where $\epsilon_n:=\frac{d_n}{n}$ is the contamination mass, and $F_n$ denote the empirical distribution of the uncontaminated data.
\end{theorem}

\begin{proof}
Given $n$ bivariate samples $(x_1,y_1),\dots,(x_n,y_n)$, denote
\begin{align*}
    T(x,y)&:= \frac{f_{X,Y}(x, y) - f_{X}(x) f_{Y}(y) }{\sqrt{ f_{X}(x)f_{Y}(y)}},\quad T_i := T(x_i,y_i); \\
    \widehat{T}(x,y)&:= \frac{\widehat{f}_{X,Y}(x, y) - \widehat{f}_{X}(x) \widehat{f}_{Y}(y) }{\sqrt{ \widehat{f}_{X}(x)\widehat{f}_{Y}(y)}},\quad \widehat{T}_i := \widehat{T}(x_i,y_i)
 \end{align*}
where $\widehat{f}_{X,Y}, \widehat{f}_{X}, \widehat{f}_{Y}$ are some density estimator for $f_{X,Y}, f_X, f_Y$. Then the empirical aLDG can be written as
\begin{align*}
    \widehat{\text{aLDG}}_t & =  \frac{1}{n}\sum_{i=1}^n\ones\left\{ \widehat{T}_i \geq t \right\}, 
\end{align*}

Denote the density estimator under the contaminated model as $\widehat{f}_{X}', \widehat{f}_{Y}', \widehat{f}_{XY}'$, and the corresponding statistics as $\widehat{T}_i'$, and $\widehat{\text{aLDG}}_{t}'$. First we have 
\begin{align*}
    & \widehat{f}_{X}'(\cdot) = \frac{1}{n} \sum_{j=d_n+1}^n K_{h_n}(\cdot, x_j)  +  \epsilon_n K_{h_n}(\cdot,x'),\quad \widehat{f}_{Y}'(\cdot) = \frac{1}{n} \sum_{j=d_n+1}^n K_{h_n}(\cdot, Y_j)+  \epsilon_n K_{h_n}(\cdot,y'),\nonumber\\
    & \widehat{f}_{XY}'(\cdot,\cdot) = \frac{1}{n} \sum_{j=d_n+1}^n K_{h_n}(\cdot, x_j)K_{h_n}(\cdot, y_j) + \epsilon_n K_{h_n}(\cdot,x')K_{h_n}(\cdot,y').
\end{align*}
And consequently, 
%for $1 \leq i \leq d_n$, 
% \begin{align}\label{fhatcon}
%     & \widehat{f}_{X}'(x_i') = 
%       \epsilon_n / h_n  , \quad \widehat{f}_{Y}'(y_i) = \epsilon_n / h_n , \quad \widehat{f}_{XY}'(x_i,y_i) =  \epsilon_n / h_n^2;
% \end{align}
% while 
for $d_n+1 \leq i \leq n$,
\begin{align*}
    & \widehat{f}_{X}'(x_i') = 
      \widehat{f}_{X}(x_i) - \epsilon_n\frac{1}{d_n} \sum_{j=1}^{d_n} K_h(x_i,x_j),\quad \widehat{f}_{Y}'(y_i) = \widehat{f}_{Y}(y_i) - \epsilon_n\frac{1}{d_n} \sum_{j=1}^{d_n} K_h(y_i,y_j) \\
    & \widehat{f}_{XY}'(x_i,y_i) = \widehat{f}_{X,Y}(x_i, y_i) - \epsilon_n\frac{1}{d_n} \sum_{j=1}^{d_n}K_h(x_i,x_j) K_h(y_i,y_j).
\end{align*}
We assume that the true marginal densities $f_X$ and $f_Y$ are bounded by some constant $c_{\max}$ and the corresponding density estimation error is uniformly bounded by  $\eta_n$ with high probability. Denote
\begin{equation*}
    \widehat{c}_{\max} := \max\Big\{\sup_x\frac{1}{d_n} \sum_{j=1}^{d_n} K_h(x,x_j),\quad \sup_y \frac{1}{d_n} \sum_{j=1}^{d_n} K_h(y,y_j),\quad \sup_{x,y}\frac{1}{d_n} \sum_{j=1}^{d_n}K_h(x,x_j) K_h(y,y_j)\Big\},
\end{equation*}
then we have 
\begin{equation*}
    c_{\max}- \eta_{d_n} \leq \widehat{c}_{\max}  \leq c_{\max}+ \eta_{d_n},
\end{equation*}
with high probability. Consequently we have
\begin{equation*}
\max_{d_n+1 \leq i\leq n} |\widehat{T}_i' - \widehat{T}_i| \leq \epsilon_n \widehat{c}_{\max} \leq \epsilon_n (c_{\max} + \eta_{d_n})
\end{equation*}
with high probability.

Therefore, for all $i$, with high probability, we can conclude 
\begin{align*}
     & \widehat{T}_i \geq t+\epsilon_n (c_{\max} + \eta_{d_n})\quad \text{or}\quad \widehat{T}_i < t - \epsilon_n (c_{\max} + \eta_{d_n}) \\  & \quad \quad \quad \Longrightarrow  \ones\{\widehat{T}_i > t\} = \ones\{\widehat{T}_i' > t\}.
\end{align*}
This implies, with high probability,
\begin{align*}
   |\widehat{\text{aLDG}}_{t}' - \widehat{\text{aLDG}}_t| \ & \leq \epsilon_n + 
    \frac{1}{n}\sum_{i=d_n+1}^n \ones\{ t  - \epsilon_n (c_{\max}+ \eta_{d_n}) < \widehat{T}_i \leq  t + \epsilon_n (c_{\max}+ \eta_{d_n}) \}\\
    & = \epsilon_n + (1-\epsilon_n)\Big(\widehat{P}_{n-d_n}(\widehat{S}_{t-\epsilon_n (c_{\max}+ \eta_{d_n})}) - \widehat{P}_{n-d_n}(\widehat{S}_{t+\epsilon_n (c_{\max}+ \eta_{d_n})})\Big)\\
    & \leq \epsilon_n + (1-\epsilon_n)\Big(\widehat{P}_{n-d_n}(S_{t-\epsilon_n (c_{\max}+ \eta_{d_n})-c\eta_{n-d_n}})  - \widehat{P}_{n-d_n}(S_{t + \epsilon_n (c_{\max}+ \eta_{d_n})+ c\eta_{n-d_n}})\Big) \\
    & \leq \epsilon_n + (1-\epsilon_n) \Big( P(D_t) + |\widehat{P}_{n-d_n}(D_t)-P(D_t)|\Big),
\end{align*} 
where 
\begin{align*}
    & \widehat{S}_t:=\{(x,y): \widehat{T}> t\},\quad S_t:=\{(x,y): T> t\}, \\
    & D_t:= S_{t-\epsilon_n (c_{\max}+ \eta_n)-c\eta_{n-d_n}} \setminus S_{t+\epsilon_n (c_{\max}+ \eta_n)+c\eta_{n-d_n}}.
\end{align*}
\noindent
Since we assume that $\text{aLDG}_t$ is L-Lipschitz smooth around $t$, therefore
\begin{align*}
    P(D_t) \leq 2L(\epsilon_n (c_{\max}+ \eta_{d_n})+c\eta_{n-d_n})  \asymp O(\epsilon_n + \eta_{n-d_n})\to 0.
\end{align*}
Then using the Bernstein inequality for Bernoulli variable with mean $P(D_t) \ll 1$, with high probability we have
\begin{align*}
|\widehat{P}_{n-d_n}(D_t)-P(D_t)| & \leq \sqrt{ \frac{P(D_t)\log{(n-d_n)}}{n-d_n}}\\
&\stackrel{<}{\sim} \sqrt{ \frac{(\epsilon_n + \eta_{n-d_n}) \log{n}}{n-d_n}} = \sqrt{\frac{\epsilon_n + \eta_{n-d_n}}{1-\epsilon_n}} \sqrt{\frac{\log{n}}{n}}.
\end{align*}
Combine the above results, with high probability we have,
\begin{align*}
   |\widehat{\text{aLDG}}'_t - \widehat{\text{aLDG}}_t| & \stackrel{<}{\sim} \epsilon_n + (1-\epsilon_n)\Big(\epsilon_n + \eta_{n-d_n}+ \sqrt{\frac{\epsilon_n+\eta_{n-d_n}}{1-\epsilon_n}} \sqrt{\frac{\log{n}}{ n}} \Big) \\
   & < 2\epsilon_n + \eta_{n-d_n}+ 2\sqrt{\epsilon_n+\eta_{n-d_n}} \sqrt{\frac{\log{n}}{ n}}.
\end{align*}
 
\noindent
Finally, we can conclude, if the contamination mass $\epsilon_n \to 0$ as $n\to\infty$, and  satisfy $\epsilon_n = O(\eta_n \vee \frac{\log{n}}{n})$,  then with high probability, we have
\begin{equation*}
    \text{MIF}((x',y'), \widehat{\text{aLDG}}, F_n) < 2\epsilon_n + \eta_{n-d_n}+ 2\sqrt{\epsilon_n+\eta_{n-d_n}} \sqrt{\frac{\log{n}}{n}} \ll 1,
\end{equation*}
which goes to zero as $n$ goes to infinity.
\end{proof}

\section{Discussion on thresholding methods}\label{app:chooset}
Another intuitive way we found for selecting $t$ is based on the 
%using an observed fact about 
curve of $aLDG_t$ versus $t$. This function
%Specifically, from extensive examples of $\text{aLDG}_t$ versus $t$ curves, we observe that all these curves appear a similar trend: that is it 
tends to decrease rapidly near zero and then reaches an inflection point, after which it declines very slowly
%while suddenly very slowly at some point away from zero; in other words, it has an inflection point around which increment of $t$ suddenly unable to reduce further $\text{aLDG}_t$ much. 
(e.g., \figref{chooset}). %for explicit evidence (one curve for one random shuffle). 
We propose selecting the threshold $t$ to be 
the inflection point $t^\star$. Since the increment of $t$ around $t^\star$ is suddenly unable to reduce further $\text{aLDG}_t$ much, therefore, we expect this choice to strike a balance between robustness and sensitivity. To stabilize the estimation of such inflection point, we use the median of estimated inflection point from $\max\{\lfloor 1000/n \rfloor, 5\}$ different random shuffles as the final estimation.
%visually this makes sure that $aLDG$ is small enough under independence (also robust enough) and is still big enough under dependence (i.e., sensitive enough).  
We call this $t$ selection method the \emph{inflection point} method. 

In \figref{chooset} and \figref{aldgt}, we compare the above three proposed methods of selecting $t$. We use 18 different bivariate distributions to make the comparison (see \figref{data} for the explicit display of each distribution). We believe this series of distributions are representative enough as it covers cases from linear to nonlinear, monotone to nonmonotone, and also probabilistic mixtures. We find that the \emph{asymptotic norm} method is often too conservative given the small sample size. In contrast, the \emph{uniform error} and \emph{inflection point} method are often similar to each other. On the other hand, \figref{aldgt} shows that \emph{uniform error} method gives more stable value than the \emph{inflection point} method, while \emph{asymptotic norm} is the most stabilised among the three. Therefore in practice, we recommend people use \emph{uniform error} over \emph{asymptotic norm} when the sample size is not too big (e.g., no bigger than 200); while using \emph{asymptotic norm} when the sample size is big enough (e.g., bigger than 200) and the computation budget is limited. 

%way of selecting the cutoff $t$ as the uniform estimation error of $T$ 

\begin{figure}[H]
    \centering
    \includegraphics[width=1.1\linewidth]{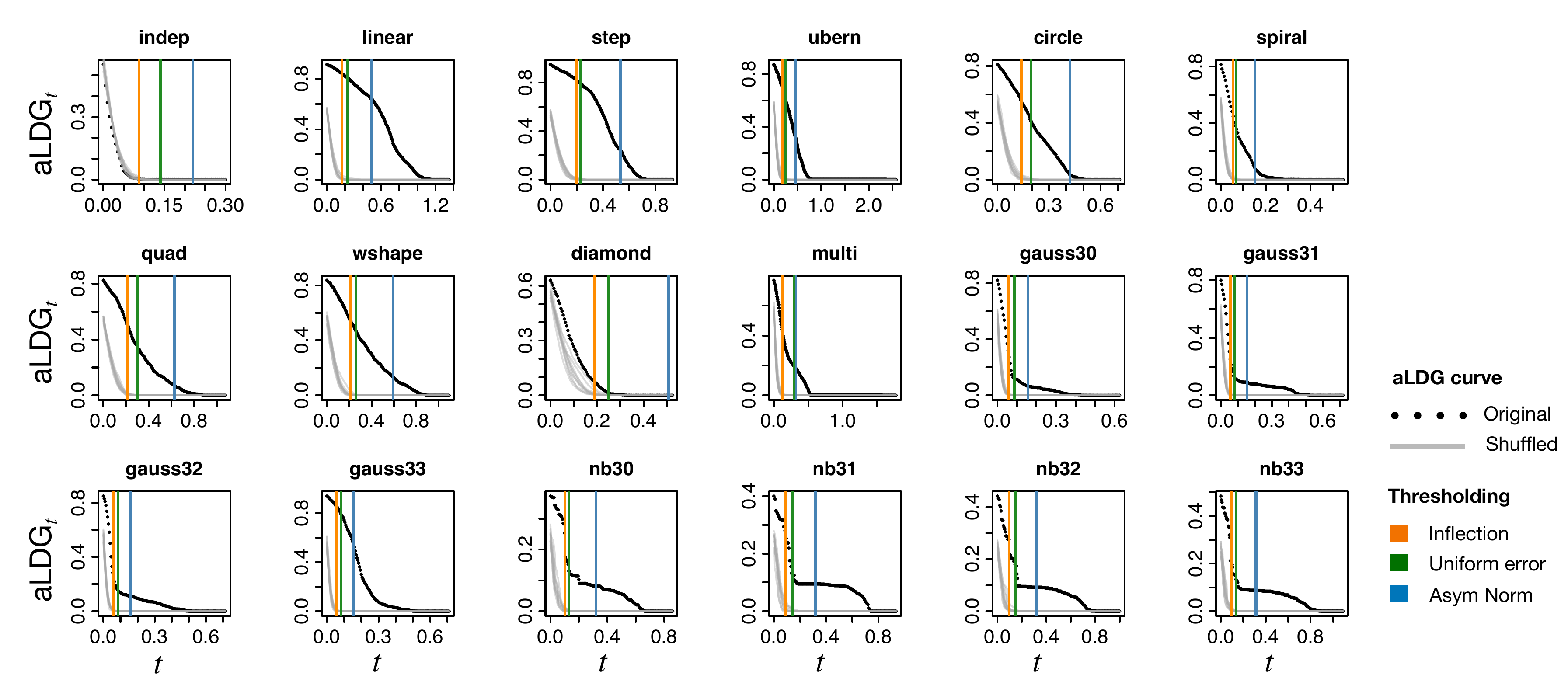}
    \caption{The curves of $\text{aLDG}_t$ versus $t$ estimated by 1000 samples. Each plot represents different bivariate distribution annotated by the subtitle (see \figref{data} for explicit display of each distribution). In each plot, the black dot curve represents the $\text{aLDG}_t$ estimated using original data samples, and the gray dot curves represent the $\text{aLDG}_t$ estimated using shuffled data samples (one curve each random shuffle, 20 curves in total); The vertical lines represent different choices of the thresholding: the orange one represents the \emph{inflect point} method; the green one represents the \emph{uniform error} method; and the blue one represent the \emph{asymptotic norm} one. }
    \label{fig:chooset}
\end{figure}

\begin{figure}[H]
    \centering
    \includegraphics[width=1.1\linewidth]{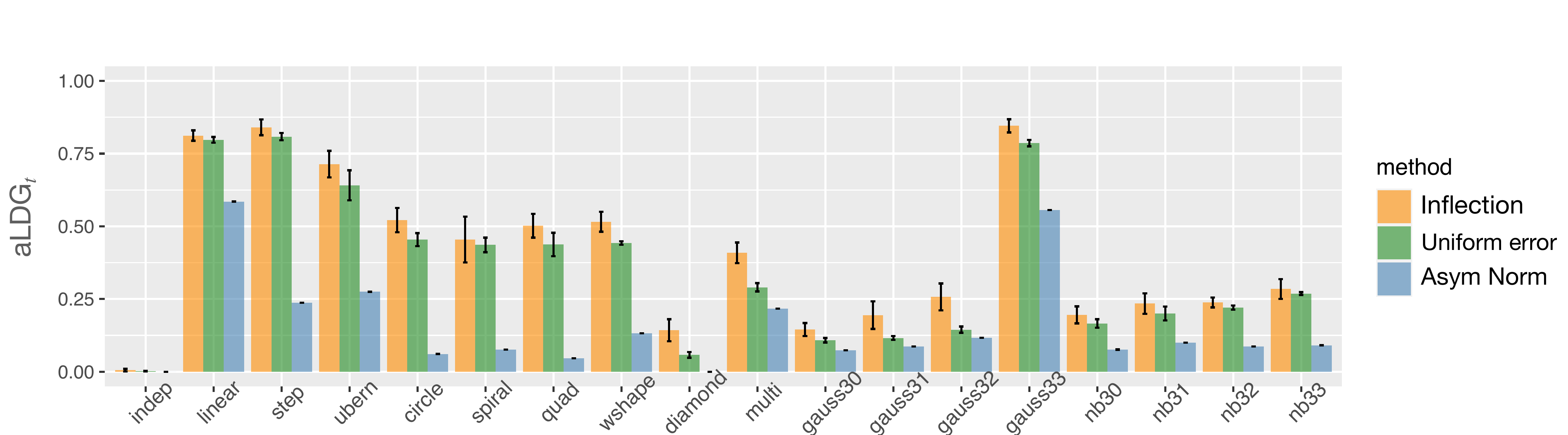}
    \caption{The value of $\text{aLDG}_t$ estimated by 1000 samples using different method of choosing $t$. The x-axis represents different bivariate distribution (see \figref{data} for explicit display of each distribution). For each distribution, we show the mean value of $\text{aLDG}_t$ over 20 trials with error bar, where different thresholding method is annotated by different color.}
    \label{fig:aldgt}
\end{figure}
\section{Detailed example for merits of thresholding}\label{app:thred}
Consider the following product kernel density mixture:
\begin{align*}
    & f_{X}(x) = \alpha k_{0,r}(x) + (1-\alpha)k_{0,1}(x), \quad  f_{Y}(y) = k_{0,r}(y) + (1-\alpha)k_{0,1}(y), \\
    & f_{XY}(x,y) = \alpha k_{0,r}(x)k_{0,r}(y) + (1-\alpha)k_{0,1}k_{0,1}, 
\end{align*}
where $\alpha \in (0,1)$, $0<r\leq 1$ and  $k_{\mu,r}(\cdot):=\frac{1}{r}k(\frac{\cdot-\mu}{r})$, with $k$ as the density of a one dimensional uniform distribution supported on $[-1,1]$. 

With $\alpha/r \to \infty$, $\alpha \to 0$ and $r \to 0$, we have
\begin{equation*}
   \EEst{\frac{f_{XY}(X,Y)-f_{X}(X)f_{Y}(Y)}{\sqrt{f_{X}(X)f_{Y}(Y)}}}{|X|<r\ \&\ |Y|<r} \approx \frac{\alpha(1-\alpha)/r^2}{\alpha/r} = \frac{1-\alpha}{r}
\end{equation*}
and
\begin{equation*}
    \EEst{\frac{f_{XY}(X,Y)-f_{X}(X)f_{Y}(Y)}{\sqrt{f_{X}(X)f_{Y}(Y)}}}{|X|>r \text{ or } |Y|>r} \approx -\frac{\alpha(1-\alpha)/r}{\alpha/r} = \alpha-1,
\end{equation*}
\begin{equation*}
    \EEst{\frac{f_{XY}(X,Y)-f_{X}(X)f_{Y}(Y)}{\sqrt{f_{X}(X)f_{Y}(Y)}}}{|X|>r\ \&\ |Y|>r} \approx -\frac{(1-\alpha)\alpha}{1-\alpha} = -\alpha,
\end{equation*}
therefore using the law of total expectation, we finally have 
\begin{equation}\label{nothred}
   \EE{\frac{f_{XY}(X,Y)-f_{X}(X)f_{Y}(Y)}{\sqrt{f_{X}(X)f_{Y}(Y)}}} \approx p_1 \frac{1-\alpha}{r} +p_2 (\alpha-1) + p_3\alpha, 
\end{equation}
where
\begin{align*}
   & p_1:=\PP{|X|\leq r\ \&\  |Y|\leq r}=\alpha+(1-\alpha)r^2, \\ &
    p_2:=\PP{(|X|>r\ \&\ |Y|\leq r)\ or\ (|X|\leq r\ \&\ |Y| > r)}= (1-\alpha)(2r-2r^2),\\ &
    p_3:= \PP{|X|>r \ \&\  |Y|>r} = (1-\alpha)(1-2r+r^2).
\end{align*}
Simplifying \eqref{nothred} we have,
\begin{equation*}
   \EE{\frac{f_{XY}(X,Y)-f_{X}(X)f_{Y}(Y)}{\sqrt{f_{X}(X)f_{Y}(Y)}}} \approx \frac{\alpha}{r} \to \infty.
\end{equation*}

\section{Additional plots}
\begin{figure}[H]
    \centering
    \includegraphics[width=\linewidth]{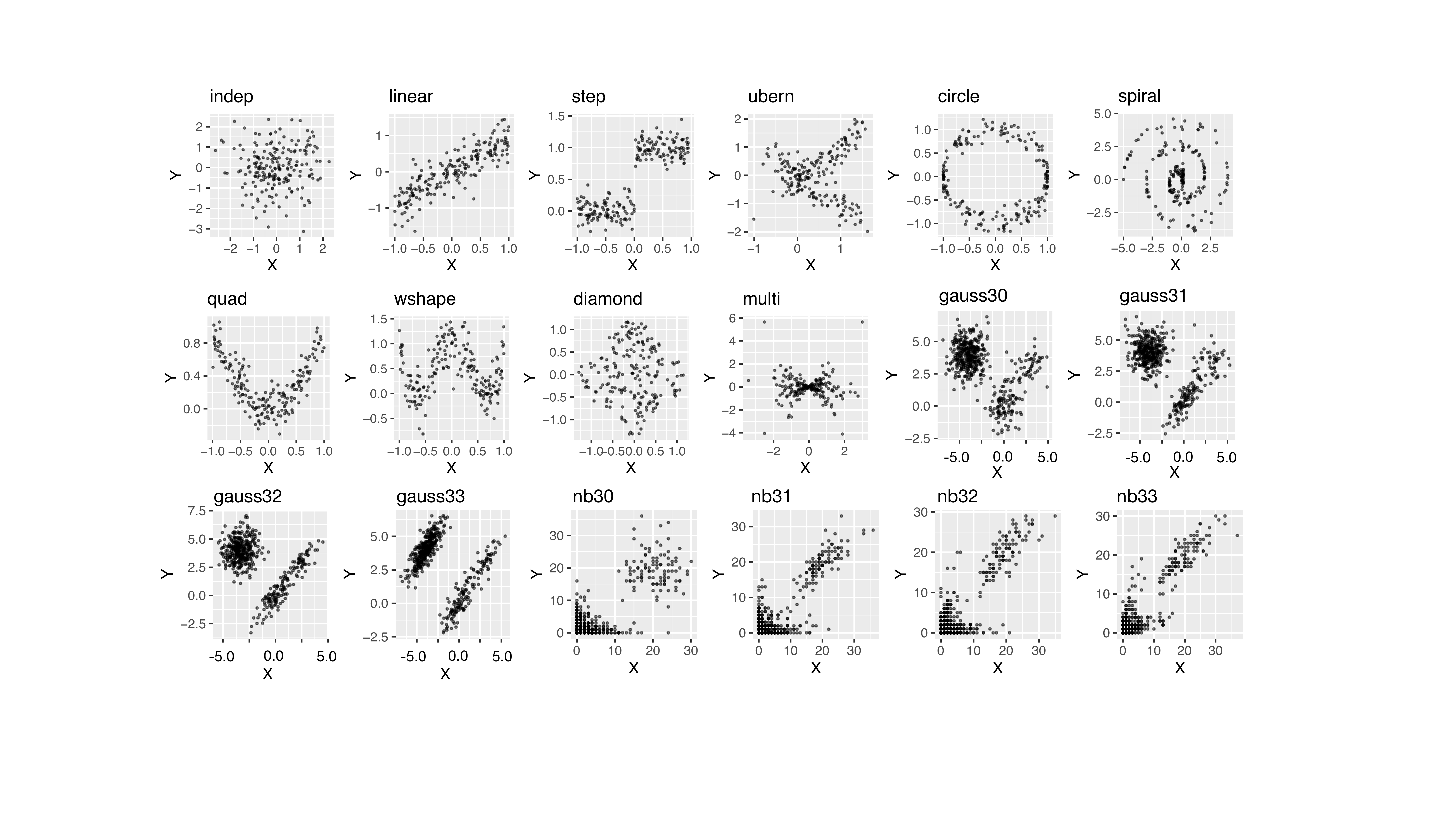}
    \caption{A summary of all the synthetic bivariate data distribution we considered in this paper. For each data distribution we plot the corresponding scatter plot using 1000 samples. We believe this series of distributions are representative enough as it covers cases from linear to nonlinear, monotone to nonmonotone, and also probabilistic mixture.}
    \label{fig:data}
\end{figure}

\begin{figure}[H]
    \centering
    \includegraphics[width=\linewidth]{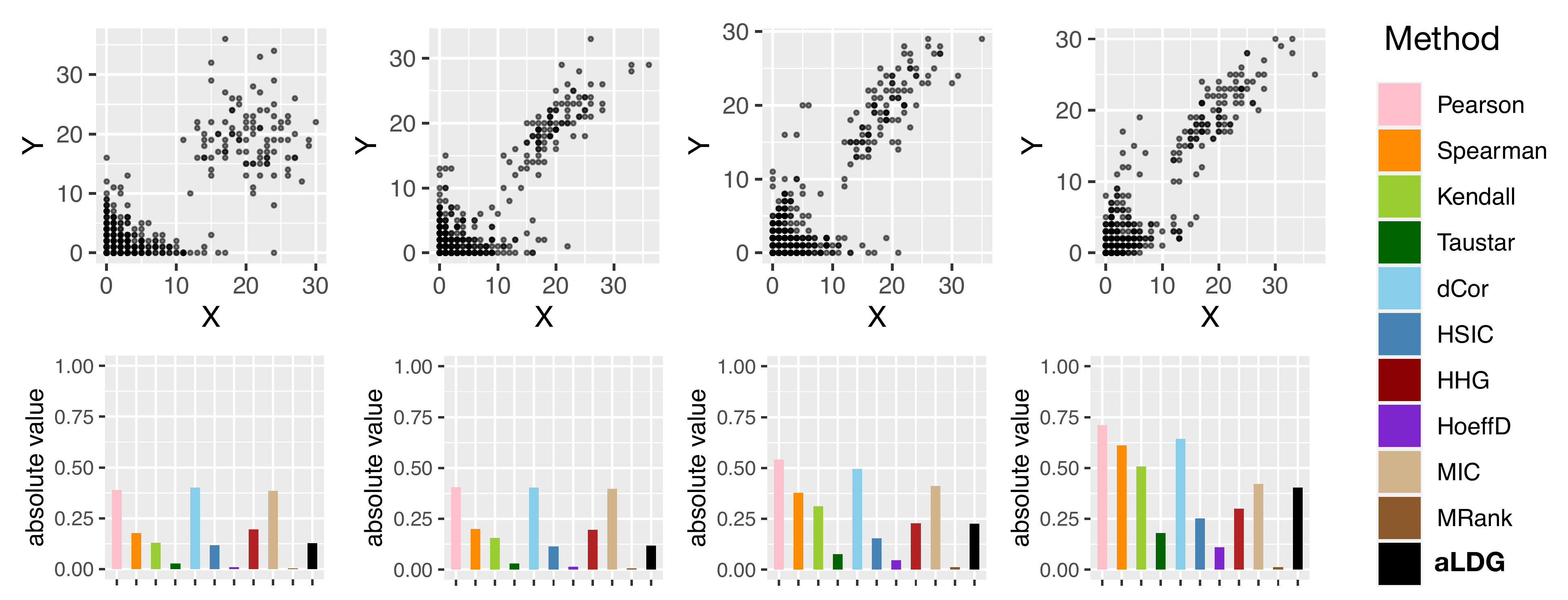}
    \caption{Empirical aLDG value for Negative Binomial mixture. The upper row shows the scatter plot, while the lower row shows the corresponding dependence level given by different measures. The data are generated as a three-component Negative Binomial mixture. From left to right there are 0,1,2,3 out of 3 components has correlation 0.8, while the rest has correlation 0, i.e. the dependence level increases from left to right.}
    \label{fig:nbmix}
\end{figure}

\begin{figure}[H]
    \centering
    \includegraphics[width=\linewidth]{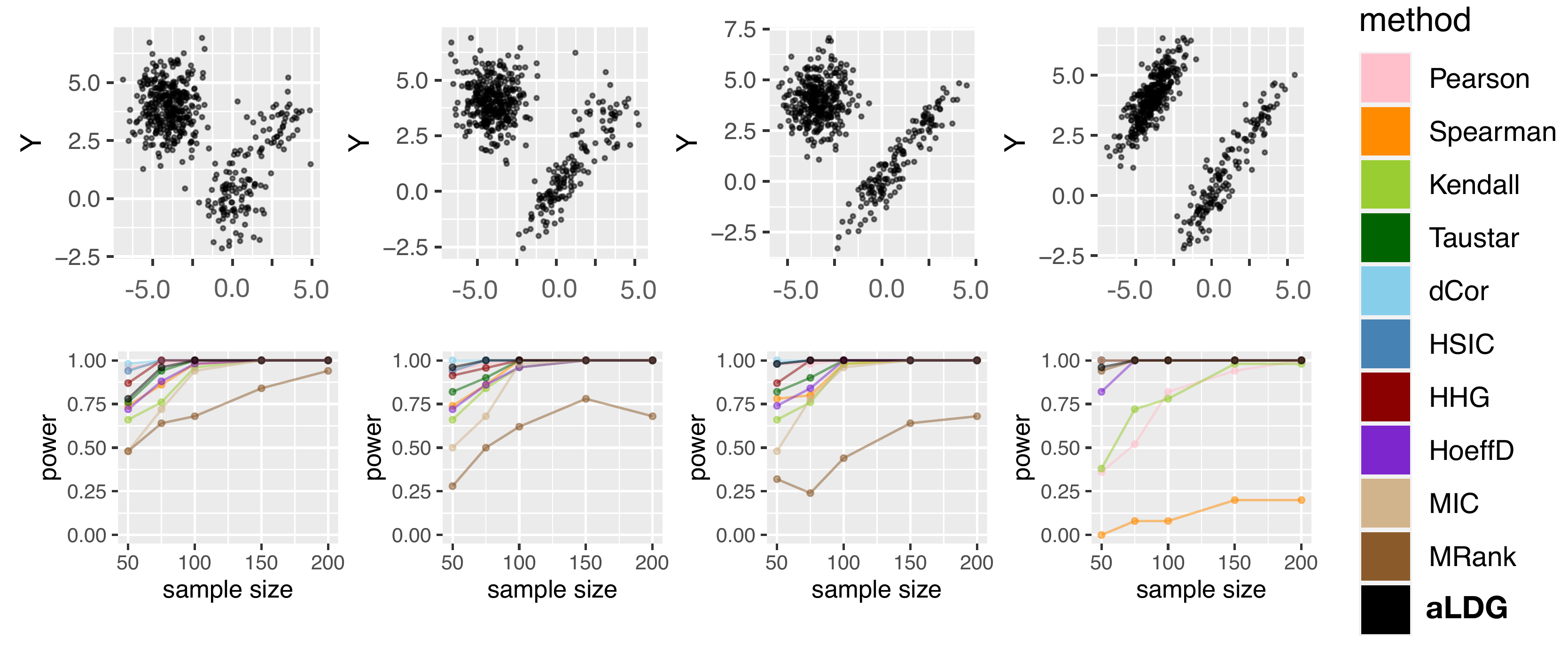}
    \caption{The empirical power of permutation test at level 0.05, based on different dependency measures under different Gaussian mixture distributions and sample sizes. The power is estimated using 50 independent trials. The data are generated as a three-component Gaussian mixture. From left to right the overall dependence level increases: specifically, 0,1,2 and 3 of the 3 components have correlation of 0.8, while the remaining components have no correlation. }
    \label{fig:gaussmixpower}
\end{figure}

\begin{figure}[H]
    \centering
    \includegraphics[width=\linewidth]{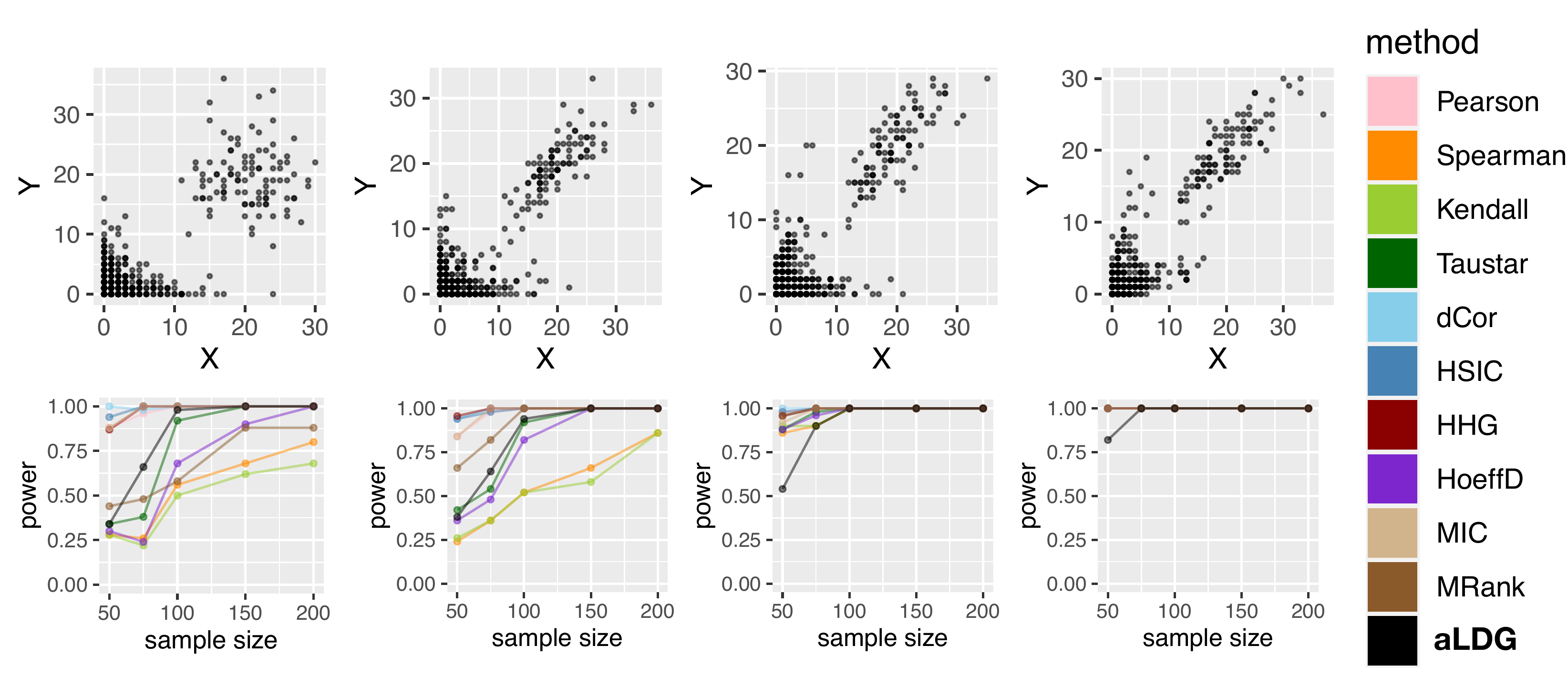}
    \caption{The empirical power of permutation test at level 0.05, based on different dependency measures under different negative binomial mixture distributions and sample sizes. The power is estimated using 50 independent trials. The data are generated as three-component Negative Binomial mixture. From left to right the overall dependence level increases: specifically, 0,1,2 and 3 of the 3 components have correlation of 0.8, while the remaining components have no correlation.}
    \label{fig:nbmixpower}
\end{figure}

\begin{figure}[H]
\centering
    \includegraphics[width=0.6\linewidth]{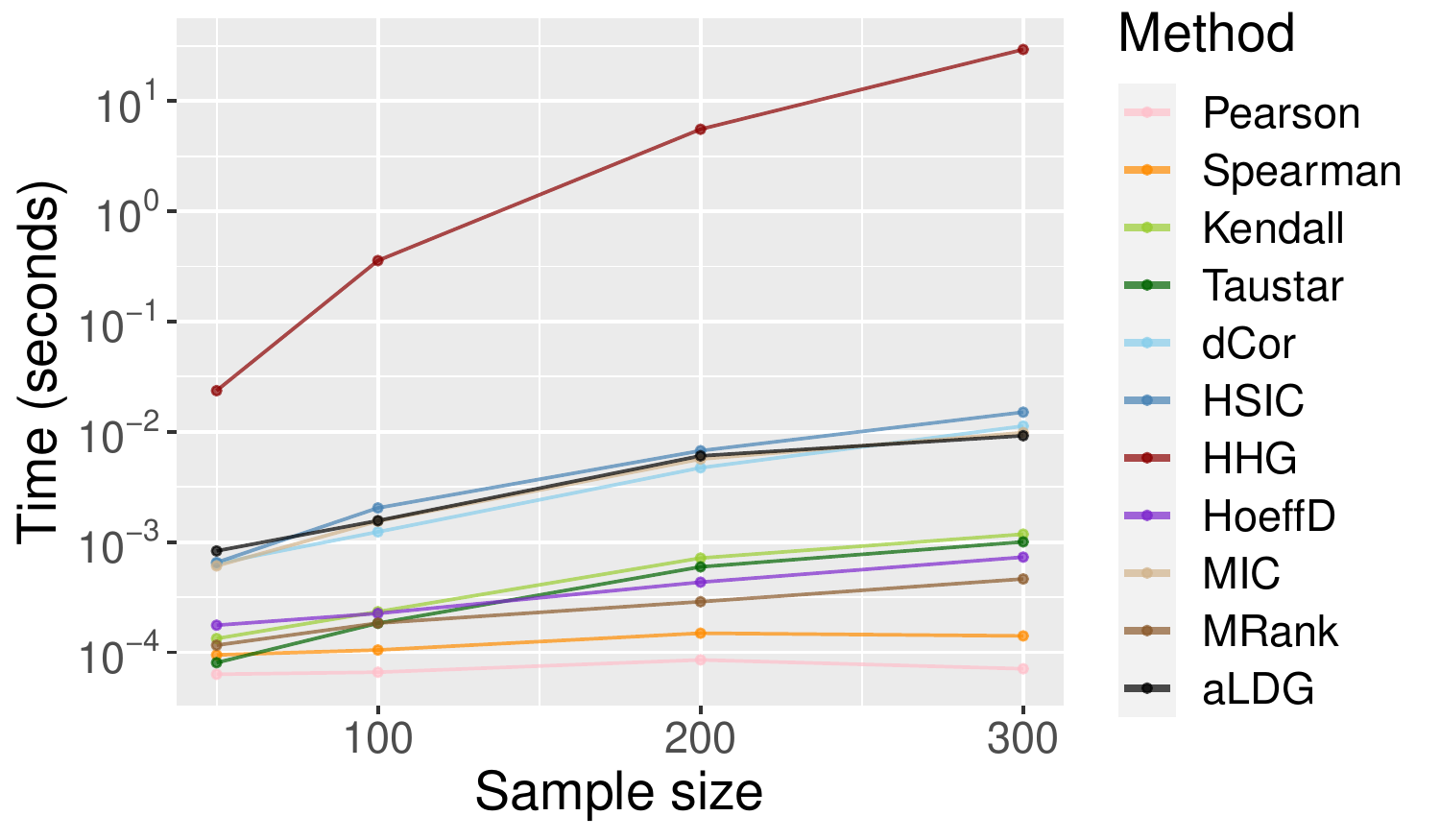}
    \caption{Computation time ($\log_{10}$ scaled) versus sample size for different methods, averaged over 10 independent trials. We can see that HHG is much slower than the others as sample size grows, while aLDG is roughly as fast as dCor, HSIC, MIC.}
    \label{fig:time}
\end{figure}

\end{appendices}

\end{document}